\documentclass{article}

\usepackage[margin=0.7in]{geometry}
\usepackage[parfill]{parskip}
\usepackage[utf8]{inputenc}
\usepackage{physics}
\usepackage{bbm} 
\usepackage{color}
\usepackage{float}
\usepackage{hyperref} 
\usepackage{cleveref} 
\usepackage{scrextend}
\usepackage{xargs}
\usepackage[dvipsnames]{xcolor}

\usepackage{graphics}
\usepackage{graphicx}
\usepackage{pstricks}
\usepackage{pst-node}
\usepackage{pst-tree}
\usepackage{tikz} 
\usepackage{enumitem}
\setlistdepth{9}
\newlist{myEnumerate}{enumerate}{9}
\setlist[myEnumerate,1,2,3,4,5,6,7,8,9]{label*=\arabic*.} 
\renewlist{itemize}{itemize}{9}
\setlist[itemize]{label=\textbullet}
\setlist[itemize,2]{label=--}
\setlist[itemize,3]{label=*}
\setlist[itemize,5]{label=--}
\setlist[itemize,6]{label=*}
\setlist[itemize,8]{label=--}

\usepackage{titlesec}

\titleformat{\paragraph}
{\normalfont\normalsize\bfseries}{\theparagraph}{1em}{}
\titlespacing*{\paragraph}
{0pt}{3.25ex plus 1ex minus .2ex}{1.5ex plus .2ex}

\crefname{paragraph}{section}{sections}

\usepackage{amsmath,amssymb,amsfonts,amsthm}
\usepackage{blkarray} 

\makeatletter
\DeclareRobustCommand{\bigDelta}{\mathop{\vphantom{\sum}\mathpalette\bigDelta@\relax}\slimits@}
\newcommand{\bigDelta@}[2]{\vcenter{\sbox\z@{$#1\sum$}\hbox{\resizebox{.9\dimexpr\ht\z@+\dp\z@}{!}{$\m@th\Delta$}}}}
\makeatother

\newcommand{\VM}{\ensuremath{\mathrm{VERTEXMINOR}}}
\newcommand{\kVM}{\ensuremath{k\mathrm{-VERTEXMINOR}}}
\newcommand{\SVM}{\ensuremath{\mathrm{STARVERTEXMINOR}}}
\newcommand{\kSVM}{\ensuremath{k\mathrm{-STARVERTEXMINOR}}}
\newcommand{\QM}{\ensuremath{\mathrm{QUBITMINOR}}}
\newcommand{\SOET}{\ensuremath{\mathrm{SOET}}}
\newcommand{\HAMSOET}{\ensuremath{\mathrm{HAMSOET}}}
\newcommand{\kSOET}{\ensuremath{k\mathrm{-SOET}}}
\newcommand{\NP}{\ensuremath{\mathbb{NP}}}
\newcommand{\Poly}{\ensuremath{\mathbb{P}}}
\newcommand{\CubHam}{\ensuremath{\mathrm{CubHam}}}
\newcommand{\DPP}[1]{\ensuremath{#1\mathrm{-DPP}}}
\newcommand{\LCLPMCC}{\ensuremath{\mathrm{LC}+\mathrm{LPM}+\mathrm{CC}}}

\newcommand{\bs}[1]{\ensuremath{\boldsymbol{#1}}}

\usepackage{algpseudocode} 
\usepackage{algorithm} 

\newtheorem{theorem}{Theorem}[section]
\newenvironment{thm}{$\vspace{-0.5em}$\begin{theorem}}{\hfill$\diamond$\end{theorem}}
\crefname{thm}{theorem}{theorems}

\newtheorem{corollary}{Corollary}[theorem]
\newenvironment{cor}{$\vspace{-0.5em}$\begin{corollary}}{\hfill$\diamond$\end{corollary}}
\crefname{cor}{corollary}{corollaries}

\newtheorem{lemma}{Lemma}[section]
\newenvironment{lem}{$\vspace{-0.5em}$\begin{lemma}}{\hfill$\diamond$\end{lemma}}
\crefname{lem}{lemma}{lemmas}

\newtheorem{definition}{Definition}[section]
\newenvironment{mydef}{$\vspace{-0.5em}$\begin{definition}}{\hfill$\diamond$\end{definition}}
\crefname{mydef}{definition}{definitions}

\newsavebox{\mybox}
\newtheorem{problem}{Problem}[section]
\newenvironment{prm}{$\vspace{-0.5em}$\begin{problem}}{\hfill$\diamond$\end{problem}}
\crefname{problem}{problem}{problems}

\usepackage{subcaption}
\usepackage[colorinlistoftodos,prependcaption]{todonotes}
\newcommandx{\unsure}[2][1=]{\todo[linecolor=red,backgroundcolor=red!25,bordercolor=red,#1]{#2}}
\newcommandx{\change}[2][1=]{\todo[linecolor=blue,backgroundcolor=blue!25,bordercolor=blue,#1]{#2}}
\newcommandx{\info}[2][1=]{\todo[linecolor=OliveGreen,backgroundcolor=OliveGreen!25,bordercolor=OliveGreen,#1]{#2}}
\newcommandx{\improvement}[2][1=]{\todo[linecolor=Plum,backgroundcolor=Plum!25,bordercolor=Plum,#1]{#2}}
\newcommandx{\thiswillnotshow}[2][1=]{\todo[disable,#1]{#2}}

\newcommand{\poly}{\mbox{poly}}

\title{How to transform graph states using single-qubit operations: \\
\large computational complexity and algorithms}
\author{Axel Dahlberg \and Jonas Helsen \and Stephanie Wehner}

\begin{document}
\maketitle

\begin{abstract}
    Graph states are ubiquitous in quantum information with diverse applications ranging from quantum network protocols to measurement based
quantum computing.
    Here we consider the question whether 
	one graph (\emph{source}) state can be transformed into another graph (\emph{target}) state, using a specific set of quantum operations (\LCLPMCC):
single-qubit Clifford operations (LC), single-qubit Pauli measurements (LPM) and classical communication (CC) between sites holding the individual qubits.

    We first show that deciding whether a graph state $\ket{G}$ can be transformed into another graph state $\ket{G'}$ using \LCLPMCC\ is \NP-Complete, even if $\ket{G'}$ is restricted to be the GHZ-state.
    However, we also provide efficient algorithms for two situations of practical interest: 
\begin{enumerate}
\item $\ket{G}$ has \emph{Schmidt-rank width} one and $\ket{G'}$ is a GHZ-state. 
The Schmidt-rank width is an entanglement measure of quantum states, meaning this algorithm is efficient if the original state has little entanglement. 
Our algorithm has runtime $\mathcal{O}(\abs{V(G')}\abs{V(G)}^3)$, and is 
also efficient in practice even on small instances as further showcased by a freely available software implementation.

\item $\ket{G}$ is in a certain class of states with unbounded Schmidt-rank width, and $\ket{G'}$ is a GHZ-state of a constant size. 
    Here the runtime is $\mathcal{O}(\poly(|V(G)|))$, showing that more efficient algorithms can in principle be found even for states holding a large amount
of entanglement, as long as the output state has constant size. 
\end{enumerate}

    Our results make use of the insight that deciding whether a graph state $\ket{G}$ can be transformed to another graph state $\ket{G'}$ is equivalent to a known decision problem in graph theory, namely the problem of deciding whether a graph $G'$ is a \emph{vertex-minor} of a graph $G$.
    We prove that the vertex-minor problem is \NP-Complete by relating it to a new decision problem on 4-regular graphs which we call the \emph{semi-ordered Eulerian tour} (SOET) problem, and show that a version of the Hamiltonian cycle problem can be reduced to SOET. 
	The SOET problem and many of the technical tools developed to obtain our results may be of independent interest.
\end{abstract}

\newpage

\tableofcontents

\newpage


\section{Introduction}
A key concept in realizing quantum technologies is the preparation of specific resource states, which then enable further quantum processing. 
For example, many quantum network protocols 
first ask to prepare a specific resource state that is shared amongst the network nodes, followed by a subsequent measurement and exchange of classical communication. 
The simplest instance of this concept is indeed quantum key distribution~\cite{ekert1991quantum,bennett2014quantum}, in which we first produce a maximally entangled state, followed by random measurements. 
Similarly, measurement based quantum computing~\cite{Raussendorf2001} proceeds by first preparing the quantum 
device in a large resource state, followed by subsequent measurements on the qubits. 

An important class of such resource states are known as graph states. These states can 
be described by a simple undirected and unweighted graph where the vertices correspond to the qubits of the state~\cite{Hein2006}.
Apart from their broad range of applications, an appealing feature of graph states is that they can be efficiently described classically. 
Specifically, to describe a graph state on $n$ qubits, only $\frac{n(n-1)}{2}$ bits are needed to specify the edges of the graph. This is in sharp contrast to the 
$2^n$ complex numbers required to describe a general quantum state~\cite{Nielsen2010}.
It turns out that for graph states, and indeed the more general class of \emph{stabilizer states}, their evolution under \emph{Clifford} operations and \emph{Pauli} measurement can be simulated efficiently on a classical computer~\cite{Gottesman1999}.

Well known applications of graph states include cluster states~\cite{nielsen2006cluster} used in measurement based quantum computing, where together with arbitrary single-qubit measurements, these states nevertheless do form a universal resource for measurement-based quantum computation~\cite{Raussendorf2001}.
Graph states also arise as logical codewords of many error-correcting codes~\cite{Schlingemann2002}.
In the domain of quantum networking, a specific class of graph states is of particular interest. Specifically, these are states which are GHZ-like, i.e., they are equivalent to the GHZ state up to single-qubit Clifford operations. 
GHZ-states have been shown to be useful for applications such as quantum secret sharing~\cite{Markham2008}, anonymous transfer~\cite{Christandl2005}, conference key agreement~\cite{Ribeiro2017} and clock synchronization~\cite{komar2014quantum}.
It turns out that graph states described by either a star graph or a complete graph are precisely those GHZ-states~\cite{Hein2006}.

Given the desire for graph states, we may thus ask how they can effectively be prepared, and transformed. We consider the situation in which we already have a specific starting state (the \emph{source} state), and we wish to transform it to a desired \emph{target} state, using an available set of operations. Motivated by the fact that on a quantum network or distributed quantum processor, local operations
are typically much faster and easier to implement, we consider the set of operations consisting of single-qubit Clifford operations (LC), single-qubit Pauli measurements (LPM), and classical communication (CC).
Applications of an efficient algorithm that finds a series of operations to transform a source to a target state includes
the ability to make effective routing - in this context, state preparation - decisions, on a distributed quantum processor or network. 
Here, fast decisions are essential since quantum memories are inherently noisy and the source state will therefore become useless if too much time is spent on forming a decision.
Such algorithms could also be used as a design tool in the study of quantum repeater schemes~\cite{azuma2015all},
and the discovery of effective code switching procedures in quantum error correction~\cite{gottesman1997stabilizer,nautrup2017fault}.

To understand which graph states are related under \LCLPMCC\ operations 
we introduced the notion of a \emph{qubit-minor} in~\cite{Dahlberg2018}.
A qubit-minor of a graph state $\ket{G}$ is another graph state $\ket{G'}$ such that $\ket{G}$ can be transformed to $\ket{G'}$ using only single-qubit Clifford operations, single-qubit Pauli measurement and classical communication.
It turns out that the question of whether a graph state has a certain qubit-minor or not, can be completely phrased in graph theoretical terms.
In this setting the single-qubit Clifford operations correspond to a graph operation called a \emph{local complementation}.
Consequently, the single-qubit Pauli measurements and classical communication correspond to local complementations and vertex-deletions.
In graph theory, a well studied notion is that of a \emph{vertex-minor}s of a graph, which are exactly the graphs reachable by performing local complementations and vertex-deletions.
Interestingly, we show in~\cite{Dahlberg2018} that the notion of qubit-minors are equivalent to vertex-minors, in the sense that $\ket{G'}$ is a qubit-minor of $\ket{G}$ if and only if $G'$ is a vertex-minor of $G$.

Vertex-minors play an important role in algorithmic graph theory, together with the notion of \emph{rank-width}, which is a complexity measure of graphs.
Specifically, one can efficiently decide membership of a graph in some set of graphs, if this set is closed under 
taking vertex-minors and of fixed (bounded) rank-width~\cite{Oum2005}.
An example of such a set is the set of \emph{distance-hereditary graphs}, which are in fact exactly the graphs with rank-width one~\cite{Oum2005}.
Another example of a set of graphs which is closed under taking vertex-minors are \emph{circle graphs}, which are however of unbounded rank-width(\cite[Proposition 6.3]{Oum2006} and \cite{Courcelle2008}).
An appealing connection between the rank-width of graphs, and the entanglement in the corresponding graph states was identified in~\cite{VandenNest2007}, where it is shown that the rank-width corresponds to the Schmidt-rank width of the graph state, which is an entanglement
measure.
Specifically, the higher rank-width a graph has, the more entanglement there is in terms of this measure.
Another interpretation of the \emph{Schmidt-rank width} is that it captures how complex the quantum state is.
One way to describe quantum states is by a technique called \emph{tree-tensor networks} and it was shown in~\cite{VandenNest2007} that the minimum dimension of the tensors needed to describe a state is in fact given by the \emph{Schmidt-rank width}.

In the domain of complexity theory, the rank-width and related measures such as the tree and clique-width~\cite{courcelle1993handle,bertele1972nonserial} also form a measure
of the inherent complexity of the underlying problem, and feature prominently in the study of fixed-parameter tractable (FPT) algorithms~\cite{Downey1999}. Specifically, a problem is called fixed-parameter tractable in terms of a parameter $r$, if any instance $I$ of the problem of 
fixed $r$, is solvable in time $f(r) \cdot |I|^{\mathcal{O}(1)}$, where $|I|$ is the size of the instance and $f$ is an arbitrary function 
of $r$~\cite{Downey1999}. In this work, the $r$ is the rank-width, and for graphs of constant rank-width 
the techniques of Courcelle~\cite{Courcelle2007} and its generalizations~\cite{Courcelle2011}, can be used to obtain polynomial time 
algorithms for problems such as Graph Coloring~\cite{Ganian2010}, or Hamiltonian Path~\cite{Langer2014}. 
While very appealing from a complexity theory point of view, however, a direct application of these techniques does however not usually lead to polynomial time algorithms that are also efficient in practice, since $f(r)$ is often prohibitively large.

Since the problem of deciding whether a graph state $\ket{G'}$ is a qubit-minor of $\ket{G}$ (\QM) is equivalent to deciding if $G'$ is a vertex-minor of $G$ (\VM)~\cite{Dahlberg2018}, an efficient algorithm for \VM, directly provide an efficient algorithm for \QM. This in turn can be used for fast decisions on how to transform graph states in a quantum network or distributed quantum processor.
However, not much was previously known about the computational complexity of \VM\ and therefore if efficient algorithms exists.
For a related but slightly more restrictive minor-relation, namely \emph{pivot-minors} it has been shown in~\cite{Pivot-minors2016} that checking whether a graph $G$ has a pivot-minor isomorphic to another graph $G'$ is \NP-Complete.
However the complexity of deciding whether $G'$ is a vertex-minor of $G$ was left as an open problem. What's more, we emphasize that 
for our application we are interested in preparing a specific target state $G'$ on a specific set of qubits, as qubits are generally not interchangeable in the applications of our algorithm. As such, our question is not whether we can obtain a graph that is isomorphic to $G'$, but rather whether we can prepare $G'$ on a specific set of vertices.

Evidently, for fixed rank-width, it is not difficult to apply the techniques of Courcelle~\cite{Courcelle2007}, to obtain an FPT algorithm for our problem that is efficient if both the size of $G'$, as well as the rank-width of $G$ are bounded 
 (as we have shown in~\cite{Dahlberg2018}).
Indeed, a powerful method for deciding if a graph problem is fixed-parameter tractable is by Courcelle's theorem and its generalizations~\cite{Courcelle2011}. It turns out that also for our case, however, a direct implementation of Courcelle's theorem does not give an algorithm that is efficient
in practice.
In fact, in the case of \VM, this constant factor obtained by applying the techniques of Courcelle in~\cite{Dahlberg2018} can be shown to be\footnote{Far greater than the number of atoms in the universe.} a power of twos
\begin{equation}\label{eq:tower}
    f(r) = 2^{2^{\cdot^{\cdot^{\cdot^{2^r}}}}}
\end{equation}
where $r$ is the rank-width of the input graph $G$ and the height of the tower is 10~\cite{Dahlberg2018}.

\subsubsection*{Results}
In this paper we determine the computational complexity of \VM\ and therefore of \QM.
In particular we prove that it is in general \NP-Complete to decide whether a graph $G'$ is a vertex-minor of another graph $G$.
We however also give efficient algorithms for this problem whenever the input graphs belong to particular graph classes.
Any overview of the complexity of the problem for different classes of graphs considered in this paper can be seen in \cref{fig:graph_classes}.

We point out that our results of \NP-Completeness and the presented algorithms also apply to the more general class of stabilizer states of relevance in quantum error correction.
This is because any stabilizer state can be transformed to some graph state using only single-qubit Clifford operations.
Furthermore, given a stabilizer state on $n$ qubits, a graph state equivalent under single-qubit Clifford operations can be found efficiently in time $\mathcal{O}(n^3)$~\cite{VandenNest2004}.

\begin{enumerate}
    \item We show that the problem of deciding whether a graph $G'$ is a vertex-minor of another graph $G$ is \NP-Complete.
        This implies that \QM\ is also \NP-Complete.
        The hardness part of the above completeness result is obtained by considering a less general version of \VM~where $G'$ is restricted to be a star graph, which corresponds to transforming graph states to GHZ states. We call this problem \SVM~and show that it is also \NP-Complete, even when $G$ is in a strict subclass of circle graphs.
        To show \NP-Completeness of \SVM~we introduce the concept of a \emph{semi-ordered Eulerian tour} (\SOET)\footnote{Pronounced as "suit".}. This new graph-theoretical structure might be of independent interest. We show that the problem of deciding whether a graph allows for a 
        \SOET~ (which we call the \SOET~problem) can be reduced to \SVM. Furthermore, we show that deciding whether a $3$-regular graph is Hamiltonian (\CubHam) can be reduced to \SOET.
        Since the \CubHam\ has previously been shown to be \NP-Complete, this implies that \VM~is as well, since we also show that \VM\ is in \NP.
        As a key technical tool we also introduce a new class of graphs which we call triangular-expanded graphs.
        These graphs allow us to relate Hamiltonian tours on $3$-regular graphs to certain Eulerian tours on $4$-regular graphs and in turn to reduce \CubHam\ to \SOET.
    \item We provide an efficient algorithm for \SVM, which we prove to be correct if $G$ is distance-hereditary, or equivalently if $G$ has rank-width one.
        The run-time of this algorithm is $\mathcal{O}(\abs{V(G')}\abs{V(G)}^3)$, where $V(G)$ denotes the vertex-set of $G$.
        This algorithm can therefore be used to decide how to transform graph states, with Schmidt-rank width one, to GHZ-states using single-qubit Clifford operations, single-qubit Pauli measurements and classical communication.
        As mentioned above, a more general method to find efficient algorithms for certain graph problems on graphs with bounded rank-width is by using Courcelle's theorem~\cite{Courcelle2011}.
        Compared to the algorithm provided by a direct implementation of Courcelle's theorem, see~\cite{Dahlberg2018}, our algorithm presented here does not suffer from a huge constant factor in the runtime, as in \cref{eq:tower}.
        In fact, we have implemented our efficient algorithm and see that it typically takes for example $~50$ ms to run for the case when $\abs{V(G)}=50$ on a standard desktop computer, see \cref{fig:runtimes}.
        Furthermore, our algorithm is still efficient even if the size of $G'$ is unbounded.

        Distance-hereditary graphs, and therefore graphs with rank-width one, are exactly the graphs that can be reached by adding leaves and performing twin-splits from a graph with one vertex~\cite{Mulder1986}.
        To prove that our algorithm is correct we also present some new interesting results relating vertex-minors, distance-hereditary graphs and leaves and twins.
        For example we show that if $v$ is a leaf or a twin in $G$ but not a vertex in $G'$, then $G'$ is a vertex-minor of $G$ if and only if $G'$ is a vertex-minor of $G\setminus v$, where $\setminus v$ denotes vertex-deletion.

        \begin{figure}[H]
            \centering
            \includegraphics[width=0.6\textwidth]{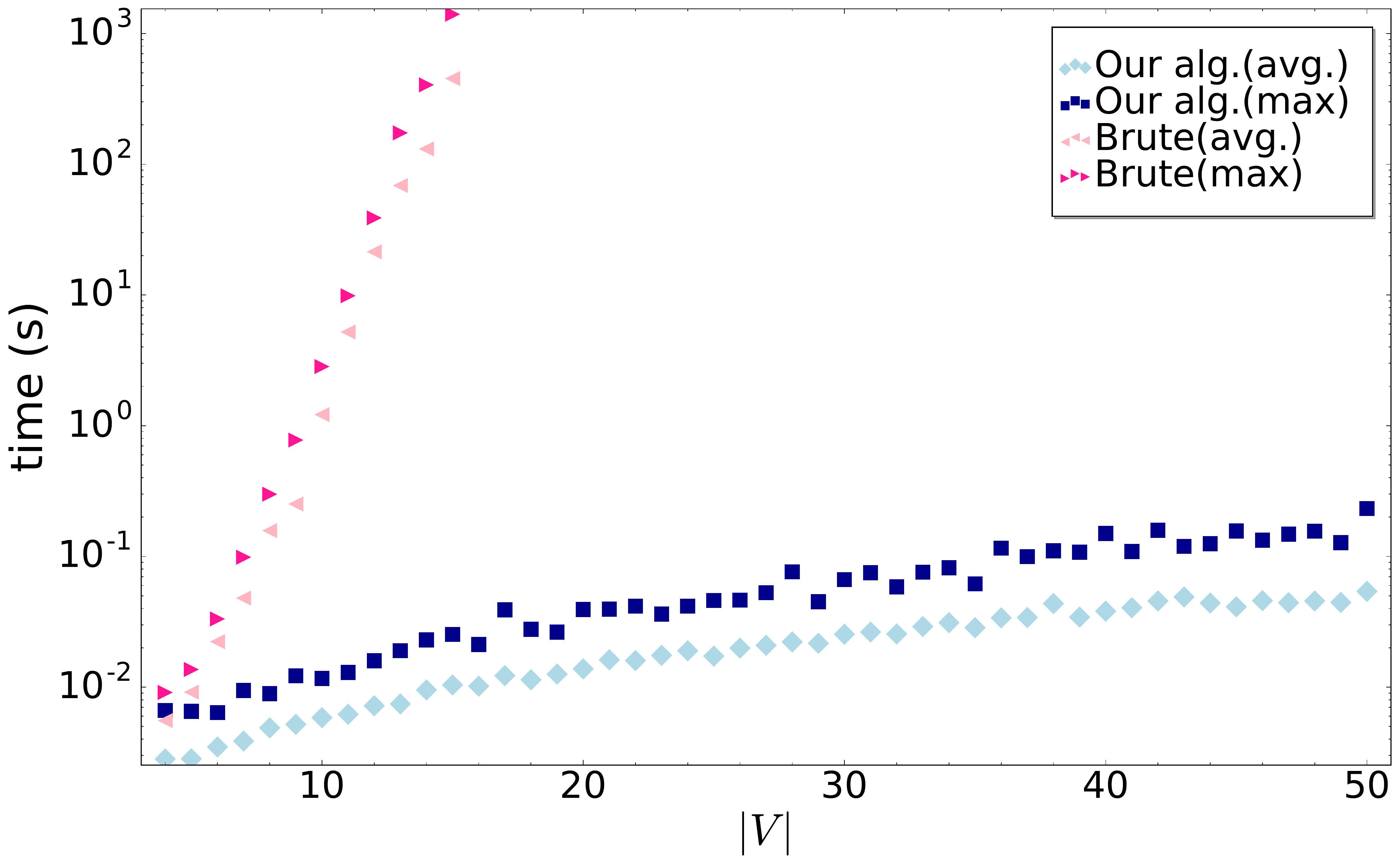}
            \caption{Average and maximal observed run-times for two algorithms that check if a star graph on four vertices is a vertex-minor of a randomly generated connected graph $G$ on vertices $V$.
                Random connected distance-hereditary graphs are generated by starting from a single-vertex graph and randomly adding leaves or performing twin-splits, see \cref{sec:DH}, which generates any connected distance-hereditary graph~\cite{Mulder1986}.
                    "Our alg." refers to the algorithm described in \cref{sec:thealg} and "Brute" is the non-efficient algorithm described in~\cite{Dahlberg2018}. The algorithm of~\cite{Dahlberg2018} based on the techniques of Courcelle~\cite{Courcelle2007} is not depicted here since
the pre factor makes an application impractical in practice whenever $|V| < f(r)$ of~\cref{eq:tower}.  
                    For each size of $V$, 10 random graphs are generated for "Brute" and 100 random graphs for "Our alg.", from which the average "(avg.)" and max "(max)" runtime is computed.
                    Both algorithms are implemented in SAGE~\cite{sage} and the tests were performed on an iMac with 3.2 GHz Intel Core i5 processor with 8 GB of 1600 MHz RAM.
            }
            \label{fig:runtimes}
        \end{figure}

    \item We call \kSVM\ the restriction of \SVM\ where $G'$ is restricted to a star graph having $k$ vertices, corresponding to a GHZ-state (up to LC) on $k$ qubits.
        We show that \kSVM\ is in \Poly\ if $G$ is a circle graph\footnote{Not to be confused with cycle graph.}.
        This means that \SVM\ is fixed-parameter tractable in the size of $G'$ on circle graphs.
        Interestingly the class of circle graphs has unbounded rank-width~(\cite[Proposition 6.3]{Oum2006} and \cite{Courcelle2008}) and the corresponding graph states therefore have unbounded entanglement according to the Schmidt-rank width.
    \item We show that any connected graph $G'$ on three vertices or less is a vertex-minor of any connected graph $G$ if and only if the vertices of $G'$ are also in $G$. An efficient algorithm for finding the transformation that takes the former graph to the latter is also provided.
    \item We also prove several technical results that may be interesting in their own right. An example of this is \cref{thm:LU} which points out the interesting behavior of a certain class of graphs with respect to a bipartition of their vertices. Another example would be \cref{thm:T_size} where we show that any distance heredity graph one more than four vertices has a \emph{foliage} (the set of leaves, axils and twins in the graph) of size at least four.
\end{enumerate}

\begin{figure}[H]
    \centering
    \includegraphics[width=0.8\textwidth]{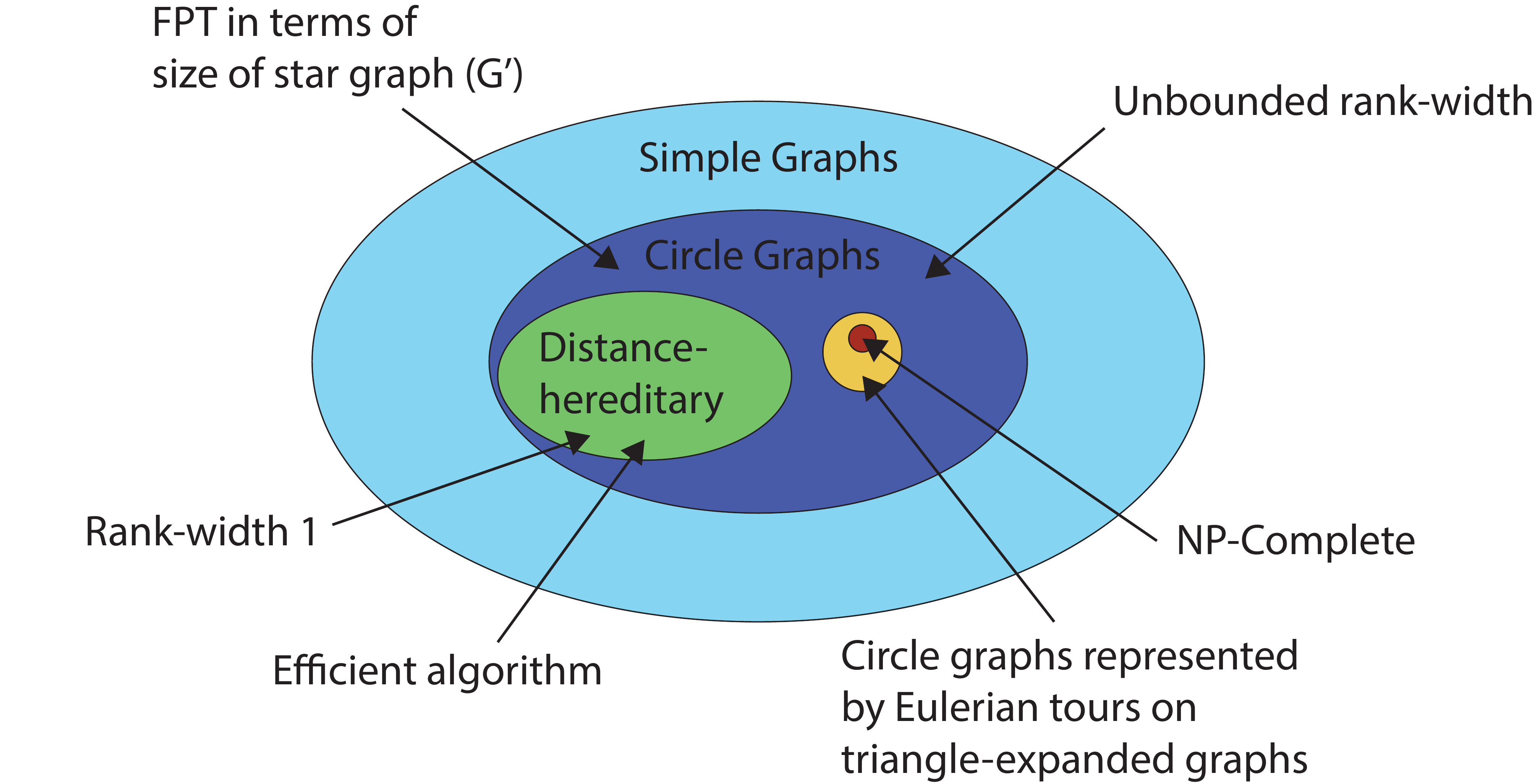}
    \caption{An overview of the graph classes discussed in this paper and what the computational complexity of solving \SVM\ on these classes is.
        The sizes of the sets in the figure are not exact, however their intersections and non-intersections are.
        The definitions of these classes graphs are given in \cref{sec:circle} and \cref{sec:DH}.
        We show that the \SVM\ is \NP-Complete (result 1.) on graphs in a strict subclass of circle graphs (red).
        The same therefore holds for any superclass and in particular for circle graphs (dark blue).
        On the other hand we show that \kSVM\ is in \Poly\ on circle graphs (result 3.)
        On distance-hereditary graphs (green) we found an efficient algorithm for \SVM\ (result 2.)
        Finally it is still an open question if \kSVM\ is \NP-Complete on simple graphs (light blue) in general.}
    \label{fig:graph_classes}
\end{figure}

\subsubsection*{Overview}
The paper is structured as follows.
In \cref{sec:prel} we describe graph states and consider several notions of graph theory we will need throughout the paper.
We also introduce the \VM~and \SVM~problems and the notion of a semi-ordered Eulerian tour.
We also prove a few technical results concerning distance hereditary graphs, circle graphs and vertex-minors which we will need later.
Result 1. above is given in \cref{sec:complexity}, result 2. in \cref{sec:DH_alg}, result 3. in \cref{sec:FPT} and result 4. in \cref{sec:small_star}.


\section{Preliminaries}\label{sec:prel}

In this section we set our notation and recall various concepts which will be used throughout the rest of the paper.
We start by providing the definitions of graph states, qubit-minors and the relation to vertex-minors.
We then recall the definitions of local complementation and vertex deletion as operations on graphs and use these to state the problems \VM~and \SVM~which are the core objects of study in this paper. Furthermore, we discuss circle graphs and their various characterizations and discuss how local complementation behaves on these graphs. We also introduce the concept of semi-ordered Eulerian tours, which is a key technical concept for the results later in this paper. Finally we discuss distance-hereditary graphs, which form a subclass of circle graphs. We discuss how these graphs can be built up out of elementary pieces and prove some technical results which will be used later.

\subsection{Notation and definitions}
Here we introduce some notation and vocabulary that will be used throughout this paper.
We assume familiarity of the general notation of quantum information theory, see~\cite{Nielsen2010} for more details.

\underline{\emph{Quantum operations}}\newline
The Pauli matrices will be denoted as
\begin{equation}
    \mathbb{I}=\begin{pmatrix}1 & 0\\0 & 1\end{pmatrix},\quad X=\begin{pmatrix}0 & 1\\1 & 0\end{pmatrix},\quad Y=\begin{pmatrix}0 & -\mathrm{i}\\\mathrm{i} & 0\end{pmatrix},\quad Z=\begin{pmatrix}1 & 0\\0 & -1\end{pmatrix}.
\end{equation}
The single-qubit Clifford group $\mathcal{C}$ consists of operations which leave the Pauli group $\mathcal{P}=\langle \mathrm{i}\mathbb{I},X,Z\rangle$ invariant.
More formally, $\mathcal{C}$ is the normalizer of the Pauli group, i.e.
\begin{equation}
    \mathcal{C}=\big\{C\in \mathcal{U}:(\forall P\in\mathcal{P}:CPC^\dagger\in\mathcal{P})\big\},
\end{equation}
where $\mathcal{U}$ is the single-qubit unitary operations.

\underline{\emph{Sequences and words}}\newline
A \emph{sequence} $\bs{X}=x_1x_2\dots x_k$ is an ordered, possibly empty, tuple of elements in some set $X$.
We also call a sequence a \emph{word} and its elements \emph{letters}.
We write $\bs{X}\subseteq X$, when all letters of $\bs{X}$ are in the set $X$.
A \emph{sub-word} $\bs{X}'$ of $\bs{X}$, is a word which can be obtained from $\bs{X}$ by iteratively deleting the first or last element of $\bs{X}$.
We denote the concatenation of two words $\bs{X}_1=x_1\dots x_{k_1}$ and $\bs{X}_2=y_1\dots y_{k_2}$ as $\bs{X}_1\Vert\bs{X}_2=x_1\dots x_{k_1}y_1\dots y_{k_2}$. We also denote the `mirror image' by an overset tilde, e.g. if $\bs{X} =ab$ then $\widetilde{\bs{X}} = ba$.

\underline{\emph{Sets}}\newline
The set containing the natural numbers from $1$ to $n$ is denoted $[n]$.
The symmetric difference $X\Delta Y$ between two sets $X$ and $Y$ is the set of elements of $X$ and $Y$ that occur in $X$ or $Y$ exclusively, i.e. $X\Delta Y=(X\cup Y)\setminus (X\cap Y)$.

\underline{\emph{Graphs}}\newline
A simple undirected graph $G=(V,E)$ is a set of vertices $V$ and a set of edges $E$.
Edges are 2-element subsets of $V$ for simple undirected graphs.
Importantly, we only consider labeled graphs, i.e. we consider a complete graph with vertices $\{1,2,3\}$ to be different from a complete graph with vertices $\{2,3,4\}$, even though these graphs are isomorphic.
The reason for considering labeled graphs is that these will be used to represent graph states on specific qubits, possibly at different physical locations in the case of a quantum network.
In a simple undirected graph, there are no multiple edges or self-loops, in contrast with a multi-graph: An undirected multi-graph $H=(V,E)$ is a set of vertices $V$ and a multi-set of edges $E$.
For undirected multi-graphs, edges are unordered pairs of elements in $V$.
We will often write $V(G)=V$ and $E(G)=E$ to mean the vertex- and edge-set of the (multi-)graph $G=(V,E)$.

Next we list some glossary about (multi-)graphs:
\begin{itemize}
    \item If a vertex $v\in V$ is an element of an edge $e\in E$, i.e. $v\in e$, then $v$ and $e$ are said to be \emph{incident} to one another.
    \item Two vertices which are incident to a common edge are called \emph{adjacent}.
    \item The set of all vertices adjacent to a given vertex $v$ in a (multi-)graph $G$ is called the \emph{neighborhood} $N^{(G)}_v$ of $v$.
        We will sometimes just write $N_v$ if it is clear which (multi-)graph is considered.
    \item The \emph{degree} of a vertex $v$ is the number of neighbors of $v$, i.e. $\abs{N_v}$.
    \item A $k$-\emph{regular} (multi-)graph is a (multi-)graph such that every vertex in the (multi-)graph has degree $k$.
    \item A \emph{walk} $W=v_1e_1v_2\dots e_kv_{k+1}$ is an alternating sequence of vertices and edges such that $e_i$ is incident to $v_{i}$ and $v_{i+1}$ for $i\in[k]$.
    \item The vertices $v_1$ and $v_{k+1}$ are called the \emph{ends} of $W$.
    \item If the ends of a walk are the same vertex, it is called \emph{closed}.
    \item A \emph{trail} is a walk which does not include any edge twice.
    \item A closed trail is called a \emph{tour}.
    \item A \emph{path} is a walk which does not include any vertex twice, apart from possibly the ends.
    \item A closed path is called a \emph{cycle}.
    \item Two vertices $u$ and $v$ are called \emph{connected} if there exist a path with $u$ and $v$ as ends.
    \item A (multi-)graph is called \emph{connected} if any two vertices are connected in the (multi-)graph.
    \item $G'=(V',E')$ is a \emph{subgraph }of $G=(V,E)$ if $V'\subseteq V$ and $E'\subseteq E$.
    \item An \emph{induced subgraph} $G[V']$ of $G=(V,E)$ is a subgraph on a subset $V'\subseteq V$ and with the edge-set 
        \begin{equation}
            E'=\{(u,v)\in E:u,v\in V'\}.
        \end{equation}
    \item A \emph{connected component} of a (multi-)graph $G=(V,E)$ is a connected induced subgraph $G[V']$ such that no vertex in $V'$ is adjacent to a vertex in $V\setminus V'$ in the (multi-)graph $G$.
    \item A \emph{cut-vertex} $v$ of a (multi-)graph $G = (V,E)$ is a vertex such that $G[V\setminus \{v\}]$ has strictly more connected components than $G$
    \item The \emph{distance} $d_G(v,w)$ between two vertices $v,w$ in a (multi-)graph $G$ is equal to the number of edges in the shortest path that connects $v$ and $w$.
    \item The \emph{complement} $G^C$ of a graph $G=(V,E)$ is a graph with vertex-set $V^C=V$ and edge-set
        \begin{equation}
            E^C=\{(u,v)\in V\times V:(u,v)\notin E\land u\neq v\}.
        \end{equation}
\end{itemize}

A graph $G$ is assumed to be simple and undirected, unless specified and will be denoted as $G$, $G_i$, $G'$, $\tilde{G}$ or similar.
A multi-graph is assumed to be undirected, unless specified and will be denoted as $H$, $H_i$, $H'$, $\tilde{H}$ or similar.
Furthermore, 3-regular simple graphs will be denoted as $R$, $R_i$, $R'$ or similar and 4-regular multi-graphs as $F$, $F_i$, $F'$ or similar.
We will denote the \emph{complete graph} on a set of vertices $V$ as $K_V$ and the \emph{star graph} on a set of vertices $V$ with center $c$ by $S_{V,c}$. We will often not care about the choice of center, writing $S_V$ to mean any choice of star graph on the vertex set $V$.

\subsection{Graph states}\label{sec:graph_states}
A graph state is a multi-partite quantum state $\ket{G}$ which is described by a graph $G$, where the vertices of $G$ correspond to the qubits of $\ket{G}$.
The graph state is formed by initializing each qubit $v\in V(G)$ in the state $\ket{+}_v=\frac{1}{\sqrt{2}}(\ket{0}_v+\ket{1}_v)$ and for each edge $(u,v)\in E(G)$ applying a controlled phase gate between qubits $u$ and $v$.
Importantly, all the controlled phase gates commute and are invariant under changing the control- and target-qubits of the gate.
This allows the edges describing these gates to be unordered and undirected.
Formally, a graph state $\ket{G}$ is given as
\begin{equation}
    \ket{G}=\prod_{(u,v)\in E(G)}C_Z^{(u,v)}\left(\bigotimes_{v\in V(G)}\ket{+}_v\right),
\end{equation}
where $C_Z^{(u,v)}$ is a controlled phase gate between qubit $u$ and $v$, i.e.
\begin{equation}
    C_Z^{(u,v)}=\ket{0}\bra{0}_u\otimes\mathbb{I}_v+\ket{1}\bra{1}_u\otimes Z_v
\end{equation}
and $Z_v$ is the Pauli-$Z$ matrix acting on qubit $v$.

Any graph state is also a stabilizer state~\cite{Hein2006}.
The GHZ states are an important class of stabilizer states given as
\begin{equation}
    \ket{\mathrm{GHZ}}_k=\frac{1}{\sqrt{2}}\left(\ket{0}^{\otimes k}+\ket{1}^{\otimes k}\right).
\end{equation}
It is easy to verify that any graph state given by a star or complete graph, i.e. $\ket{S_{V,c}}$ or $\ket{K_V}$, can be turned into a GHZ state on the qubits $V$ using only single-qubit Clifford operations.

In the next section we discuss local complementations and vertex-deletions on graph states.
It turns out that single-qubit Clifford operations (LC), single-qubit Pauli measurements (LPM) and classical communication (CC): \LCLPMCC, which take graph states to graph states, can be completely characterized by local complementations and vertex-deletions on the corresponding graphs.
More concretely, any sequence of single-qubit Clifford operations, mapping graph states to graph states, can be described as some sequence of local complementations on the corresponding graph.
Moreover, measuring qubit $v$ of a graph state $\ket{G}$ in the Pauli-$X$, Pauli-$Y$ or Pauli-$Z$ basis, gives a stabilizer state that is single-qubit Clifford equivalent to $\ket{X_v(G)}$, $\ket{Y_v(G)}$, $\ket{Z_v(G)}$ respectively.
The operations $X_v$, $Y_v$ and $Z_v$ are graph operations consisting of sequences of local complementations together with the deletion of vertex $v$, which we define in \cref{def:XYZ}.
As mentioned the post-measurement state of for example a Pauli-$X$ measurement on qubit $v$ is only single-qubit Clifford equivalent to the graph state $\ket{X_v(G)}$.
The single-qubit Clifford operations that take the post-measurement state to $\ket{X_v(G)}$ depend on the outcome of the measurement of the qubit $v$ and act on qubits adjacent to $v$~\cite{Hein2006}.
This means classical communication is required to announce the measurement result at the vertex $v$ to its neighboring vertices.

In~\cite{Dahlberg2018} we introduced to notion of a \emph{qubit-minor} which captures exactly which graph states can be reached from some initial graph state under  \LCLPMCC.
Formally we define a qubit-minor as:

\begin{mydef}[Qubit-minor]\label{def:QM} Assume $\ket{G}$ and $\ket{G'}$ are graph states on the sets of qubits $V$ and $U$ respectively.
    $\ket{G'}$ is called a qubit-minor of $\ket{G}$ if there exists a sequence of single-qubit Clifford operations (LC), single-qubit Pauli measurements (LPM) and classical communication (CC) that takes $\ket{G}$ to $\ket{G'}$, i.e.
    \begin{equation}
        \ket{G}\xrightarrow[\mathrm{LPM}+\mathrm{CC}]{\mathrm{LC}}\ket{G'}\otimes\ket{\mathrm{junk}}_{V\setminus U}.
    \end{equation}
    If $\ket{G'}$ is a qubit-minor of $\ket{G}$, we denote this as
    \begin{equation}
        \ket{G'}<\ket{G}.
    \end{equation}
\end{mydef}

In~\cite{Dahlberg2018} we have shown that the notion of qubit-minors for graph states is equivalent to the notion of \emph{vertex-minors} for graphs.
We will define and discuss vertex-minors in \cref{sec:vertex-minors}, however we formally state the relation between vertex-minors here.
For a proof see~\cite{Dahlberg2018}.

\begin{thm}[Theorem 2.2 in~\cite{Dahlberg2018}]\label{thm:QMVM}
    Let $\ket{G}$ and $\ket{G'}$ be two graph states such that no vertex in $G'$ has degree zero.
    Then $\ket{G'}$ is a qubit-minor of $\ket{G}$ if and only if $G'$ is a vertex-minor of $G$, i.e.
    \begin{equation}
        \ket{G'}<\ket{G}\quad\Leftrightarrow\quad G'<G.
    \end{equation}
\end{thm}

Note that one can also include the case where $G'$ has vertices of degree zero.
Let's denote the vertices of $G'$ which have degree zero as $I$.
We then have that
\begin{equation}
    \ket{G'}<\ket{G}\quad\Leftrightarrow\quad G'[V(G')\setminus I]<G.
\end{equation}

\Cref{thm:QMVM} is very powerful since it allows us to consider graph states under \LCLPMCC, purely in terms of vertex-minors of graphs.
We will therefore in the rest of this paper use the formalism of vertex-minors to study the computational complexity of transforming graph states using \LCLPMCC~and provide efficient algorithms for transforming graph state using \LCLPMCC.

\subsection{Local complementations and vertex-deletions}
Local complementation is a fundamental operation on graphs~\cite{Bouchet1990}.
We have the following definition.

\begin{mydef}[Local complementation]\label{def:LC}
    A local complementation $\tau_v$ is a graph operation specified by a vertex $v$, taking a graph $G$ to $\tau_v(G)$ by replacing the induced subgraph on the neighborhood of $v$, i.e.  $G[N_v]$, by its complement.
    The neighborhood of any vertex $u$ in the graph $\tau_v(G)$ is therefore given by
    \begin{equation}
    N_u^{(\tau_v(G))}=\begin{cases}N_u\Delta (N_v\setminus\{u\}) & \quad \text{if } (u,v)\in E(G) \\ N_u & \quad \text{else}\end{cases},
    \end{equation}
    where $\Delta$ denotes the symmetric difference between two sets.
    Given a sequence of vertices $\bs{v}=v_1\dots v_k$, we denote the induced sequence of local complementations, acting on a graph $G$, as
    \begin{equation}
        \tau_{\bs{v}}(G)=\tau_{v_k}\circ\dots\circ\tau_{v_1}(G).
    \end{equation}
\end{mydef}
Below we show a simple example of the action of local complementation on a graph (in particular we consider a local complementation on the vertex labeled $2$).
\begin{equation}\label{eq:vis_LC}
    \raisebox{-0.08\textwidth}{\includegraphics[scale=0.5]{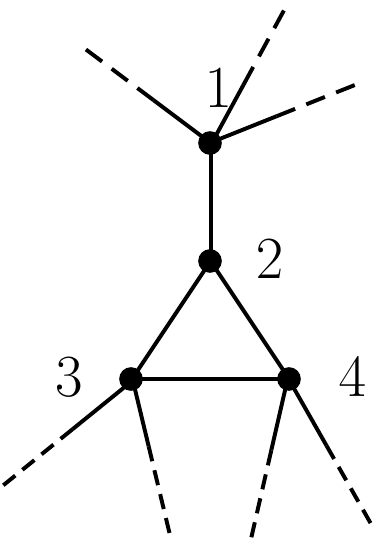}}\quad\xrightarrow{\tau_2}\quad\raisebox{-0.08\textwidth}{\includegraphics[scale=0.5]{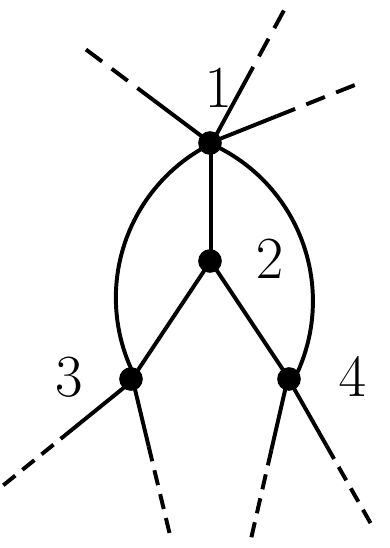}}
\end{equation}

If two graphs $G_1$ and $G_2$ are related by a sequence of local complementations, i.e. $\exists \bs{v}:\tau_{\bs{v}}(G_1)=G_2$, we call the two graphs LC-equivalent and denote this as $G_1\sim_\mathrm{LC}G_2$.
Checking whether two graphs are LC-equivalent can be done in time $\mathcal{O}(n^4)$, where $n$ is the size of the graphs, as shown in~\cite{Bouchet1991}.
This result was used in~\cite{VandenNest2004a} to find an efficient algorithm for checking whether two graph states are equivalent under single-qubit Clifford operations, by proving that two graph state are equivalent under single-qubit Cliffords if and only if their corresponding graphs are LC-equivalent.

Notable about local complementation is its action on star and complete graphs. For a vertex set $V$ and $c\in V$ we have that $\tau_c(S_{V,c})= K_V$ and for any $v \in V$ we have $\tau_v(K_V)= S_{V,v}$. This means all star graphs on a vertex set $V$ are equivalent to each other under local complementation and also to the complete graph on $V$. Moreover, no other graph is equivalent to the star or complete graph.

Another operation which we will make use of is the pivot.

\begin{mydef}[Pivot]\label{def:pivot}
    A pivot $\rho_e$ is a graph operation specified by an edge $e=(u,v)$, taking a graph $G$ to $\rho_e(G)$ such that
    \begin{equation}
        \rho_e(G)=\tau_v\circ\tau_u\circ\tau_v(G).
    \end{equation}
\end{mydef}
The pivot can simply be specified by an undirected edge since
\begin{equation}
    \tau_v\circ\tau_u\circ\tau_v(G)=\tau_u\circ\tau_v\circ\tau_u(G)\quad \text{if }(u,v)\in E(G)
\end{equation}
as shown in~\cite{Bouchet1988a}.

It will be useful to be able to specify a pivot simply by a vertex $v$.
We therefore also introduce the following definition:
\begin{mydef}\label{def:pivot_vertex}
    The graph operation $\rho_v$ is specified by a vertex, taking a graph $G$ to $\rho_v(G)$ such that
    \begin{equation}
        \rho_v(G)=\begin{cases}\rho_{e_v}(G) & \quad\text{if }\abs{N_v}>0 \\ G & \quad\text{if }\abs{N_v}=0\end{cases}
    \end{equation}
    where $e_v$ is an edge incident on $v$ chosen in some consistent way.
    For example we could assume that the vertices of $G$ are ordered and that $e_v=(v,\min(N_v))$.
    The specific choice will not matter but importantly $e_v$ only depends on $G$ and $v$, and the same therefore holds for $\rho_v(G)$.
\end{mydef}

Another fundamental operation on a graph is that of vertex-deletion, which relates to measuring a qubit of a graph state in the standard basis~\cite{Hein2006}.
We denote the deletion of vertex $v$ from the graph $G$ as $G\setminus v=G[V(G)\setminus \{v\}]$.
It turns out that given a sequence of local complementations and vertex-deletions, acting on some graph, one can always perform the vertex-deletions at the end of the sequence and arrive at the same graph.
This fact follows inductively from the following lemma.

\begin{lem}\label{lem:VD_last}
    Let $G=(V,E)$ be a graph and $v,u\in V$ be vertices such that $v\neq u$, then
    \begin{equation}
        \tau_v(G\setminus u)=\tau_v(G)\setminus u.
    \end{equation}
\end{lem}
\begin{proof}
    Note first that it is important that $v\neq u$ since the operation $\tau_v(G\setminus u)$ is otherwise undefined.
    To prove that the graphs $G_1=\tau_v(G\setminus u)$ and $G_2=\tau_v(G)\setminus u$ are equal, we show that the neighborhoods of any vertex in the graphs are the same, i.e. $N^{(G_1)}_w=N^{(G_2)}_w$ for all $w\in V(G)\setminus u$.
    The local complementation only changes the neighborhoods for vertices which are adjacent to $v$, so for any vertex $w\neq u$ which is not adjacent to $v$, we have that
    \begin{equation}
        N^{(G_1)}_w=N^{(G_2)}_w=N^{(G)}_w\setminus \{u\}.
    \end{equation}
    On the other hand, for a vertex $w$ which is adjacent to $v$, its neighborhood becomes
    \begin{equation}
        N^{(G_1)}_w=\big(N^{(G)}_w\setminus \{u\}\big)\Delta\Big(\big(N^{(G)}_v\setminus \{u\}\big)\setminus\{w\}\Big)=\Big(N^{(G)}_w\Delta\big(N^{(G)}_v\setminus \{w\}\big)\Big)\setminus\{u\}=N^{(G_2)}_w
    \end{equation}
    by the definition of a local complementation.
\end{proof}

\subsubsection{Vertex-minors}\label{sec:vertex-minors}

Using the two operations local complementation and vertex-deletion, we can formulate the notion of a vertex-minor of a graph.

\begin{mydef}[Vertex-minor]\label{def:vertex-minor}
    A graph $G'$ is called a vertex-minor of $G$ if and only if there exist a sequence of local complementations and vertex-deletions that takes $G$ to $G'$.
    Since vertex-deletions can always be performed last in such a sequence (see \cref{lem:VD_last}), an equivalent definition is the following: A graph $G'$ is called a vertex-minor of $G$ if and only if there exist a sequence of local complementations $\bs{v}$ such that $\tau_{\bs{v}}(G)[V(G')]=G$.
    If $G'$ is a vertex-minor of $G$ we write this as
    \begin{equation}
        G'<G
    \end{equation}
    and if $G'$ is not a vertex-minor of $G$ then
    \begin{equation}
        G'\nless G.
    \end{equation}
\end{mydef}

Vertex-minors were first studied in~\cite{Bouchet1988a} but by the name of $l$-reductions.
Note that if $G_1$ and $G_2$ are two LC-equivalent graphs, then $G'<G_1$ if and only if $G'<G_2$.
It is interesting to consider under which conditions a graph $G'$ is a vertex-minor of another graph $G$.
As \cref{thm:multi_vertex-minor} below states, to decide whether $G'<G$ it is sufficient to check whether $G'$ is LC-equivalent to at least one out of $3^{\abs{V(G)}-\abs{V(G')}}$ graphs.
To formally state the theorem we introduce the following three operations.

\begin{mydef}\label{def:XYZ}
    The graph operations $X_v$, $Y_v$ and $Z_v$, specified with a vertex $v$, act on a graph $G$ by transforming it to
    \begin{equation}
        X_v(G)=\rho_v(G)\setminus v,\quad Y_v(G)=\tau_v(G)\setminus v,\quad Z_v(G)=G\setminus v
    \end{equation}
    When we need to specify which edge incident on $v$ the pivot of $X_v$ acts on, we write $X^{(u)}_v(G)=\rho_{(u,v)}(G)\setminus v$.
\end{mydef}

The three graph operations $X_v$, $Y_v$ and $Z_v$ correspond to how Pauli-$X$, -$Y$ and -$Z$ measurements act on graph states (as proven in~\cite{Hein2006}).
As mentioned in \cref{sec:graph_states}, measuring qubit $v$ of a graph state $\ket{G}$ in the Pauli-$X$, -$Y$ or -$Z$ basis gives a stabilizer state which is single-qubit Clifford equivalent to $\ket{X_v(G)}$, $\ket{Y_v(G)}$ and $\ket{Z_v(G)}$ respectively.
\Cref{eq:measZ,eq:measY,eq:measX} shows examples of how these operations can act on graphs.
\begin{align}
    Z_6\left(\raisebox{-0.06\textwidth}{\includegraphics[scale=0.5]{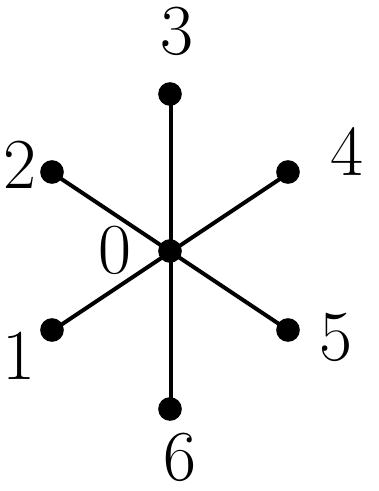}}\right)\;&=\;\raisebox{-0.04\textwidth}{\includegraphics[scale=0.5]{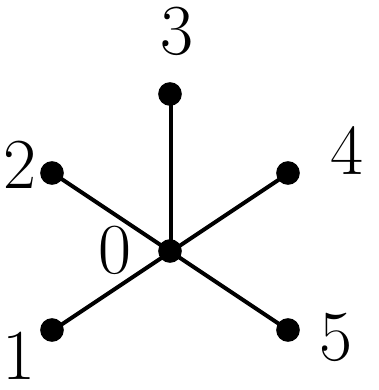}}\label{eq:measZ}\\
    Y_5\left(\raisebox{-0.05\textwidth}{\includegraphics[scale=0.5]{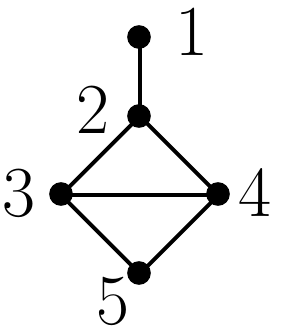}}\right)&=\raisebox{-0.04\textwidth}{\includegraphics[scale=0.5]{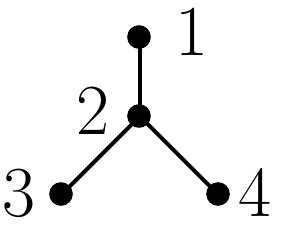}}\label{eq:measY}\\
    X^{(2)}_1\left(\raisebox{-0.05\textwidth}{\includegraphics[scale=0.5]{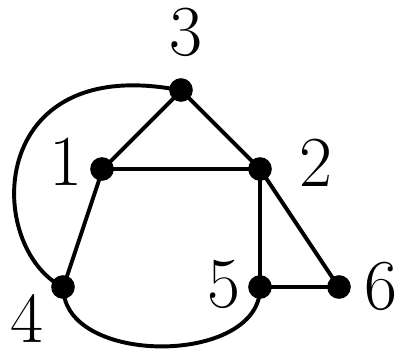}}\right)&=\raisebox{-0.05\textwidth}{\includegraphics[scale=0.5]{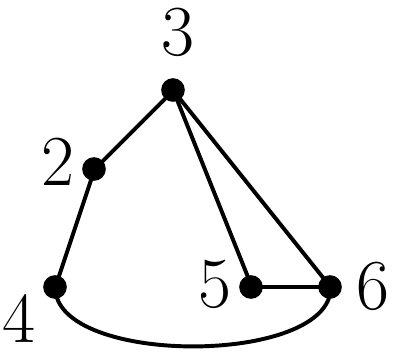}}\label{eq:measX}
\end{align}

The operation $X_v^{(u)}$ is the most complicated one, so we will here quickly describe what happens to a graph when $X_v^{(u)}$ is applied.
One can check that after the operation $X_v^{(u)}$, the vertex $u$ will have the neighbors that $v$ previously had, except $v$ itself.
Furthermore, some edges between vertices in $(N_v\cup N_u)\setminus\{u,v\}$ will be complemented, i.e. removed if present or added if not.
To know which of these edges gets complemented, let's introduce the following three sets
\begin{equation}\label{eq:CFDpartition}
    V_{vu}=N_v\cap N_u,\quad V_v=N_v\setminus(N_u\cup\{u\}),\quad V_u=N_u\setminus(N_v\cup\{v\})
\end{equation}
which form a partition of $(N_v\cup N_u)\setminus\{u,v\}$.
In \cref{eq:measX}, these sets are $V_{12}=\{3\}$, $V_1=\{4\}$ and $V_2=\{5,6\}$.
An edge $(w_1,w_2)$ between vertices in $(N_v\cup N_u)\setminus\{u,v\}$ gets complemented if and only if $w_1$ and $w_2$ belong to different sets of the partition $(V_{vu},V_v,V_u)$.
All other edges in the graph, i.e. edges containing a vertex not in $N_v\cup N_v$, will be unchanged.\\

It turns out that the three operations $\{X_v,Y_v,Z_v\}$ are sufficient to check whether some graph is a vertex-minor of another graph.
This is formalized in the following theorem we proved in~\cite{Dahlberg2018}.

\begin{thm}[Theorem 3.1 in~\cite{Dahlberg2018}]\label{thm:multi_vertex-minor}
    Let $G$ and $G'$ be two graphs and let $\mathbf{u}=(v_1,\dots,v_l)$, where $l=\abs{  V(G)\setminus V(G')}$ be an ordered tuple such that each element of   $V(G)\setminus V(G')$ occurs exactly once in $\mathbf{u}$.
    Furthermore, let $\mathcal{P}_\mathbf{u}$ denote the set of graph operations
    \begin{equation}\label{eq:P_u}
        \mathcal{P}_\mathbf{u}=\{P_{v_l}\circ\dots\circ P_{v_1}\;:\;P_{v}\in\{X_v,Y_v,Z_v\}\}
    \end{equation}
    Then we have that
    \begin{equation}
        G'<G\quad\Leftrightarrow\quad \exists P\in\mathcal{P}_\mathbf{u}\;:\;G'\sim_\mathrm{LC}P(G).
    \end{equation}
\end{thm}
Note that in~\cite{Dahlberg2018} we indexed $\mathcal{P}_\mathbf{u}$ simply with the set associated to the word $\bs{u}$ since the statement is independent of the ordering of the elements of $\bs{u}$.
A direct corollary of the above theorem is therefore:

\begin{cor}\label{cor:order_independent}
    Let $G$ and $G'$ be two graphs.
    Furthermore, let $\mathbf{u}$ and $\mathbf{u}'$ be two ordered tuples such that each element of $V(G)\setminus V(G')$ occurs exactly once in both $\mathbf{u}$ and $\mathbf{u}'$.
    Then we have that
    \begin{equation}\label{eq:pu_eq_pup}
        \exists P\in\mathcal{P}_\mathbf{u}\;:\;G'\sim_\mathrm{LC}P(G)\quad\Leftrightarrow\quad \exists P\in\mathcal{P}_{\mathbf{u}'}\;:\;G'\sim_\mathrm{LC}P(G).
    \end{equation}
\end{cor}
\begin{proof}
    This follows directly from \cref{thm:multi_vertex-minor} since both sides in \cref{eq:pu_eq_pup} are true if and only if $G'<G$.
\end{proof}

Note that \cref{thm:multi_vertex-minor} does not give an efficient method to check if $G'$ is a vertex-minor of $G$, since the set $\mathcal{P}_\mathbf{u}$ is of exponential size for all $\bf{u}$.
To study this problem further we formally define the vertex-minor problem.

\begin{prm}[\VM]\label{prob:vertex_minor}
Given a graph $G$ and a graph $G'$ defined on a subset of $V(G)$, decide whether $G'$ is a vertex-minor of $G$.
\end{prm}
Note again that we deal with labeled graphs here. We will often consider the special case where $G'$ is a star graph $S_{V'}$ defined on a subset $V'$ of $V(G)$.
Remember that a graph state described by a star graph is single-qubit Clifford equivalent to a GHZ-state.
Thus checking if $S_{V'}$ is a vertex-minor of $G$ is equivalent to checking if $\ket{G}$ can be transformed to GHZ-state on the qubits $V'$ by only using \LCLPMCC.
We will give this problem a separate name.

\begin{prm}[\SVM]\label{prob:star_vertex_minor}
Given a graph $G$ and a vertex subset $V'$ of $V(G)$, decide whether $S_{V'}$ is a vertex-minor of $G$.
\end{prm}
Note that we have not specified which star graph on $V'$ we use. This is not ambiguous since all star graphs on $V'$ are equivalent under local complementation. In the rest of the text we will often leave the choice of star graph open.

\subsection{Circle graphs}\label{sec:circle}

Here we introduce circle graphs and representations of these under the action of local complementations.
Circle graphs are graphs with edges represented as intersections of chords on a circle.
These graphs are also sometimes called alternance graphs since they can be described by a double occurrence word such that the edges of the graph are the given by the alternances induced by this word.
We will make use of the latter description here, which was introduced by Bouchet in~\cite{Bouchet1972} and also described in~\cite{Bouchet1994}.
This description is also related to yet another way to represent circle graphs, as Eulerian tours of 4-regular multi-graphs, introduced by Kotzig in~\cite{Kotzig1977}. 
For an overview and the history of circle graphs see for example the book by Golumbic~\cite{Golumbic2004}.

\subsubsection{Double occurrence words}
Let us first define double occurrence words and equivalence classes of these. This will allow us to define circle graphs.

\begin{mydef}[Double occurrence word]\label{def:dow}
    A double occurrence word $\bs{X}$ is a word with letters in some set $V$, such that each element in $V$ occur exactly twice in $\bs{X}$.
    Given a double occurrence word $\bs{X}$ we will write $V(\bs{X})=V$ for its set of letters.
\end{mydef}

\begin{mydef}[Equivalence class of double occurrence words]\label{def:d_X}
    We say that a double occurrence word $\bs{Y}$ is equivalent to another $\bs{X}$, i.e. $\bs{Y}\sim \bs{X}$, if $\bs{Y}$ is equal to $\bs{X}$, the mirror $\widetilde{\bs{X}} $ or any cyclic permutation of $\bs{X}$ or $\widetilde{\bs{X}}$.
    We denote by $\bs{d}_{\bs{X}}= \{\bs{Y} \::\:\bs{Y}\sim \bs{X}\}$ the equivalence class of $\bs{X}$, i.e. the set of words equivalent to $\bs{X}$.
\end{mydef}

Next we define alternances of these equivalence classes, which will represent the edges of an alternance graph.

\begin{mydef}[Alternance]\label{def:alternance}
    An \emph{alternance} $(u,v)$ of the equivalence class $\bs{d}_{\bs{X}}$ is a pair of distinct elements $u,v\in V$ such that a double occurrence word of the form $\dots u\dots v\dots u\dots v\dots$ is in $\bs{d}_{\bs{X}}$.
\end{mydef}

Note that if $(u,v)$ is an alternance of $\bs{d}_{\bs{X}}$ then so is $(v,u)$, since the mirror of any word in $\bs{d}_{\bs{X}}$ is also in $\bs{d}_{\bs{X}}$.

\begin{mydef}[Alternance graph]\label{def:alternance_graph}
    The alternance graph $\mathcal{A}(\bs{X})$ of a double occurrence word $\bs{X}$ is a graph with vertices $V(\bs{X})$ and edges given exactly by the alternances of $\bs{d}_{\bs{X}}$, i.e.
    \begin{equation}
        E(\mathcal{A}(\bs{X}))=\{(u,v)\in V(\bs{X})\times V(\bs{X})\,:\,(u,v)\text{ is an alternance of }\bs{d}_{\bs{X}}\}
    \end{equation}
\end{mydef}

Note that since $\mathcal{A}(\bs{X})$ only depends on the equivalence class of $\bs{X}$, the alternance graphs $\mathcal{A}(\bs{X})$ and $\mathcal{A}(\bs{Y})$ are equal if $\bs{X}\sim\bs{Y}$.
Now we can formally define\footnote{Circle graphs are usually defined as graphs which edges are represented by intersections of chords on a circle. However, the definition in terms of alternances of double occurrence words turn out to be equivalent~\cite{Bouchet1994}.} circle graphs.

\begin{mydef}[Circle graph]\label{def:circle_graph}
    A graph $G$ which is the alternance graph of some double occurrence word $\bs{X}$ is called a circle graph.
\end{mydef}

\subsubsection{Eulerian tours on 4-regular multi-graphs}
There is yet another way to represent circle graphs, closely related to double occurrence words, as Eulerian tours of $4$-regular multi-graphs.

\begin{mydef}[Eulerian tour]
    Let $F$ be a connected 4-regular multi-graph.
    An Eulerian tour $U$ on $F$ is a tour that visits each edge in $F$ exactly once.
\end{mydef}

Any $4$-regular multi-graph is Eulerian, i.e. has a Eulerian tour, since each vertex has even degree~\cite{biggs1976graph}.

Furthermore, any Eulerian tour on a $4$-regular multi-graph $F$ traverses each vertex exactly twice, except for the vertex which is both the start and the end of the tour.
Such a Eulerian tour induces therefore a double occurrence word, the letters of which are the vertices of $F$, and consequently a circle graph as described in the following definition.

\begin{mydef}[Induced double occurrence word]\label{def:eul_tour}
    Let $F$ be a connected $4$-regular multi-graph on $k$ vertices $V(F)$.
    Let $U$ be a Eulerian tour on $F$ of the form
    \begin{equation}\label{eq:eul_tour}
        U=x_1e_1x_2\dots x_{2k-1}e_{2k-1}x_{2k}e_{2k}x_1.
    \end{equation}
    with $x_i\in V$. Note that every element of $V$ occurs exactly twice in $U$.
    From a Eulerian tour $U$ as in \cref{eq:eul_tour} we define an induced double occurrence word as
    \begin{equation}
        m(U)=x_1x_2\dots x_{2k-1}x_{2k}.
    \end{equation}
    To denote the alternance graph given by the double occurrence word induced by a Eulerian tour, we will write $\mathcal{A}(U)\equiv\mathcal{A}(m(U))$.
\end{mydef}

Similarly to double occurrence words, we also introduce equivalence classes of Eulerian tours under cyclic permutation or reversal of the tour.

\begin{mydef}[Equivalence class of Eulerian tours]
    Let $F$ be a connected 4-regular multi-graph and $U$ be an Eulerian tour on $F$.
    We say that an Eulerian tour $U'$ on $F$ is equivalent to $U$, i.e. $U\sim U'$, if $U'$ is equal to $U$, the reversal $\widetilde{U}$ or any cyclic permutation of $U$ or $\widetilde{U}$.
    We denote by $\bs{t}_U$ the equivalence class of $U$, i.e. the set of Eulerian tours on $F$ which are equivalent to $U$.
\end{mydef}

It is clear that if the Eulerian tours $U$ and $U'$ on a 4-regular multi-graph $F$ are equivalent, then so are the double occurrence words $m(U)$ and $m(U')$.
Furthermore, as for double occurrence words, two equivalent Eulerian tours on a connected 4-regular multi-graph induce the same alternance graph.


\subsubsection{Local complementations on circle graphs}

We will now introduce an operation $\widetilde{\tau}_v$ on double occurrence words that will be the equivalent of performing a local complementation on the corresponding alternance graph.

\begin{mydef}[$\widetilde{\tau}_v$]
    Let $\bs{X}$ be a double occurrence word and $v$ be an element in $V(\bs{X})$.
    We can then always find sub-words $\bs{A}$, $\bs{B}$ and $\bs{C}$ not containing $v$, such that $\bs{X}=\bs{A}v\bs{B}v\bs{C}$.
    Note that some of the sub-words $\bs{A}$, $\bs{B}$ and $\bs{C}$ are possibly empty.
    The operation $\widetilde{\tau}_v$ acting on a double occurrence word is then defined as
    \begin{equation}\label{eq:LC_on_m}
        \widetilde{\tau}(\bs{A}v\bs{B}v\bs{C})=\bs{A}v\widetilde{\bs{B}}v\bs{C}.
    \end{equation}
    If $\mathbf{v}=(v_1,\dots,v_l)$ is a sequence of elements of $V(\bs{X})$ we use the notation $\widetilde{\tau}_\mathbf{v}(\bs{X})=\widetilde{\tau}_{v_l}\circ\dots\circ\widetilde{\tau}_{v_1}(\bs{X})$.
\end{mydef}

The operation $\widetilde{\tau}_v$ in the above definition maps equivalence classes to equivalence classes, as defined in \cref{def:d_X}.
That is, if $\bs{X}\sim \bs{Y}$ and $v\in V(\bs{X})$, then $\widetilde{\tau}_v(\bs{X})\sim\widetilde{\tau}_v(\bs{Y})$.
For example, assume that $\bs{Y}$ is the mirror of $\bs{X}$, i.e. $\bs{Y}=\widetilde{\bs{X}}$.
Then we see that
\begin{equation}
    \widetilde{\tau}_v(\bs{X})=\bs{A}v\widetilde{\bs{B}}v\bs{C}\sim\widetilde{\bs{A}v\widetilde{\bs{B}}v\bs{C}}=\widetilde{\bs{C}}v\bs{B}v\widetilde{\bs{A}}=\widetilde{\tau}_v(\bs{Y}).
\end{equation}
The case when $\bs{Y}$ is a obtained by a cyclic permutation of $\bs{X}$ can be checked similarly.

In~\cite{Bouchet1994} it was shown that the alternance graph of $\mathcal{A}(\widetilde{\tau}_v(\bs{X})$, where $\bs{X}$ is a double occurrence word and $v\in V(\bs{X})$, is the same as the graph obtained by performing a local complementation at $v$, i.e.
\begin{equation}\label{eq:tautau}
    \tau_v(\mathcal{A}(\bs{X}))=\mathcal{A}(\widetilde{\tau}_v(\bs{X})).
\end{equation}

Similar to the above we can also define an operation $\bar{\tau}_v$ on Eulerian tours $U$ on $4$-regular multi-graphs which also has the effect of a local complementation on the graph $\mathcal{A}(U)$.

\begin{mydef}[$\bar{\tau}_v$]\label{def:tau_on_eul}
    Let $F$ be a connected 4-regular multi-graph.
    Let $U$ be a Eulerian tour on $F$ and $v$ be a vertex in $V$.
    Let $P_v$ be the first subtrail of $U$ that starts and ends at $v$, i.e. $U=U_1P_vU_2$, from some $U_1$ and $U_2$.
    We define $\bar{\tau}_v(U)$ to be the Eulerian tour obtained by traversing $U_1$, the reversal of $P_v$ and then $U_2$, i.e. $\bar{\tau}_v(U)=U_1\widetilde{P_v}U_2$.
    When $\mathbf{v}=v_1\dots v_l$ is a sequence of vertices in $V$ we write $\bar{\tau}_\mathbf{v}(U)\equiv \bar{\tau}_{v_l}\circ\dots\circ\bar{\tau}_{v_1}(U)$.
\end{mydef}

Note in particular that $\bar{\tau}_v(U)$, where $U$ is an Eulerian tour on $F$, is also a Eulerian tour on $F$.

We have now defined $\tau$-operations on circle graphs, $\widetilde{\tau}$-operations on double occurrence words and $\bar{\tau}$-operations on Eulerian tours of $4$-regular multi-graphs.
They are given similar names since they are in some sense the same operation but in different representations of circle graphs.
To see this note that
\begin{equation}\label{eq:tautau2}
    m(\bar{\tau}_v(U))=m(U_1\widetilde{P_v}U_2)=\widetilde{\tau}_v(m(U))
\end{equation}
where $U=U_1P_vU_2$ as in \cref{def:tau_on_eul}.
From \cref{eq:tautau} and the shorthand $\mathcal{A}(U) = \mathcal{A}(m(U))$ we also have that
\begin{equation}\label{eq:tautau3}
    \mathcal{A}(\bar{\tau}_v(U))=\mathcal{A}(\widetilde{\tau}_v(m(U)))=\tau_v(\mathcal{A}(U)).
\end{equation}
The operation $\bar{\tau}_v$ on Eulerian tours of 4-regular multi-graphs was introduced by Kotzig in \cite{Kotzig1966}, where he called it a $\kappa$-transformation.

As stated by Bouchet in~\cite{Bouchet1994}, Kotzig~\cite{Kotzig1966} proved that any two Eulerian tours of a 4-regular multi-graph are related by a sequence of $\kappa$-transformations.

\begin{thm}[Proposition 4.1 in~\cite{Bouchet1994},~\cite{Kotzig1966}]\label{thm:Kotzig}
    Let $U$ and $U'$ be Eulerian tours on the same connected 4-regular multi-graph. Then there exist a sequence $\mathbf{v}$ such that $\tau_\mathbf{v}(U)\sim U'$.
\end{thm}

\subsubsection{Vertex-deletion on circle graphs}

When we are considering vertex-minors of circle graphs, it is useful to have an operation on the double occurrence word that corresponds to the deletion of a vertex in the corresponding alternance graph.
Let $\bs{X}=\bs{A}v\bs{B}v\bs{C}$ be a double occurrence word and $v$ be an element in $V(\bs{X})$.
We will denote by $\bs{X}\setminus v$ the deletion of the element $v$, i.e.
\begin{equation}\label{eq:VD_on_m}
    \bs{X}\setminus v\equiv (\bs{A}v\bs{B}v\bs{C})\setminus v=\bs{A}\bs{B}\bs{C}.
\end{equation}
The resulting word $\bs{ABC}$ is also a double occurrence word and furthermore we have that
\begin{equation}
    \mathcal{A}(\bs{X})\setminus v=\mathcal{A}(\bs{X}\setminus v).
\end{equation}
If $W=\{w_1,w_2\dots,w_l\}$ is a subset of $V$, we will write $\bs{X}\setminus W$ as the deletion of all elements in $W$, i.e.
\begin{equation}
    \bs{X}\setminus W=(\dots((\bs{X}\setminus w_1)\setminus w_2)\dots)\setminus w_l.
\end{equation}
Connected to this we can also define an induced double occurrence sub-word $m[W]=\bs{X}\setminus (V\setminus W)$.
The reason for calling this an induced double occurrence sub-word stems from its relation to induced subgraphs of the alternance graph as
\begin{equation}\label{eq:induced_alternance}
    \mathcal{A}(\bs{X})[W]=\mathcal{A}(\bs{X}[W]).
\end{equation}


\subsubsection{Vertex-minors of circle graphs}
Since we now have expressions for local complementation and vertex deletion on circle graphs in terms of double occurrence words, we can consider vertex-minors of circle graphs completely in terms of double occurrence words.
More precisely we have the following theorem.

\begin{thm}\label{thm:vm_of_dow}
    Let $G$ be a circle graph such that $G=\mathcal{A}(\bs{X})$ for some double occurrence word $\bs{X}$.
    Then $G'$ is a vertex-minor of $G$ if and only if there exist a sequence $\mathbf{v}$ of elements in $V(G)=V(\bs{X})$ such that
    \begin{equation}\label{eq:dow_vertex-minor}
        G'=\mathcal{A}(\widetilde{\tau}_\mathbf{v}(\bs{X})[V(G')]).
    \end{equation}
\end{thm}
\begin{proof}
    By using \cref{eq:induced_alternance} and \cref{eq:tautau3} on the right hand side of \cref{eq:dow_vertex-minor} we have that
    \begin{equation}
        \mathcal{A}\Big(\widetilde{\tau}_\mathbf{v}(\bs{X})[V(G')]\Big)=\mathcal{A}\Big(\widetilde{\tau}_\mathbf{v}(\bs{X})\Big)[V(G')]=\tau_\mathbf{v}(\mathcal{A}(\bs{X}))[V(G')]
    \end{equation}
    Since $G'$ is a vertex-minor of $G=\mathcal{A}(\bs{X})$ if and only if there exist a sequence $\mathbf{v}$ of elements in $V(G)$ such that
    \begin{equation}
        G'=\tau_\mathbf{v}(G)[V(G')]
    \end{equation}
    the theorem follows.
\end{proof}

We can also consider vertex minors of circle graphs in terms of their representations as Eulerian tours on connected $4$-regular multi-graphs, which we will use in \cref{sec:complexity} to prove that \VM\ is \NP-Complete.
\Cref{thm:Kotzig}, together with \cref{eq:tautau3}, implies that connected 4-regular multi-graphs describe equivalence classes of circle graphs under local complementations.
Bouchet pointed out this fact in~\cite{Bouchet1994}.
We formalize this here as a theorem together with a formal proof:

\begin{thm}\label{thm:equiv_eul}
    Let $U$ be an Eulerian tour of a connected $4$-regular multi-graph $F$ with vertices $V$.
    Then (1) any graph LC-equivalent to $\mathcal{A}(U)$ is an alternance graph of some Eulerian tour of $F$ and (2) any alternance graph of a Eulerian tour of $F$ is a graph LC-equivalent to $\mathcal{A}(U)$.
\end{thm}
\begin{proof}
    We start by proving (1), so let us therefore assume that $G$ is a graph LC-equivalent to $\mathcal{A}(U)$.
    This means, by definition, that there exist a sequence $\mathbf{v}$ of vertices in $V$ such that $G=\tau_\mathbf{v}(\mathcal{A}(U))$.
    By using \cref{eq:tautau3} we have that 
    \begin{equation}
        G=\mathcal{A}(\bar{\tau}_\mathbf{v}(U)).
    \end{equation}
    which shows that $G$ is an alternance graph induced by a Eulerian tour of $F$, since $\bar{\tau}_\mathbf{v}(U)$ is a Eulerian tour on $F$.
    To prove (2), assume that $U'$ is a Eulerian tour of $F$.
    We will now prove that the alternance graph of $U'$, $\mathcal{A}(U')$, is LC-equivalent to $\mathcal{A}(U)$.
    By \cref{thm:Kotzig}, we know that there exist a sequence of $\bar{\tau}_v$-transformations that relates $U$ and $U'$, i.e. there exist a sequence $\mathbf{v}$ such that
    \begin{equation}
        \bar{\tau}_\mathbf{v}(U)\sim U'.
    \end{equation}
    Since these Eulerian tours are equivalent, their induced alternance graphs are equal, i.e.
    \begin{equation}
        \mathcal{A}(\bar{\tau}_\mathbf{v}(U))=\mathcal{A}(U').
    \end{equation}
    Finally, using \cref{eq:tautau3} on the above equation gives
    \begin{equation}
        \tau_\mathbf{v}(\mathcal{A}(U))=\mathcal{A}(U')
    \end{equation}
    which shows that $\mathcal{A}(U)$ and $\mathcal{A}(U')$ are indeed LC-equivalent.
\end{proof}

Similarly to \cref{thm:vm_of_dow} we can decide if a circle graph has a certain vertex-minor by considering Eulerian tours of a 4-regular graph, which is captured in the following theorem.

\begin{thm}\label{thm:vm_of_eul}
    Let $F$ be a connected $4$-regular multi-graph and let $G$ be a circle graph such that $\mathcal{A}(U)$ for some Eulerian tour $U$ on $F$.
    Then $G'$ is a vertex-minor of $G$ if and only if there exist a Eulerian tour $U'$ on $F$ such that
    \begin{equation}\label{eq:vm_of_eul}
        G'=\mathcal{A}(m(U')[V(G')]).
    \end{equation}
\end{thm}
\begin{proof}
    Lets first assume that $G'$ is a vertex-minor of $G$.
    This means that there exists a sequence $\mathbf{v}$ such that $G'=\tau_\mathbf{v}(G)[V(G')]$.
    Since $G=\mathcal{A}(U)$ we have that
    \begin{align}
        G'&=\tau_\mathbf{v}\Big(\mathcal{A}(U)\Big)[V(G')]\\
          &=\mathcal{A}\Big(\bar{\tau}_v(U)\Big)[V(G')]\\
          &=\mathcal{A}\Big(m(\bar{\tau}_v(U))[V(G')]\Big)
    \end{align}
    where we used \cref{eq:tautau3} in the first step and \cref{eq:induced_alternance} in the second.
    We therefore see that $U'=\bar{\tau}_v(U)$ is a Eulerian tour on $F$ satisfying \cref{eq:vm_of_eul}.

    To prove the converse let us assume that there exist a Eulerian tour $U'$ on $F$ satisfying \cref{eq:vm_of_eul}.
    From \cref{thm:Kotzig} we know that there exist a sequence $\mathbf{v}$ such that $U'=\bar{\tau}_\mathbf{v}(U)$.
    We can then replace $U'$ by $\bar{\tau}_\mathbf{v}(U)$ in \cref{eq:vm_of_eul} such that
    \begin{align}
        G'&=\mathcal{A}\Big(m(\bar{\tau}_\mathbf{v}(U))[V(G')]\Big)\\
          &=\mathcal{A}\Big(m(\bar{\tau}_\mathbf{v}(U))\Big)[V(G')]\\
          &=\tau_\mathbf{v}\Big(\mathcal{A}(U)\Big)[V(G')]\label{eq:proof_vm_of_eul}
    \end{align}
    where we again made use of \cref{eq:induced_alternance} and \cref{eq:tautau3}.
    From \cref{eq:proof_vm_of_eul} we see that $G'$ is indeed a vertex-minor of $G$, see \cref{def:vertex-minor}, since $G=\mathcal{A}(U)$.
\end{proof}

\subsubsection{Semi-Ordered Eulerian tours}\label{sec:soet}

From the previous sections we have seen that circle graphs and their vertex-minors can be described by Eulerian tours on connected 4-regular multi-graphs.
As discussed in \cref{sec:graph_states}, the question of whether a given graph $G$ has vertex-minors on the subset $V'\subseteq V(G)$ in the form of star or complete graphs is of importance in quantum information theory, since it corresponds to transforming the graph state $\ket{G}$ to a GHZ-state on the qubits $V'$.
A natural question is therefore: Given a set of vertices $V'$, what property should a connected 4-regular multi-graph $F$ satisfy, such that $S_{V'}$ is a vertex-minor of $\mathcal{A}(U)$, for some Eulerian tour $U$ on $F$.
As we will see in this section, a necessary and sufficient condition is that $F$ allows for what we call a \emph{semi-ordered Eulerian tour} (SOET) with respect to $V'$.

The existence of a \SOET~on a four regular graph $F$ with respect to some vertex set $V'$ will therefore be a key technical tool when considering \SVM\ on circle graphs, as described in \cref{sec:complexity}. We formally define a SOET as follows.

\begin{mydef}[SOET]\label{def:SOET}
    Let $F$ be a 4-regular multi-graph and let $V'\subseteq V(F)$ be a subset of its vertices.
    Furthermore, let $\bs{s}=s_1s_2\dots s_k$ be a word with letters in $V'$ such that each element of $V'$ occurs exactly once in $\bs{s}$ and where $k=\abs{V'}$.
    A semi-ordered Eulerian tour $U$ with respect to $V'$ is a Eulerian tour such that $m(U)=\bs{X}_0s_1\bs{X}_1s_2\dots s_k\bs{X}_{k}s_1\bs{Y}_{1}s_2\dots s_k\bs{Y}_{k}$ for some $\bs{s}$ and where $\bs{X}_0,\bs{X}_1,\dots,\bs{X}_{k},\bs{Y}_1,\dots,\bs{Y}_{k}$ are words (possibly empty) with letters in $V\setminus V'$.
    This can also be stated as $m(U)[V']=\bs{s}\bs{s}$, for some $\bs{s}$.
\end{mydef}
Note that the multi-graph $F$ is not assumed to be simple, so multi-edges and self-loops are allowed.
A SOET is a Eulerian tour on $F$ that traverses the elements of $V'$ in some order once and then again in the same order.
The particular order in which the \SOET~traverses $V'$ will not be important here, only that it traverses $V'$ in the same order twice.
An example of a graph that allows for a SOET with respect to the set $V'=\{a,b,c,d\}$ can be seen in \cref{fig:SOET_yes}.
A \SOET\ for this graph is for example $m(U)=abcdaebced$.
The graph in \cref{fig:SOET_no} on the other hand does not allow for any SOET with respect to the set $V'=\{a,b,c,d\}$.

\begin{figure}[H]
    \centering
    \begin{subfigure}{0.15\textwidth}
        \includegraphics[width=1\textwidth]{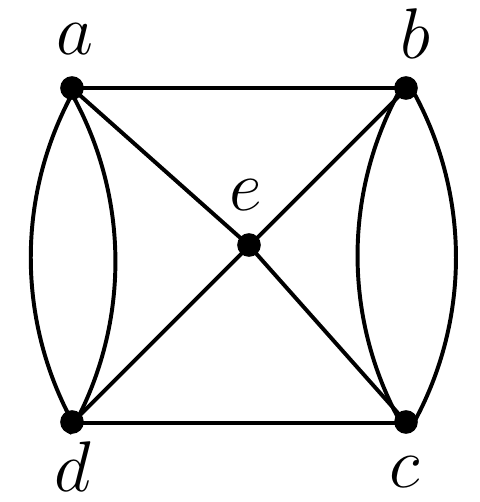}
        \caption{}
        \label{fig:SOET_yes}
    \end{subfigure}
    \hspace{2cm}
    \begin{subfigure}{0.18\textwidth}
        \includegraphics[width=1\textwidth]{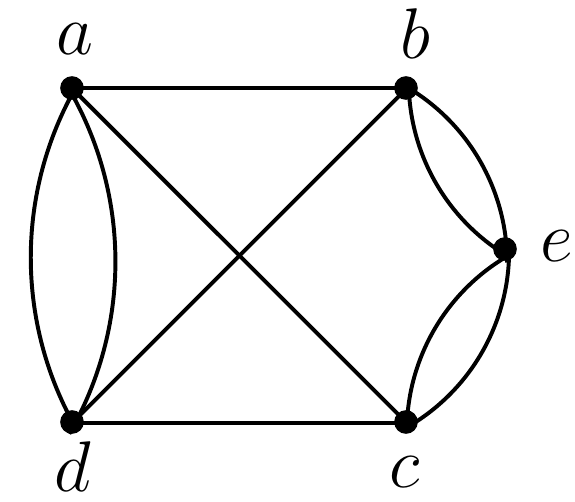}
        \caption{}
        \label{fig:SOET_no}
    \end{subfigure}
    \caption{Examples of two $4$-regular multi-graphs.
    The graph in \cref{fig:SOET_yes} allows for a SOET with respect to the set $\{a,b,c,d\}$ but the graph in \cref{fig:SOET_no} does not.}
    \label{fig:SOET_ex}
\end{figure}

We also formally define the \SOET-decision problem, which takes a $4$-regular multi-graph $F$ and a subset $V'$ of the vertices as input and asks to decide whether or not the graph $F$ allows for a Semi-Ordered Eulerian Tour with respect to the vertex set $V'$. 

\begin{prm}[SOET]\label{prob:SOET}
Let $F$ be a $4$-regular multi-graph and let $V'$ be a subset $V(F)$. Decide whether there exists a \SOET~$U$ on $F$ with respect to the set $V'$.
\end{prm}

As mentioned, the reason for introducing the notion of a \SOET\ is that a 4-regular multi-graph $F$ allows for a \SOET\ with respect to a subset $V'\subseteq V(F)$ if and only if a star graph on $V'$ is a vertex-minor of an alternance graph $\mathcal{A}(U)$ induced by a Eulerian tour $U$ on $F$.
This is captured in the following theorem, formulated as a corollary of \cref{thm:vm_of_eul}.

\begin{cor}\label{cor:reg_to_circle}
    Let $F$ be a connected 4-regular multi-graph and let $G$ be a circle graph given by the alternance graph of a Eulerian tour $U$ on $F$, i.e. $G=\mathcal{A}(U)$.
    Then $S_{V'}$ is a vertex-minor of $G$ if and only if $F$ allows for a \SOET~(see \cref{def:SOET}) with respect to $V'$.
\end{cor}
\begin{proof}
    Note first that $S_{V'}\leq G$ if and only if $K_{V'}\leq G$, since $S_{V'}$ and $K_{V'}$ are LC~equivalent.
    From \cref{thm:vm_of_eul} we know that $K_{V'}$ is a vertex-minor of $G$ if and only if there exist an Eulerian tour $U'$ on $F$ such that
    \begin{equation}
        K_{V'}=\mathcal{A}(m(U')[V']).
    \end{equation}
    It is easy to verify that $\mathcal{A}(\bs{X})$ is a complete graph on $V'$ if and only if $\bs{X}=s_1s_2\dots s_ks_1s_2\dots s_k$ where $\bs{s}=s_1s_2\dots s_k$ is a word with letters in $V'$ such that each element of $V'$ occur exactly once in $\bs{s}$.
    The result then follows, since $m(U')[V']$ is of this form if and only if $U'$ is a \SOET~with respect to $V'$.
\end{proof}

One can see that the existence of a \SOET~on a $4$-regular multi-graph $F$ with respect to $V'$, imparts an ordering on the subset of vertices $V'$. We will in particular be interested in vertices in $V'$ that are `consecutive' with respect to the \SOET. Consecutiveness is defined as follows.

\begin{mydef}[Consecutive vertices]\label{def:consecutive}
    Let $F$ be a $4$-regular graph and $U$ a~\SOET~on $F$ with respect to a subset $V'\subseteq V(F)$.
    Two vertices $u,v\in V'$ are called consecutive in $U$ if there exist a sub-word $u\bs{X}v$ or $v\bs{X}u$ of $m(U)$ such that no letter of $\bs{X}$ is in $V'$.
\end{mydef}

We also define the notion of a ``maximal sub-word" associated with two consecutive vertices.

\begin{mydef}[Maximal sub-words]\label{def:maximal}
    Let $F$ be a $4$-regular multi-graph and $U$ a SOET on $F$ with respect to a subset $V'\subseteq V(F)$.
    The double occurrence word induced by $U$ is then of the form $m(U)=\bs{X}_0s_1\bs{X}_1s_2\dots s_k\bs{X}_{k}s_1\bs{Y}_{1}s_2\dots s_k\bs{Y}_{k}$, where $k=\abs{V'}$, $s_1,\dots,s_k\in V'$ and $X_0,\dots,X_k,
    Y_1,\dots,Y_k$ are words (possibly empty) with letters in $V(F)\setminus V'$.
    For $i\in[k-1]$, we call $\bs{X}_i$ and $\bs{Y}_i$ the two maximal sub-words associated with the consecutive vertices $s_{i}$ and $s_{i+1}$.
    Furthermore, we call $\bs{X}_k$ and $\bs{Y}_k\bs{X}_0$ the two maximal sub-words associated with the consecutive vertices $s_{k}$ and $s_1$.
    Given two consecutive vertices $u$ and $v$, we will denote their two maximal sub-words as $\bs{X}$ and $\bs{X}'$, $\bs{Y}$ and $\bs{Y}'$ or similar.
\end{mydef}

\subsection{Leaves, twins and axils}\label{sec:reductions}
In this section we will consider certain vertices called leaves, twins and axils.
First we will prove that such vertices can in many cases be removed when considering the vertex-minor problem, which can simplify the problem significantly. 
We capture this in \cref{thm:reductions}.
This motivates us to consider distance-hereditary graphs, since it turns out that these are exactly the graphs that can be reached from a single-vertex graph by adding leaves or performing twin-splittings. We will leverage these properties in \cref{sec:DH_alg} to find an efficient algorithm for \SVM~when the input graph is distance hereditary.
We define and consider distance-hereditary graphs in \cref{sec:DH}.

Let us first formally define leaves, twins and axils.

\begin{mydef}[Leaves and axils]\label{def:leaf}
    A \emph{leaf} is vertex with degree one.
    An \emph{axil} is the unique neighbor of a leaf.
\end{mydef}

\begin{mydef}[Twin]\label{def:twin}
    A \emph{twin} is a vertex $v$ such that there exist a different vertex $u$ with the same neighborhood, i.e. $v$ is a twin if and only if
    \begin{equation}\label{eq:twin_def}
        \exists u\in V\setminus\{v\}:\big(N_v\setminus\{u\}=N_u\setminus\{v\}\big).
    \end{equation}
    A vertex $u$ as in \cref{eq:twin_def} is called a \emph{twin-partner} of $v$ and $v,u$ form a \emph{twin-pair}.
    If $v$ and $u$ are adjacent, they form a \emph{true twin-pair} and otherwise a \emph{false twin-pair}.
\end{mydef}


\begin{mydef}[Foliage]\label{def:T_set}
    The \emph{foliage} of a graph $G$ is the set of leaves, axils and twins in a graph $G$ and is denoted
    \begin{equation}
        T(G)=\{v\in V(G):\text{$v$ is a leaf, axil or twin}\}
    \end{equation}
\end{mydef}

We are now ready to prove the following theorem which can be used to simplify some instances of \VM, in particular when considering distance-hereditary graphs, see \cref{sec:DH}.

\begin{thm}\label{thm:reductions}
    Let $G$ and $G'$ be graphs and $v$ be a vertex in $G$ but not in $G'$.
    Then the following is true:
    \begin{itemize}
        \item If $v$ is a leaf or a twin, then $G'$ is a vertex-minor of $G$ if and only if $G'$ is a vertex-minor of $G\setminus v$, i.e.
            \begin{equation}
                G'<G\quad\Leftrightarrow\quad G'<(G\setminus v).
            \end{equation}
        \item If $v$ is an axil, then $G'$ is a vertex-minor of $G$ if and only if $G'$ is a vertex-minor of $\tau_w\circ\tau_v(G)\setminus v$, where $w$ is the leaf associated to $v$, i.e.
            \begin{equation}
                G'<G\quad\Leftrightarrow\quad G'<(\tau_w\circ\tau_v(G)\setminus v).
            \end{equation}
    \end{itemize}
\end{thm}
\begin{proof}
    Firstly, if $G'$ is a vertex-minor of $G\setminus v$, then clearly $G'$ is also a vertex-minor of $G$.\\

    This means we only need to prove the other direction. Assume therefore that $G'$ is a vertex-minor of $G$.
    We start by proving the case where $v$ is a leaf in $G$. The cases where $v$ is an axil or a twin in $G$ then follow by a short argument. \\

    Hence assume that $v$ is a leaf in $V\setminus V'$, where $V=V(G)$ and $V'=V(G')$.
    Furthermore, let $\bs{u}$ be a sequence of vertices such that each element of $V\setminus V'$ occurs exactly once in $\bs{u}$.
    Since $G'$ is a vertex-minor of $G$, we know by \cref{thm:multi_vertex-minor} that there exists some sequence of operations $P\in\mathcal{P}_{\bs{u}}$, such that $P(G)\sim_\mathrm{LC}G'$.
    Let us denote the $i$-th operation in $P$ as $P^{(i)}$, such that $P=P^{(n-k)}\circ\dots\circ P^{(1)}$, where $n=\abs{G}$ and $k=\abs{G'}$.
    Remember that each operation $P^{(i)}$ deletes the $i$-th vertex of $\bs{u}$ from the graph.
    Furthermore, let's denote the sequence of operations from $i$ through $j$ in $P$ as
    \begin{equation}
        P_i^j=P^{(j)}\circ\dots P^{(i+1)}\circ P^{(i)}.
    \end{equation}
    By \cref{cor:order_independent} we know that such a $P$ exist for all orderings $\bs{u}$ of the vertices in $V\setminus V'$.
    Without loss of generality we can assume that $v$ is the first element in $\bs{u}$. This means that $P^{(1)}$ is either $Z_v$, $Y_v$ or $X_v$. We will now treat all three these cases separately.\\

    If $P^{(1)}$ is $Z_v$ or $Y_v$, then since $v$ is a leaf we have that
    \begin{equation}
        P^{(1)}(G)=G\setminus v.
    \end{equation}
    Then it is easy to see that $G'$ is also a vertex-minor of $G\setminus v$, since
    \begin{equation}
        G'\sim_\mathrm{LC}P(G)=P_2^{n-k}\circ P^{(1)}(G)=P_2^{n-k}(G\setminus v)
    \end{equation}
    \indent If $P^{(1)}$ is $X_v$ then the axil of $v$ cannot be in $V\setminus V'$, since the operation $X_v$ on a leaf disconnects the axil from its neighbors.
    Lets denote the axil of $v$ by $w$ and assume again w.l.o.g. that the ordering of $V\setminus V'$ is such that $w$ is the second element of $\bs{u}$.
    Since $w$ is a disconnected vertex after $P^{(1)}$, any of the three operations $\{X_w,Y_w,Z_w\}$ act the same, i.e. deleting $w$.
    So the action of $X_v$ followed by $P^{(2)}\in\{X_w,Y_w,Z_w\}$ is the same as deleting both $v$ and $w$ or in other words
    \begin{equation}
        P_1^2(G)=Z_w(G\setminus v)
    \end{equation}
    It is again clear that $G'$ is then a vertex-minor of $G\setminus v$, since
    \begin{equation}
        G\sim_\mathrm{LC}P(G)=P_3^{n-k}\circ P_1^2(G)=P_3^{n-k}\circ Z_w(G\setminus v)
    \end{equation}
    with a satisfying sequence taking $G\setminus v$ to an LC-equivalent graph of $G'$ being $(Z_w,P^{(3)},\dots,P^{(n-k)})$. This proves the theorem when $v$ is a leaf.\\
    Now assume that $v$ is a twin in $G$.
    To prove that the theorem also hold for twins, we first show that a twin can always be transformed into a leaf by local complementations.
    Assume that $v$ and $w$ are false twins, and denote one of their common neighbors as $n$.\footnote{Note that twins always have at least one common neighbor, except for the graph $K_2$ where the twins are anyway also leaves.}
    Then the graph $\tilde{G}=\tau_w\circ\tau_n(G)$ is a graph where $v$ is a leaf and $w$ is an axil.
    Since LC-equivalent graphs have the same vertex-minors, $G'$ is also a vertex-minor of $\tilde{G}$.
    From what we showed above and that $v$ is a leaf, $G'$ is also a vertex-minor of $\tilde{G}\setminus v$.
    Finally, $G'$ is then also a vertex-minor of
    \begin{equation}
        \tau_n\circ\tau_w(\tilde{G}\setminus v)=\tau_n\circ\tau_w\circ\tau_w\circ\tau_n(G)\setminus v=G\setminus v
    \end{equation}
    where we used \cref{lem:VD_last}.
    An almost identical argument can be made for the case where $v$ and $w$ are true twins by considering the graph $\tilde{G}=\tau_w(G)$.\\

    Now assume $v$ is an axil in $G$.
    If $v$ is an axil in $G$ and $v\notin G'$, then $v$ is a leaf in the graph $\tilde{G}=\tau_w\circ\tau_v(G)$, where $w$ is the leaf of $v$ in $G$.
    Since by assumption $G'<G$, we know that $G'<\tilde{G}$ and from the cases of leaves we have that also $G'<\tilde{G}\setminus v$, since $v$ is a leaf in $\tilde{G}$.
    This completes the proof.
\end{proof}

\subsubsection{Distance-hereditary graphs}\label{sec:DH}
In this section we introduce distance-hereditary graphs.
As shown by Bouchet in~\cite{Bouchet1988}, distance-hereditary graphs are exactly the graphs with rank-width one.
These graphs have nice properties which we make use of in \cref{sec:DH_alg}.

\begin{mydef}[Distance-hereditary]
    A graph $G$ is distance-hereditary if and only if, for each connected induced subgraph $G[A]$ and for any two vertices $u,v\in A$ the distance between $u$ and $v$ is the same in $G$ and in $G[A]$, i.e.
    \begin{equation}
        d_G(u,v)=d_{G[A]}(u,v).
    \end{equation}
\end{mydef}

The simplest example of a graph which is not distance-hereditary is the five-cycle $C_5$.
To see this pick two vertices which have distance two in $C_5$ and denote their unique common neighbor by $v$.
The distance between the same vertices in the connected induced subgraph $C_5[V\setminus v]$ is three and thus not the same as in $C_5$.
It turns out that distance-hereditary graphs are exactly the graphs which do not contain a vertex-minor isomorphic to $C_5$~\cite{Bouchet1994}.
We also note that distance-hereditary graphs form a strict subclass of circle graphs~\cite{Bouchet1994}.



In~\cite{Mulder1986} an equivalent property of distance-hereditary is shown: A graph is distance-hereditary if and only if it can be obtained from a single-vertex graph using the following three operations:
\begin{itemize}
    \item \emph{Add a leaf}: Let $u$ be a vertex in a graph $G$. Add the vertex $v$ and the edge $(u,v)$ to $G$.
    \item \emph{False twin-split}: Let $u$ be a vertex in a graph $G$. Add the vertex $v$ and the edges $\{(v,x):x\in N_u\}$.
    \item \emph{True twin-split}: Let $u$ be a vertex in a graph $G$. Add the vertex $v$ and the edges $\{(v,x):x\in \{u\}\cup N_u\}$.
\end{itemize}

Note that this implies that a distance-hereditary graph always has at least one leaf or twin, i.e. the foliage is non-empty. This fact will be a critical element of the algorithm presented in \cref{sec:DH_alg}\\

In the rest of this section we prove some properties of the foliage for distance-hereditary graphs, which we make use of in \cref{sec:DH_alg} to find an efficient algorithm for \VM\ on distance-hereditary graphs.
First we show that the twin relation is in fact transitive. This is a technical lemma we will use in later theorems.

\begin{lem}\label{lem:trans}
    Let $G$ be a graph and let $u$ be a vertex of $G$ that is a twin and has twin-partners $\{t_1,t_2,\dots,t_k\}.$
    Then all vertices in $u\cup\{t_1,t_2,\dots,t_k\}$ are pairwise twins.
\end{lem}
\begin{proof}
    Since $u$ and $t_i$ form a twin-pair, for $i\in\{1,2,\dots,k\}$, we have that
    \begin{equation}
        N_u\setminus\{t_i\}=N_{t_i}\setminus\{u\}
    \end{equation}
    which implies that
    \begin{equation}
        N_{t_i}=(N_u\setminus\{t_i\})\cup\{u\}.
    \end{equation}
    Thus, we have that
    \begin{align}
        N_{t_i}\setminus\{t_j\}&=((N_u\setminus\{t_i\})\cup\{u\})\setminus\{t_j\}\\
                               &=(\underbrace{(N_u\setminus\{t_j\})}_{N_{t_j}\setminus\{u\}}\cup\{u\})\setminus\{t_i\}\\
                               &=N_{t_j}\setminus\{t_i\}.
    \end{align}
    This shows that, for $i\neq j$, $t_i$ and $t_j$ form a twin-pair.
\end{proof}

Next we prove that adding leaves to a graph $G$ or performing (true or false) twin-splits never decreases the size of the foliage $T(G)$.

\begin{lem}\label{lem:T_mono}
    Assume $G$ is a connected, distance-hereditary graph.
    Let $G'$ be a graph formed by doing a twin-split on $G$ or adding a leaf to $G$.
    Then
    \begin{equation}
        \abs{T(G')}\geq\abs{T(G)},
    \end{equation}
    where $T(G)$ is the foliage of $G$.
\end{lem}
\begin{proof}
    To prove this, let us first consider the case when $\abs{G}\leq 2$.
    Since $G$ is connected it is necessary the case that $G=K_1$ or $G=K_2$:
    \begin{itemize}
        \item If $G=K_1$, then $G'=K_2$ and $\abs{T(G')}=2\geq0=\abs{T(G)}$.
        \item If $G=K_2$, then $G'=K_3$ or $G'=P_3$ and $\abs{T(G')}=3\geq2=\abs{T(G)}$.
    \end{itemize}
    Let's now consider the case when $\abs{G}>2$.
    We consider the two cases when $G'$ is formed by adding a leaf and performing a twin-split separately:
    \begin{itemize}
        \item Assume $G'$ is formed by adding a leaf $v$ to $G$, making $u$ an axil of $G'$.
            Note first that if $u\notin T(G)$, then $\abs{T(G)}$ can only increase since no vertex in $T(G)$ was affected.
            Let's therefore assume that $u\in T(G)$.
            There are then three possibilities: (1) $u$ is a leaf, (2) $u$ is an axil but not a twin and (3) $u$ is a twin.
            We consider these three cases separately:
            \begin{itemize}
                \item (1) Assume $u$ is a leaf in $G$.
                    Then the axil of $u$ in $G$, is not in $T(G')$, but both $u$ and $v$ are.
                    Therefore $\abs{T(G)}=\abs{T(G')}$.
                \item (2) Assume $u$ is an axil but not a twin in $G$.
                    Then $u$ is also an axil in $G'$ and we have that $\abs{T(G')}=\abs{T(G)}=1$.
                \item (3) Assume $u$ is a twin in $G$.
                    \begin{itemize}
                        \item Assume there is only one twin-pair containing $u$ in $G$.
                            Then the twin-partner of $u$ in $G$, is not in $T(G')$, but both $u$ and $v$ are.
                            Therefore $\abs{T(G')}=\abs{T(G)}$.
                        \item Assume there is more than one twin-pair containing $u$ in $G$.
                            Then the twin-partners of $u$ are all pairwise twins, by \cref{lem:trans}, and will still be in $G'$.
                            Therefore $\abs{T(G')}=\abs{T(G)}+1$.
                    \end{itemize}
            \end{itemize}
        \item Assume $G'$ is formed by twin-splitting $u$ in $G$, creating $v$ and making $v$ and $u$ a twin-pair.
            Note first that if $u\notin T(G)$, then $\abs{T(G)}$ can only increase since no vertex in $T(G)$ was affected.
            Let's therefore assume that $u\in T(G)$.
            There are then three possibilities: (1) $u$ is a leaf, (2) $u$ is an axil but not a twin and (3) $u$ is a twin.
            We consider these three cases separately:
            \begin{itemize}
                \item (1) Assume $u$ is a leaf in $G$.
                    Then the axil of $u$ in $G$ is either still an axil in $G'$ or not, depending on if $u$ and $v$ are true or false twins.
                    In either case, $\abs{T(G')}\geq\abs{T(G)}$ since $v\in T(G')$.
                \item (2) Assume $u$ is an axil but not a twin in $G$.
                    Note that all leaves with $u$ as an axil in $G$ are also twins.
                    These vertices are also twins in $G'$ since they are all now also adjacent to $v$.
                    Thus, $\abs{T(G')}=\abs{T(G)}+1$.
                \item (3) Assume $u$ is a twin in $G$.
                    \begin{itemize}
                        \item Assume there is only one twin-pair in $G$ containing $u$.
                            Then the twin-partner of $u$ in $G$, may or may not still be a twin-partner of $u$ in $G'$ depending on whether the considered twin-pairs are true or false.
                            Therefore the size of $T(G)$ either remains the same or increases by one since again $v\in T(G')$.
                        \item Assume there are more than one twin-pair in $G$ containing $u$.
                            Then the twin-partners of $u$ are all pairwise twins, by \cref{lem:trans}, and will still be in $G'$.
                            Therefore $\abs{T(G')}=\abs{T(G)}+1$.
                    \end{itemize}
            \end{itemize}
    \end{itemize}
\end{proof}

We now make use of the above theorem to prove that the foliage has a certain minimum size.

\begin{thm}\label{thm:T_size}
    Assume $G$ is a connected, distance-hereditary graph and $2\leq k\leq4$, then
    \begin{equation}
        \abs{G}\geq k\qq{}\Rightarrow\qq{}\abs{T(G)}\geq k.
    \end{equation}
\end{thm}
\begin{proof}
    First we explicitly check that the graphs on 2, 3 and 4 vertices has $T(G)=2$, $T(G)=3$ and $T(G)=4$, respectively.\footnote{The number of non-isomorphic graphs on 2, 3 and 4 vertices are 1, 2 and 6, respectively.}
    Then by \cref{lem:T_mono} and the fact that all distance-hereditary graphs can be built up by twin-splits and adding leaves~\cite{Bouchet1988}, the result follows.
\end{proof}

We point out that the theorem does not hold for $k>4$. Consider for example a path graph $P_k$ on more than four vertices. It is easy to see that size of the foliage in this case is $\abs{P_k}=4$.

Finally we show that an interesting property regarding the foliage, in relation to cut-vertices.\footnote{A cut-vertex is a vertex such that when it is deleted, the number of connected components increases.}

\begin{cor}\label{cor:CC}
    Assume that $G$ is a connected distance-hereditary graph and that $v\in G$ is a cut-vertex.
    Denote the connected components of $G\setminus v$ by $G_1,G_2,\dots,G_k$, where $k$ is the number of connected components of $G\setminus v$.
    Then for any $1\leq i\leq k$, there exist a vertex $u\in G_i$ such that $u\in T(G)$.
\end{cor}
\begin{proof}
    Pick an arbitrary connected component $G_i$ with vertices $V_i$.
    If $G_i$ is just a single vertex, then this vertex is necessarily a leaf in $G$ and is therefore in $T(G)$.
    Now assume that $\abs{G_i}>1$, then by using \cref{thm:T_size}, we have that there exist at least one twin-pair not containing $v$ or a leaf which is not $v$ in $G[\{v\}\cup V_i]$.
    This proves the corollary.
\end{proof}

\section{Complexity}\label{sec:complexity}
In this section we will consider the time-complexity of \VM. In particular we will show that even a highly restrictive version of the vertex-minor problem is \NP-Hard, namely when $G'$ is a star graph and $G$ is in a strict subclass of circle graphs. Since we also prove that \VM~is in \NP~this then proves that \VM~is \NP-Complete.\\
\subsection{\VM ~is in \NP}
We begin by arguing that the vertex-minor problem is in $\NP$.
Given graphs $G$ and $G'$ such that $G'<G$, a witness to this relation would be a sequence of local complementations and vertex deletions that takes $G$ to $G'$.
It is not a priori clear that this sequence is polynomial in length w.r.t. to the number of vertices of $G$.
However from \cref{thm:multi_vertex-minor} one can argue that whenever there is such a sequence, there is also a sequence of polynomial length. This leads to the following theorem.\\
\begin{thm}\label{thm:vm_in_np}
The decision problem $\mathrm{VERTEXMINOR}$ is in $\NP$.
\end{thm}
\begin{proof}
    Let $G $ and $G'$ be graphs, on $n$ and $k$ vertices respectively.
    Furthermore, let $\bs{u}$ be a sequence such that each element of $V(G)\setminus V(G')$ occur exactly once in $\bs{u}$.
    If $G' <G$ then, by \cref{thm:multi_vertex-minor}, there exists a sequence of operations $P\in\mathcal{P}_{\bs{u}}$, as specified in \cref{eq:P_u}, such that $P(G)\sim_\mathrm{LC}G'$. Furthermore, the sequence of operations $P$ consists of $\mathcal{O}(n-k)$ local complementations and vertex-deletions.
    A witness to the instance $(G,G')$ will then be the sequence of operations $P$. On the other hand, if $G'\nless G$, then by \cref{thm:multi_vertex-minor}, there exist no $P\in\mathcal{P}_{\bs{u}}$ such that $P(G)\sim_\mathrm{LC}G'$.\\

    Given $(G,G')$ and a sequence of operations $P$, a verifier can therefore perform the following protocol to check if $(G, G')$ is a \emph{yes}-instance of \SVM.
    \begin{enumerate}
        \item Compute $P(G)$.
        \item Decide if $P(G)\sim_\mathrm{LC}G'$ using Bouchet's algorithm for checking if two graphs are LC-equivalent~\cite{Bouchet1991}.
        \item Output \emph{yes} if Bouchet's algorithm outputs TRUE and \emph{no} otherwise.
    \end{enumerate}
    The verifier will therefore output \emph{yes} if $P$ is such that $P(G)\sim_\mathrm{LC}G'$ and \emph{no} if $G'\nless G$, since then $P(G)\nsim_\mathrm{LC}G'$ for any $P$.
    Computing $P(G)$ can be done in time $\mathcal{O}(n^2(n-k))$, since each local complementation can be performed in time $\mathcal{O}(n^2)$~\cite{Bouchet1991}.
    Furthermore, checking whether $P(G)$ and $G'$ are LC-equivalent can be done in time $\mathcal{O}(k^4)$ using Bouchet's algorithm~\cite{Bouchet1991}.
    Thus the verifier will output \emph{yes} or \emph{no} in time $\mathcal{O}(n^2(n-k))+\mathcal{O}(k^4)$.
\end{proof}

\subsection{\VM~is \NP-Complete}\label{sec:vmnpcomplete}

Next we will argue that the problem \VM~is also \NP-Hard and hence that it is \NP-Complete.
We will do this through a sequence of three reductions.
\begin{itemize}
    \item First we will reduce \SVM~to \VM. This is done in \cref{VM_NP_comp}.
    \item Secondly we will reduce a new problem, which we call the $\mathrm{SOET}$ problem, for Semi-Ordered Eulerian Tour, to \SVM. This is done in \cref{ssec:soet_to_svm}
    \item Finally we will reduce the problem of deciding whether a $3$-regular (or cubic) graph has a \emph{Hamiltonian cycle} or not (\CubHam), which is a known $\NP$-complete problem~\cite{Garey1976}, to the \SOET~problem.
    This is the most complicated part of the reduction and is done in several steps in \cref{ssec:cubham_to_soet}.
\end{itemize}
The reductions between these problems are summarized in \cref{fig:reductions}.
Eventually we will have the following theorem, which can be considered the main theorem of this section.\\

\begin{figure}[H]
    \centering
    \includegraphics[width=0.7\textwidth]{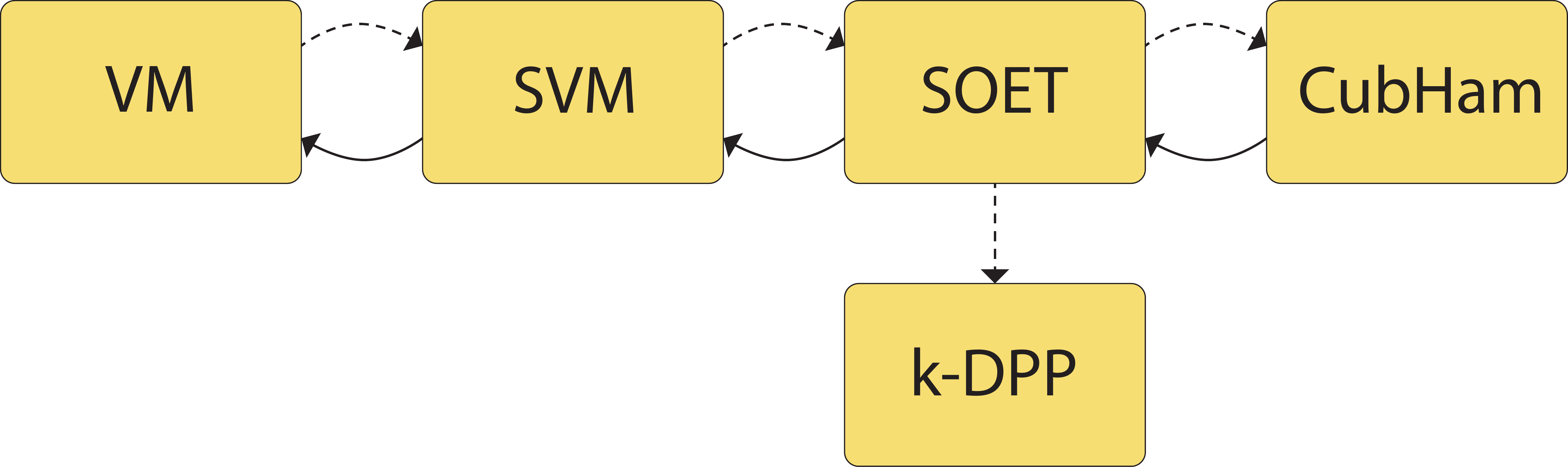}
    \caption{The different decision problems considered and how they can be reduced to each other.
    An solid arrow between problem $A$ and $B$ means that any instance of $A$ can be reduced to $B$ in polynomial time.
    A dashed line means that a subclass of instances of the $A$ can be reduced to $B$ in polynomial time.
    \VM\ (VM) and \SVM\ (SVM)  are presented in \cref{sec:vertex-minors}, \SOET\ in \cref{sec:soet}, \CubHam\ in \cref{sec:vmnpcomplete} and \DPP{k} ($k$-Disjoint Paths Problem) in \cref{sec:FPT}.
    The dashed line from \SVM\ to \SOET\ is restricted to circle graphs.
    The dashed line from \SOET\ to \CubHam\ is restricted to triangle-expanded graphs $\Lambda(R)$ where the \SOET\ is with respect to $V(R)$.
    Finally the dashed line from \SOET\ to \DPP{k} is for \SOET\ problems where the \SOET\ is with respect to a set of size $k$.}
    \label{fig:reductions}
\end{figure}

\begin{thm}\label{VM_NP_comp}
\VM~is \NP-Complete
\end{thm}
\begin{proof}
Note first that \SVM~trivially reduces to \VM.
This is so since every \emph{yes}(\emph{no})-instance of \SVM~is also an \emph{yes}(\emph{no})instance of \VM.
From \cref{thm:red_SOET_SVM} we see that we can reduce the \SOET~problem to \SVM~and finally from \cref{cor:SOETNP} we see that we can reduce~\CubHam~to the \SOET~problem.
Since \CubHam~is a known \NP-Complete problem this implies that \VM~is \NP-Hard.
From \cref{thm:vm_in_np} we have that \VM ~is in \NP~and hence it is \NP-Complete.
\end{proof}

Now we will detail every step in the above reduction. We begin with proving that the \SOET\ decision problem reduces to \SVM.

\subsubsection{Reducing the \SOET~problem to \SVM}\label{ssec:soet_to_svm}
In this section we show that the \SOET\ problem reduces to \SVM. For this we will make use of the properties of circle graphs, discussed in \cref{sec:prel}.
In \cref{cor:reg_to_circle} we showed that a 4-regular multi-graph $F$ allows for a \SOET\ with respect to a subset of its vertices $V'\subseteq V(F)$ if and only if an alternance graph $\mathcal{A}(U)$ (which is a circle graph), induced by some Eulerian tour on $F$, has $S_{V'}$ as a vertex-minor.

Since circle graphs are a subset of all simple graphs we can then decide whether a $4$-regular graph $F$ allows for a \SOET~with respect to some subset $V'$ of it's vertices by constructing the circle graph induced by an Eulerian tour on $F$ and checking whether it has a star-vertex minor on the vertex set $V'$. This leads to the following theorem.

\begin{thm}\label{thm:red_SOET_SVM}
The decision problem \SOET~reduces to \SVM.
\end{thm}
\begin{proof}
    Let $(F,V')$ be an instance of \SOET, where $F$ is a $4$-regular multi-graph and $V'$ a subset of the vertex set of $F$. Also let $G$ be a circle graph induced by any Eulerian tour $U$ on $F$. From \cref{cor:reg_to_circle} we see that $G$ has $S_{V'}$ as a vertex-minor if and only if $F$ allows for a \SOET~with respect to the vertex set $V'$. Since an Eulerian tour $U$ can be found in polynomial time~\cite{Fleury1883} and since $G$ can be efficiently constructed given $U$, this concludes the reduction.
\end{proof}

\subsubsection{Reducing \CubHam~to the \SOET~problem}\label{ssec:cubham_to_soet}

In this section we will prove that the \SOET~problem, as defined in~\cref{prob:SOET}, is $\mathbb{NP}$-Complete by reducing the problem of deciding if a $3$-regular graph is Hamiltonian (\CubHam), a well-known \NP-Complete problem~\cite{Garey1976}, to the \SOET~problem (it is in \NP~by \cref{thm:red_SOET_SVM} and \cref{thm:vm_in_np}).
For completeness we include the definition of a Hamiltonian graph.

\begin{mydef}[Hamiltonian]\label{def:hamtour}
    A graph is said to be Hamiltonian if it contains a Hamiltonian cycle.
    A Hamiltonian cycle is a cycle that visits each vertex in the graph exactly once.
\end{mydef}

We can use this to formally define the $\CubHam$ problem.

\begin{prm}[\CubHam]\label{prob:cubham}
    Let $R$ be a $3$-regular graph. Decide whether $R$ is Hamiltonian.
\end{prm}

The reduction of \CubHam~to the \SOET~problem is done by going though the following steps.
\begin{enumerate}
    \item Introduce the notion of a ($4$-regular) triangular-expansion $\Lambda(R)$ of a $3$-regular graph. This is done in \cref{def:triangular}. 
    \item Argue that given a $3$-regular graph $R$, its triangular-expansion can be constructed efficiently. This is done in \cref{lem:triang_eff}.
    \item Introduce the notions of \emph{skip} and \emph{true skip} that capture an essential behavior of \SOET s on triangular-expansions of $3$-regular graphs. This is done in \cref{sec:skips}.
    \item Prove that if a $3$-regular graph $R$ is Hamiltonian then the triangular-expansion $\Lambda(R)$ of $R$ allows for a \SOET~ with respect to the set $V(R)$. This is done in \cref{lem:HAM2SOET}.
    \item Prove that if the triangular-expansion $\Lambda(R)$ of a $3$-regular graph $R$ allows for a special kind of \SOET, called a \HAMSOET~with respect to $R$ (HAMSOETs are defined in \cref{def:HAMSOET}, but can be thought of as SOETs with no true skips), then the $3$-regular graph $R$ is Hamiltonian. This is done in \cref{lem:HAMSOET2HAM}.
    \item Prove that if the triangular-expansion $\Lambda(R)$ of a $3$-regular graph $R$ allows for a \SOET\ with respect to $V(R)$ then it also allows for a \HAMSOET\ with respect to $R$. This is done in \cref{lem:SOET2HAMSOET}.
\end{enumerate} 

Performing all these steps will lead to the following theorems.

\begin{thm}\label{thm:SOET_HAM}
    Let $R$ be a 3-regular graph and $\Lambda(R)$ be its triangular-expansion as defined in \cref{def:triangular}.
    $R$ is Hamiltonian if and only if $\Lambda(R)$ allows for a~\SOET~with respect to $V(R)$.
\end{thm}
\begin{proof}
    Let $R$ be a $3$-regular graph and let $\Lambda(R)$ be its triangular-expansion as defined in \cref{def:triangular}. If $R$ is Hamiltonian then \cref{lem:HAM2SOET} guarantees that $\Lambda(R)$ allows for a SOET with respect to the vertices $V(R)$. In the other direction, if $\Lambda(R)$ allows for a SOET with respect to the vertices $V(R)$ then we can see from \cref{lem:HAMSOET2HAM} that it also allows for a HAMSOET. The existence of a HAMSOET on $\Lambda(R)$ then implies, via \cref{lem:HAMSOET2HAM} that $R$ has a Hamiltonian cycle and hence that it is Hamiltonian. This proves the theorem.
\end{proof}

\begin{cor}\label{cor:SOETNP}
   The SOET problem is \NP-Complete. 
\end{cor}
\begin{proof}
    The Hamiltonian cycle problem (\CubHam) is \NP-Complete on 3-regular graphs~\cite{Garey1976}. We will reduce this problem to the \SOET~problem. Let $R$ be an instance of \CubHam, i.e. a $3$-regular graph. From this $3$-regular graph we can construct its triangular-expansion $\Lambda(R)$.  In \cref{lem:triang_eff} it is argued that this construction can be performed in $O(|V(R)|)$ time. We can then use \cref{thm:SOET_HAM} to see that $R$ is Hamiltonian if and only if $\Lambda(R)$ allows for a SOET with respect to the vertex set $V(R)$. Hence there exists an efficient reduction of \CubHam~to the \SOET~problem. This means that the \SOET~problem is \NP-Hard.
    Furthermore, the \SOET\ problem is in \NP\ since it can be efficiently reduced to \SVM\ by \cref{thm:red_SOET_SVM}, which is in \NP\ by \cref{thm:vm_in_np}. Hence the \SOET~problem is \NP-Complete.
\end{proof}

\begin{cor}
    The \SOET~problem is \NP-Complete on graphs which are triangular-expansions of planar 3-regular triply-connected graphs, i.e. graphs in the set
    \begin{equation}
        \{\Lambda(R):R\text{ is planar, 3-regular and triply-connected}\}.
    \end{equation}
\end{cor}
\begin{proof}
    The proof is the same as the proof of \cref{cor:SOETNP} but using the fact that \CubHam~is \NP-Complete on planar triply-connected graphs.
\end{proof}

\subsubsection{triangular-expansions}\label{sec:triangular_expansions}
It now remains to prove \cref{lem:HAM2SOET,lem:HAMSOET2HAM,lem:SOET2HAMSOET}. These lemmas will relate Hamiltonian cycles on 3-regular graphs and \SOET s on 4-regular multi-graphs by using a mapping from $3$-regular graphs to $4$-regular multi-graphs. We call this mapping `triangular-expansion'. We have the following definition.\\

\begin{mydef}[Triangular-expansion]\label{def:triangular}
    Let $R$ be a $3$-regular graph.
    A triangular-expansion $\Lambda(R)$ of a $3$-regular graph $R$ is constructed from $R$ by performing the following two steps:
    \begin{enumerate}
        \item Replace each vertex $v$ in $R$ with the subgraph below
            \begin{equation}\label{eq:triangle}
               \raisebox{-0.125\textwidth}{\includegraphics[width=0.3\textwidth]{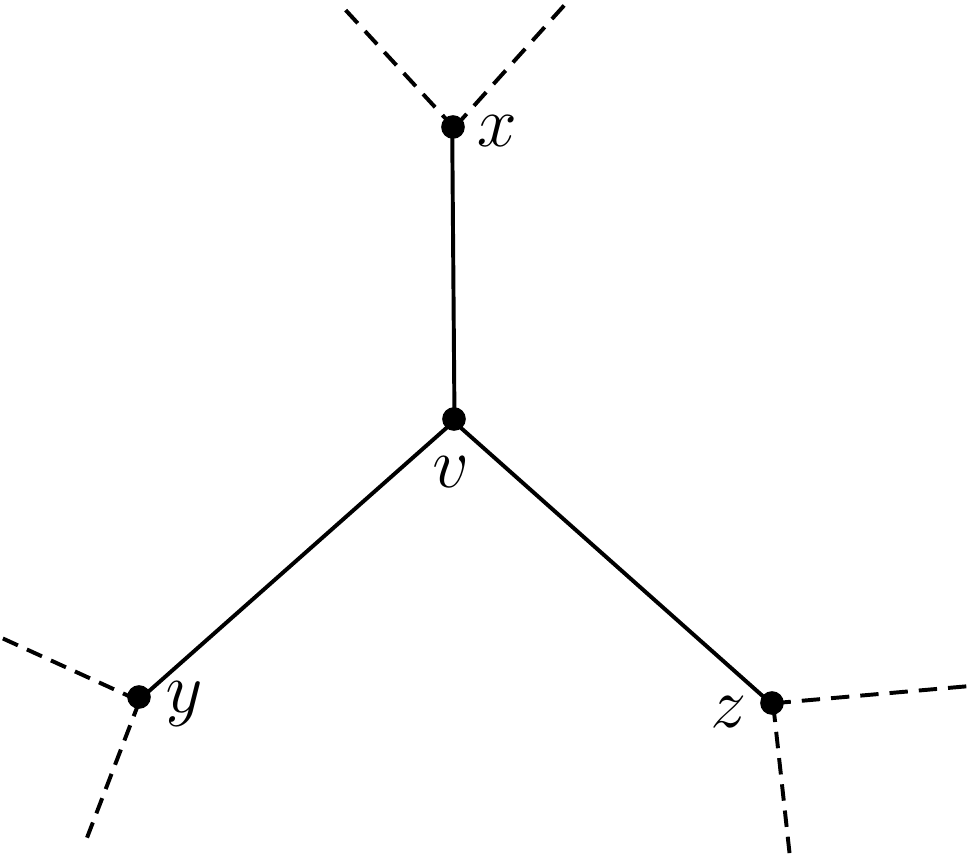}}\quad\rightarrow\quad\raisebox{-0.125\textwidth}{\includegraphics[width=0.3\textwidth]{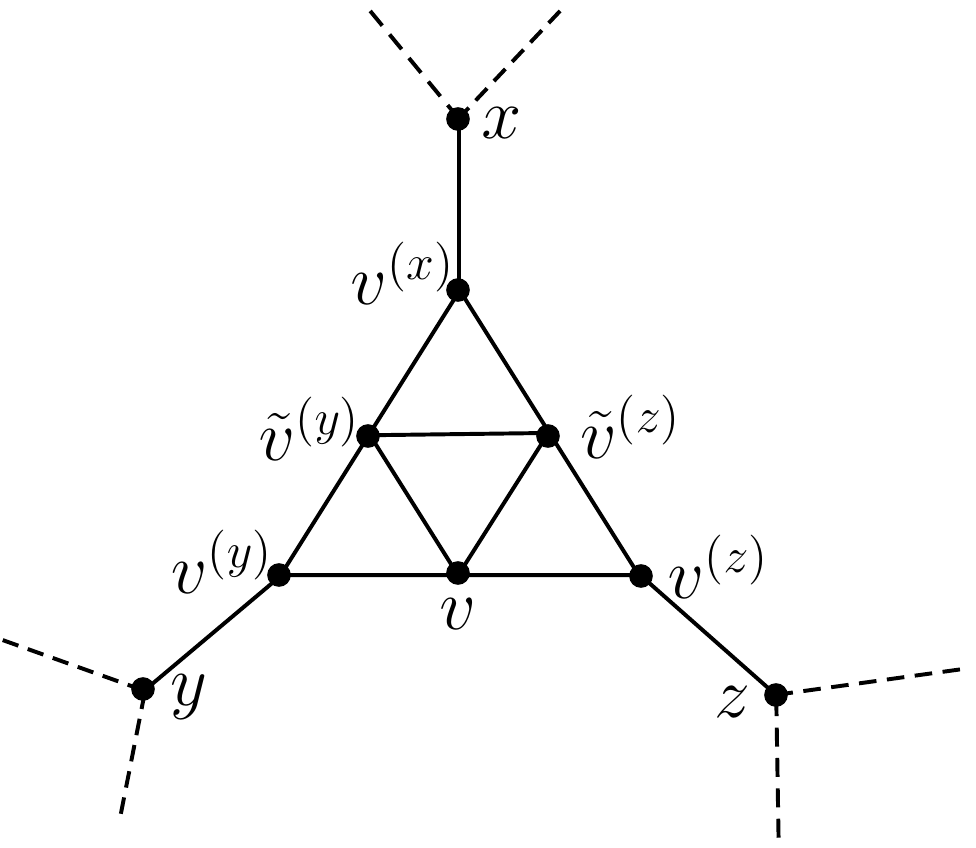}},
            \end{equation}
            where $x$, $y$ and $z$ are the neighbors of $v$. We will denote this \emph{triangle subgraph} associated to the vertex $v$ with $T_v$, i.e. $T_v=G[\{v,v^{(x)},v^{(y)},v^{(z)},\tilde{v}^{(y)},\tilde{v}^{(z)}\}]$.
        \item Double every edge that is incident on two subgraphs $T_v,T_{v'}$.
    \end{enumerate}
\end{mydef}
The graph $\Lambda(R)$ will be called a \emph{triangular-expansion} of $R$. A multi-graph $F$ that is the triangular-expansion of some $3$-regular graph $R$ will also be referred to as a triangular-expanded graph.
Note that the triangular-expansion is not uniquely defined, since for each vertex $v\in R$ there is a choice how to orient the triangle with respect to the neighbors of $v$.
Furthermore, the number of vertices in $\Lambda(R)$ is $6\cdot\abs{V(R)}$ and the number of edges is $2\cdot\abs{E(R)}+9\cdot\abs{V(R)}$. In \cref{fig:triangular_expansion_example} we show an example of a $3$-regular graph and its triangular-expansion.

For a given triangle subgraph $T_v$ in a triangular-expanded graph, we will refer to the vertices adjacent to other triangle subgraphs $T_{x},T_{y},T_{z}$ as 'outer vertices' and label them according to the triangle subgraph they are adjacent to. Concretely we label the vertex in $T_{v}$ that is adjacent to $T_{w}$ as $v^{(w)}$, the index signifies which triangle subgraph it connects to.

\begin{figure}[H]
\begin{equation*}
    \raisebox{-0.04\textwidth}{\includegraphics[width=0.11\textwidth]{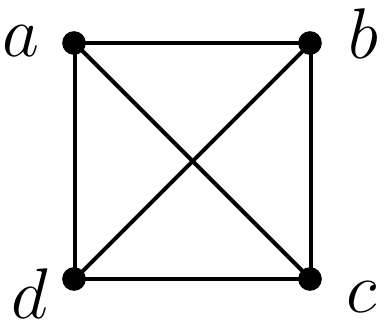}}\qquad\xrightarrow{\Lambda}\qquad\raisebox{-0.17\textwidth}{\includegraphics[width=0.42\textwidth]{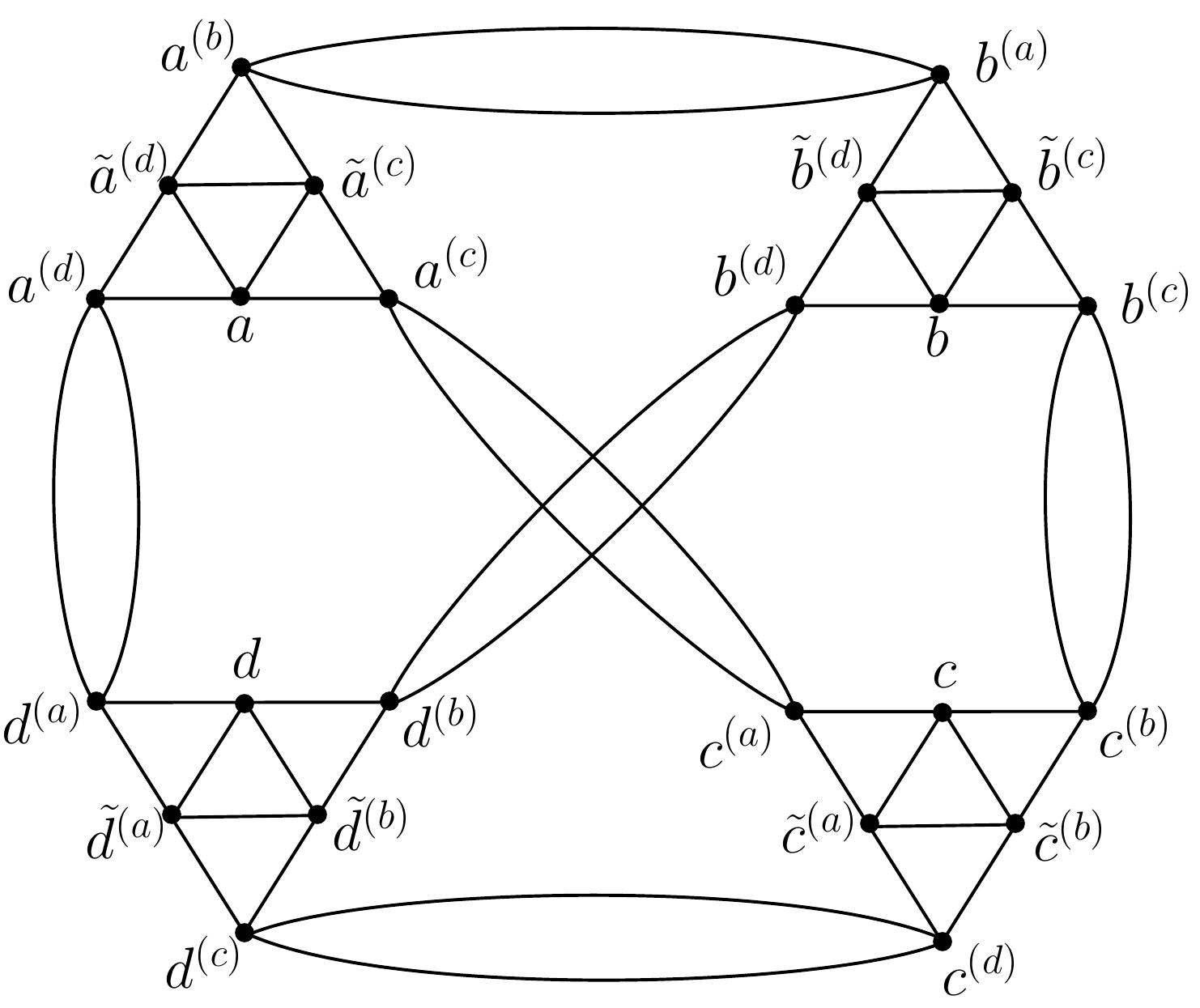}},
    \end{equation*}
    \caption{Figure showing the complete graph on vertices $V=\{a,b,c,d\}$ and its associated triangular-expansion $\Lambda(K_V)$.}\label{fig:triangular_expansion_example}
\end{figure}

In the following lemma we argue that this construction can be made efficiently in the size of $R$.

\begin{lem}\label{lem:triang_eff}
Let $R$ be a $3$-regular graph. We can construct its triangular-expansion in $O(|V(R)|^2)$ time. 
\end{lem}
\begin{proof}
Let $R$ be a $3$-regular graph. Without loss of generality assume some labeling on the vertices of $R$, i.e. $V(R) = \{v_1,\ldots,v_k\}$ where $k= |V(R)|$. We begin by constructing the vertex set $V(\Lambda(R))$ off the triangular-expansion $\Lambda(R)$ of $R$.\\

 For each $i\in [k]$ construct the set $V_i = \{v_i, \tilde{v}_i^{(v_j)},\tilde{v}_i^{(v_{j'})},v_i^{(v_j)},{v_i}^{(v_{j'})},{v}^{(v_{\hat{j}})}\}$ where $v_i\in V(R)$ and $v_j,v_{j'},v_{\hat{j}}$ are the three unique neighbors of $v_i$ in $R$. Constructing this set takes $O(|V(R)|)$ as we must search the set of edges $E(R)$ of $R$ to find the neighbors of $v_i$ and we have that $|E(R)|= O(|V(R)|)$ since $R$ is $3$-regular. Thus constructing the set $V(\Lambda(R)) = \cup_{i\in [k]}V_i$ takes $O(|V(R)|^2)$ time. \\

 Now we construct the edge multi-set $E(\Lambda(R))$ of the triangular-expansion of $R$. For each $i \in [k]$ we define the multi-set 
 \begin{align}
 E_i = \{(v_i, \tilde{v}_i^{(v_j)}),&(v_i, \tilde{v}_i^{(v_{j'})}),(v_i, {v}_i^{(v_j)}),(v_i, v_i^{(v_{j'})}),(\tilde{v}_i^{(v_{j})}, \tilde{v}_i^{(v_{j'})}), ({v}_i^{(v_{j})}, \tilde{v}_i^{(v_{j})}),({v}_i^{(v_{j'})}, \tilde{v}_i^{(v_{j'})}),\\& ({v}_i^{(v_{\hat{j}})}, \tilde{v}_i^{(v_{j})}),({v}_i^{(v_{\hat{j}})}, \tilde{v}_i^{(v_{j'})}), ({v}_i^{(v_{j})},{v}_j^{(v_{i})}),({v}_i^{(v_{j'})},{v}_{j'}^{(v_{i})}),({v}_{i}^{(v_{\hat{j}})}, 
 {v}_{\hat{j}}^{(v_{i})}) \}.
 \end{align}
This multi-set can be constructed in constant time. Hence the multi-set $E(\Lambda(R)) =\cup_{i\in [k]}E_i $  can be constructed in $O(|V(R)|)$ time. It is easy to check that the multi-graph defined by the vertex set $\Lambda(R)$ and edge multi-set $E(\Lambda(R))$ is indeed the triangular-expansion of $R$. This completes the lemma.

\end{proof}

\subsubsection{Skips and true skips}\label{sec:skips}

A key insight in the behavior of the \SOET~problem on triangular-expanded graphs is the notion of \emph{skips}.
The word skip stems from the fact that since any \SOET~$U$ on the triangular-expansion $\Lambda(R)$ of a $3$-regular graph is a Eulerian tour, it must traverse any triangle subgraph $T_v$ of $\Lambda(R)$ exactly three times.
However in order for $U$ to be a valid \SOET~with respect to $V(R)$ it must traverse the vertex $v$ exactly two of those three times.
This means it must \emph{skip} the vertex $v$ exactly once while traversing $T_v$.

We state a more formal definition of a (true) skip in terms of maximal sub-words (see \cref{def:maximal}).

\begin{mydef}[Skip]\label{def:skip}
    Let $R$ be a $3$-regular graph and let $\Lambda(R)$ be its triangular-expansion.
    Let $U$ be a SOET on $\Lambda(R)$ with respect to $V(R)$.
    Let $\bs{X}$ be a maximal sub-word of $m(U)$ (\cref{def:maximal}) associated to vertices $u,v \in V(R)$.
    We say the sub-trail described by $\bs{X}$ makes a \emph{skip} at a vertex $w\in V(R)\setminus\{u,v\}$ if $x_1^{(w)}w^{(x_1)}$ and $x_2^{(w)}w^{(x_2)}$ are sub-words of $\bs{X}$ (up to reflection), where $x_1, x_2 \in V(R)$.
    Furthermore, if $x_1\neq x_2$ then we say that the trail described by $\bs{X}$ makes a true skip at $w$ or sometimes that $T_w$ contains a \emph{true skip}.
\end{mydef}

Note that since $\bs{X}$ is a maximal sub-word associated to $u,v$ and $w\notin\{u,v\}$, $w$ cannot be a letter of $\bs{X}$.
As stated above, there is always exactly one maximal sub-word describing a sub-trail of a \SOET~that makes a skip at a certain triangle subgraph, as formalized in the following lemma. One can think of this lemma as giving necessary conditions for the existence of a \SOET~with respect to $V(R)$ on the triangular-expansion of a $3$-regular graph $R$\\

\begin{lem}\label{lem:single_skip}
    Let $R$ be a $3$-regular graph and let $\Lambda(R)$ be its triangular-expansion.
    Let $U$ be a Eulerian tour on $\Lambda(R)$.
    Let $w\in V(R)$ and let $T_w$ be its triangle subgraph in $\Lambda(R)$.
    If $U$ is a \SOET~on $\Lambda(R)$ with respect to $V(R)$ then here exist exactly one maximal sub-trail of $U$ that makes a skip at $w$.
\end{lem}
\begin{proof}
    We will prove this by showing that there are exactly three maximal sub-trails of $U$ that traverse vertices in $T_w$ and that exactly one of these makes a skip at $w$.
    Note first that the Eulerian tour $U$ will enter and exit the triangle subgraph $T_w$ exactly three times, since there are six edges incident to $T_w$.
    Hence there exists exactly three distinct edge-disjoint sub-trails, $t_1$, $t_2$ and $t_3$ of $U$ that exit and enter $T_w$, i.e. the last vertices they traverse in $T_w$ will be $x_i^{(w)}$, where $x_i$ for $i\in[3]$ are the neighbors of $w$ in $R$.
    Note that $t_1$, $t_2$ and $t_3$ each contain at least one vertex in $T_w$ and they jointly traverse all edges in $T_w$ (Since $U$ is a Eulerian tour).\\

    Now consider the vertex $w$. The Eulerian tour $U$ traverses this vertex exactly twice. There are now two options for the trails $t_1,t_2,t_3$. Either (1) one of the trails contains the vertex $w$ exactly twice or (2) there are exactly two trails that contain the vertex $w$ exactly once.\\

    Now assume $U$ is a \SOET\ with respect to $V(R)$. If option (1) is true the tour $U$ traverses the vertex $w$ twice in succession before traversing any other vertex in the set $V'$. This is in contradiction with the assumption that $U$ is a SOET. Hence if $U$ is a SOET we must have that option (2) is true, that is, exactly two of the sub-trails $t_1$, $t_2$ and $t_3$ must contain $w$.\\

    Lets assume without loss of generality that $w\in t_2$ and $w\in t_3$.
    Hence we have that $w\not\in t_1$ and thus $t_1$ induces a sub-word $m(t_1)$ not containing $w$.
    The word $m(t_1)$ can be extended to a maximal sub-word $\bs{X}$ not associated to $w$ but describing a sub-trail traversing vertices in $T_w$. 
    Therefore by \cref{def:skip} $\bs{X}$ describes a sub-trail making a skip at $w$.
    Furthermore, $m(t_2)$ has an overlap with two maximal sub-words associated with $y_1$, $w$ and $w$, $y_2$, respectively for some $y_1,y_2\in V(R)$.
    Similarly for $m(t_3)$ and the vertices $z_1$, $w$ and $w$, $z_2$.
    Thus, the maximal sub-words that have an overlap with $m(t_2)$ or $m(t_3)$ do not make a skip at $w$.\\

    Finally, there is no other maximal sub-word describing a sub-trail that traverses a vertex in $T_w$.
    The reason for this is that $t_1$, $t_2$ and $t_3$ jointly traverse all edges in $T_w$, so any maximal sub-trail traversing a vertex in $T_w$ must have an overlap with $t_1$, $t_2$ or $t_3$.
    The lemma then follows since we found exactly one maximal sub-word describing a sub-trail of $U$ making a skip at $w$, i.e. the unique one that has an overlap with $t_1$.
\end{proof}

\subsubsection{Equivalence between \SOET's~and Hamiltonian cycles}\label{sssec:equiv}

Now that we have defined the triangular-expansion $\Lambda(R)$ of a $3$-regular graph $R$ and discussed skips we can finally make the central argument of the reduction given in \cref{cor:SOETNP}. We begin by proving that if a $3$-regular graph is Hamiltonian then its corresponding triangular-expansion $\Lambda(R)$ allows for a \SOET~ w.r.t. the vertex set $V(R)$. We have the following lemma.

\begin{lem}\label{lem:HAM2SOET}
    Let $R$ be a 3-regular graph and $\Lambda(R)$ be a triangular-expansion as defined in \cref{def:triangular}.
    If $R$ is Hamiltonian, then $\Lambda(R)$ allows for a SOET with respect to $V(R)$.
\end{lem}
\begin{proof}
    Let $R$ be Hamiltonian.
    This means there exists a Hamiltonian cycle $M$ on $R$.
    We will prove that there exists a SOET $U$ with respect to $V(R)$ on the triangular-expansion $\Lambda(R)$ of $R$ by constructing, from the cycle $M$ on $R$, a tour $U$ that visits every vertex $v\in V(R)$ twice in the same order.
    We will then argue that this tour can always be lifted to a Eulerian tour and hence can be made into a SOET.\\

    Note first that $M$ induces an ordering on the vertices of $R$ which without of loss of generality we will take to be $v_1, \ldots v_k$ where $k=|V(R)|$.
    Note that for all $i\in [k-1]$ the vertices $v_i$ and $v_{i+1}$ are adjacent in $R$ and so are $v_k$ and $v_1$.
    Now consider, for each $i\in [k-1]$ the triangle subgraphs $T_{v_i}$, $T_{v_{i+1}}$ and $T_{\hat{v}}$ in the triangular-expansion $\Lambda(R)$ of $R$ where $\hat{v}_i$ is the unique vertex adjacent to $v_i$ in $R$ that is not $v_{i+1}$ or $v_{i-1}$.
    There are now three cases, depending on the orientation of the triangle subgraph $T_{v_i}$.
    Either (1) $v_i$ is adjacent to $v_i^{(v_{i-1})}$  and $v_i^{(v_{i+1})}$ or (2) $v_i$ is adjacent to $v_i^{(v_{i-1})}$  and $v_i^{(\hat{v_{i}})}$ or (3) $v_i$ is adjacent to $v_i^{(v_{i+1})}$  and $v_i^{(\hat{v}_{i})}$.\\

    For case (1), i.e. $v_i$ is adjacent to $v_i^{(v_{i-1})}$  and $v_i^{(v_{i+1})}$ in $T_{v_i}$, we can define two edge-disjoint trails on the triangular-expansion $\Lambda(R)$ by their description as words $\bs{X}_i,\bs{X}'_i$ on the vertices of $\Lambda(R)$.
    \begin{align}
        \bs{X}_i &= v_i^{(v_{i-1})}\tilde{v}_i^{(v_{i-1})} v_i^{(\hat{v}_i)} \tilde{v}_i^{(v_{i+1})} v_i v_i^{(v_{i+1})} v_{i+1}^{(v_{i})}, \label{eq:subtrail1}\\
    \bs{X}_i' &= v_i^{(v_{i-1})}v_i \tilde{v}_i^{(v_{i-1})} \tilde{v}_i^{(v_{i+1})} v_i^{(v_{i+1})}v_{i+1}^{(v_{i})}. \label{eq:subtrailtwo}
    \end{align}
    An illustration of the trails described by these words is given in \cref{fig:ham2soet1}.
    We can similarly define two edge-disjoint trails for the cases (2) and (3).
    We will abuse notation and refer to the words as $\bs{X}_i$ and $\bs{X}_i'$ in all three cases.
    Importantly, for all three these cases the trails described by $\bs{X}_i,\bs{X}_i'$ are edge-disjoint and both begin at $v_i^{(v_{i-1})}$ and end at $v_{i}^{(v_{i+1})}$.
    We can extend this definition to trails $U_0,U_0'$ also for the case of the adjacent vertices $v_k$ and $v_1$.
    Note now that the walk $U$ described by the word $\bs{W} = \bs{X}_0 \bs{X_1} \ldots \bs{X}_{k-1}\bs{X}'_0\bs{X}_1'\ldots \bs{X}_{k-1}$ is a tour and moreover that it traverses the vertices in $V(R)$ twice in the order $v_1, \ldots v_k$.\\

    The tour $U$ described by the word $\bs{W}$ above visits every vertex in $V(R)$ such that $\bs{W}[V(R)]=v_1\dots v_kv_1\dots v_k$.
    However it is not yet a Eulerian tour, since it has not traversed all edges in $\Lambda(R)$.
    Note that the edges not traversed by $U$ are precisely the triangular-expansions of the edges in $R$ not traversed by the Hamiltonian cycle $M$.
    We can easily construct a Eulerian tour out of $U$ by looping over all elements of the word $\bs{X}$ and whenever we encounter a vertex $v_i^{(\hat{v}_i)}$ for all $i \in [k]$, where $\hat{v}_i$ is defined as above, we check whether $U$ already traverses the edges $(v_i^{(\hat{v}_i)},\hat{v}_i^{(v_i)})$.
    If so we continue the loop and if not we insert the trail $(v_i^{(\hat{v}_i)},\hat{v}_i^{(v_i)}) \hat{v}^{(v_i)}(v_i^{(\hat{v}_i)},\hat{v}_i^{(v_i)})v_i^{(\hat{v}_i)}$ into $U$ at this position.
    This procedure is illustrated in \cref{fig:ham2soet2}.
    Now $U$ is Eulerian and hence a \SOET.
    This completes the proof.
\end{proof}

\begin{figure}[H]
    \centering
    \begin{subfigure}{0.35\textwidth}
        \includegraphics[width=\textwidth]{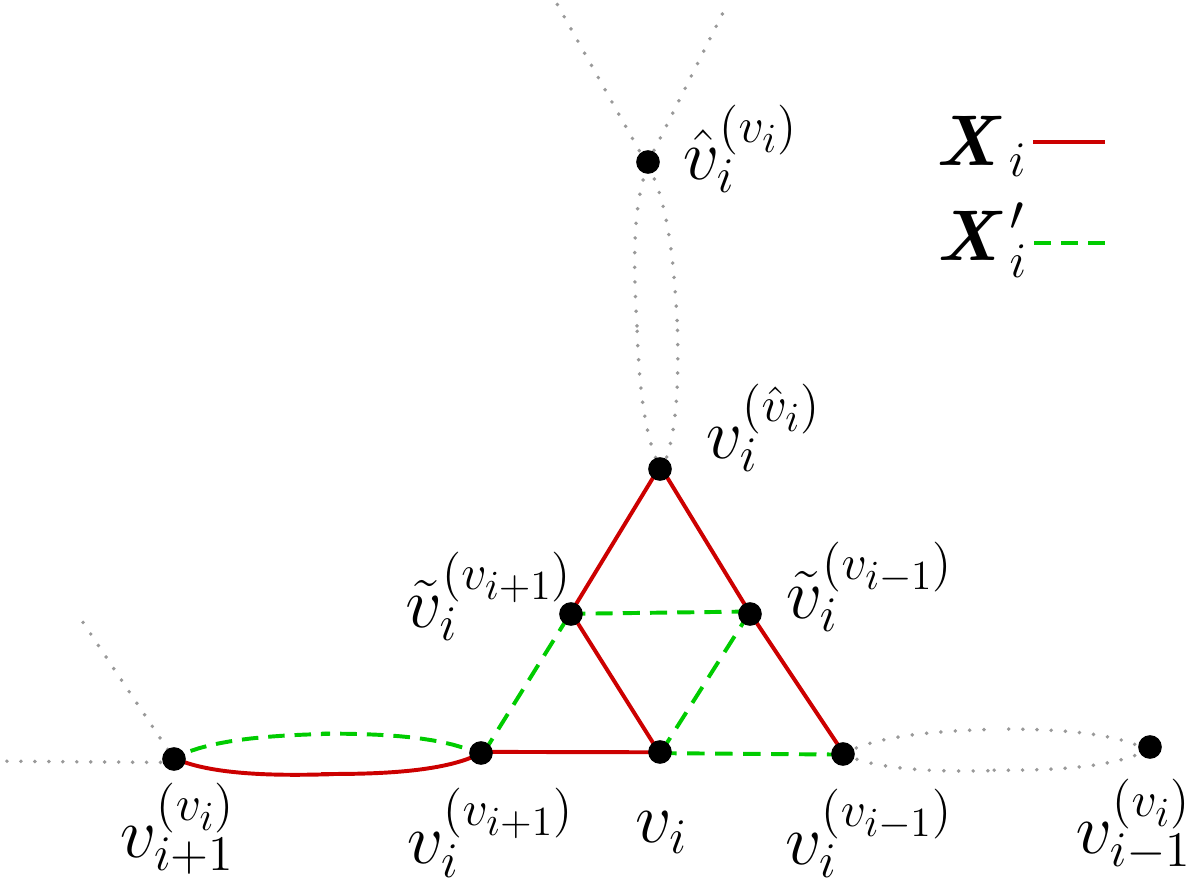}
        \caption{ }
        \label{fig:ham2soet1}
    \end{subfigure}
    \hspace{3em}
    \begin{subfigure}{0.35\textwidth}
        \includegraphics[width=\textwidth]{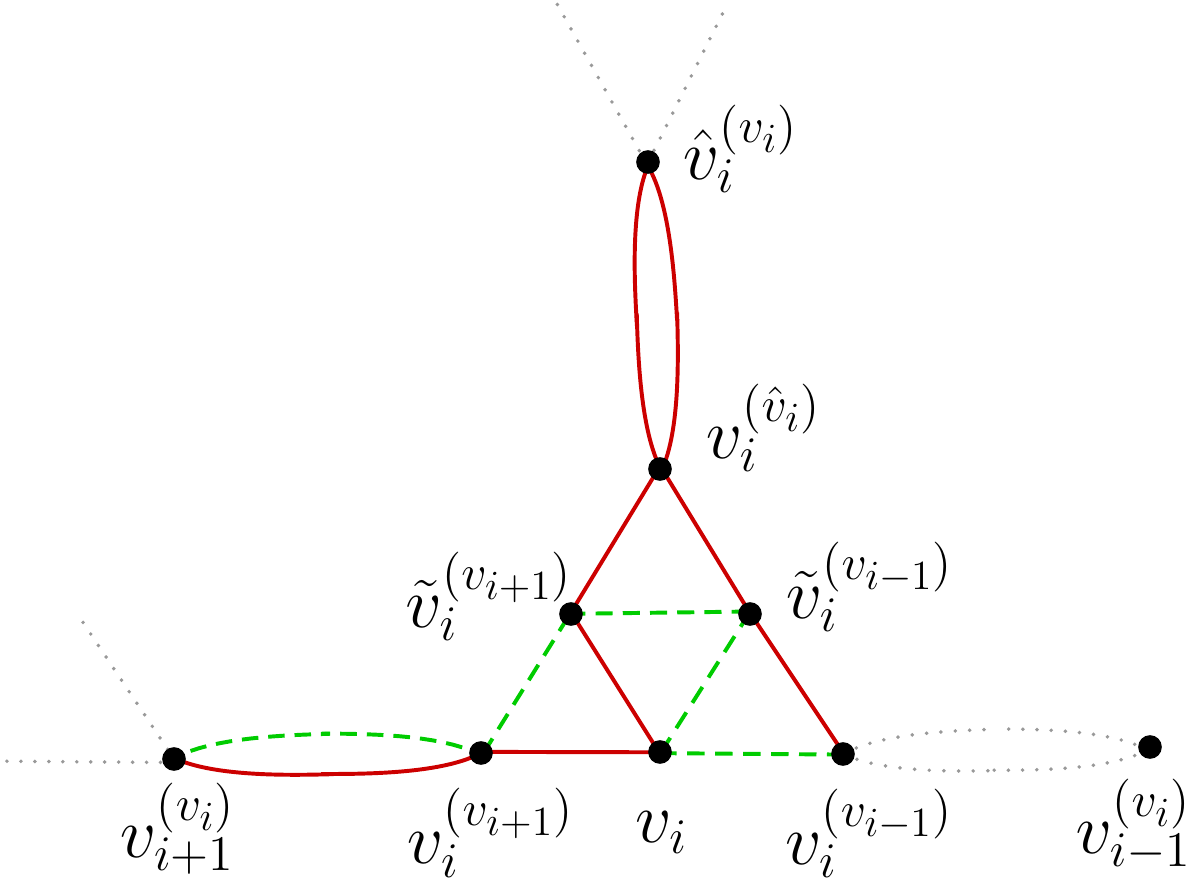}
        \caption{ }
        \label{fig:ham2soet2}
    \end{subfigure}
    \caption{The trails used in the proof of \cref{lem:HAM2SOET} to construct a \SOET\ on a triangular-expanded graph $\Lambda(R)$ from a Hamiltonian cycle on $R$.
        \Cref{fig:ham2soet1} shows the trails described by the words $\bs{X}_i$ (solid red) and $\bs{X}'_i$ (dashed green), defined in \cref{eq:subtrail1,eq:subtrailtwo}.
        \Cref{fig:ham2soet2} shows how these trails can be extended to form a Eulerian tour and therefore a \SOET\ with respect to $V(R)$.
        The dotted gray lines show edges of $\Lambda(R)$ which are not used by the trails.
    }
    \label{fig:ham2soet}
\end{figure}

Next we define a special type of \SOET~on triangular-expanded graphs, which we call \HAMSOET s.
These special \SOET s on a triangular-expanded graph $\Lambda(R)$, are closely related to Hamiltonian cycles on $R$.

\begin{mydef}[HAMSOET]\label{def:HAMSOET}
    Let $R$ be a $3$-regular graph and $\Lambda(R)$ its triangular-expansion.
    Furthermore, let $U$ be a \SOET~on $\Lambda(R)$ with respect to $V(R)$.
    $U$ is called a \HAMSOET\ with respect to $R$ if, for all vertices $u,v\in V(R)$ we have that if $u$ and $v$ are consecutive with respect to the \SOET~$U$ they are also adjacent in the graph $R$.
\end{mydef}

Next we prove that if the triangular-expansion of a $3$-regular graph $R$ allows for a \HAMSOET~with respect to $R$ then the $3$-regular graph $R$ is Hamiltonian. 

\begin{lem}\label{lem:HAMSOET2HAM}
    Let $R$ be a $3$-regular graph and let $\Lambda(R)$ be its triangular-expansion.
    If $\Lambda(R)$ allows for a \HAMSOET~with respect to $R$, then $R$ is Hamiltonian.
\end{lem}
\begin{proof}
    This follows by the definition of a \HAMSOET.
    A HAMSOET $U$ will induce a double occurrence word of the form 
    \begin{equation}
        m(U)=\bs{X}_0s_1\bs{X}_1s_2\dots s_k\bs{X}_{k}s_1\bs{X}'_{1}s_2\dots s_k\bs{X}'_{k}.
    \end{equation}
    where $s_1, \ldots s_k\in V(R)$ with $k=|V(R)|$ and where $\bs{X}_i$ and $\bs{X}'_i$ are maximal sub-words associated to $s_i,s_{i+1}\in V(R)$.
    Now consider the induced double occurrence word $m(U)[V(R)]$. We have 
    \begin{equation}
    m(U)[V(R)] = s_1 \ldots s_ks_1\ldots s_k.
    \end{equation}
    Consider now the sub-word $s_1\ldots s_k$ of $m(U)[V(R)]$ . This sub-word describes a Hamiltonian cycle on $R$.
    To see this, note that each vertex in $V(R)$ occurs exactly once in $s_1s_2\dots s_k$.
    Furthermore, since $s_i$ and $s_{i+1}$ are consecutive in $U$, they are adjacent in $R$ for all $i\in[k-1]$, by definition of a HAMSOET.
    Finally, the same also holds for $s_k$ and $s_1$. Hence the tour on $R$ described by $s_1\ldots s_k$ visits each vertex in $R$ exactly once and is hence a Hamiltonian cycle.

    \end{proof}

We would like to prove that for any $3$-regular graph $R$ the existence of a SOET on its triangular-expansion $\Lambda(R)$ implies the existence of a Hamiltonian cycle on $R$. So far we have proven this only when $\Lambda(R)$ allows for a HAMSOET. However, a priori not all SOETs on triangular-expansions $\Lambda(R)$ have to be HAMSOETs. The reason for this is that consecutive vertices in a \SOET~are not necessarily adjacent in $R$, since a \SOET\ can contain true skips.\\

In the following lemma we will prove that consecutive vertices in a \SOET~are actually adjacent in $R$, except for two special cases.
These special cases can be remedied since we show that for these cases we can always find a different \SOET~which is actually a \HAMSOET.\\

The arguments proven below are understood the easiest when one reproduces the visual aids given in \cref{fig:rrskip_explain,fig:11skip_explain} as one follows the arguments. We have the following lemma.

\begin{lem}\label{lem:SOET2HAMSOET}
    Let $R$ be a 3-regular graph and $\Lambda(R)$ be its triangular-expansion.
    If $\Lambda(R)$ allows for a \SOET ~with respect to $V(R)$, then $\Lambda(R)$ allows for a \HAMSOET~with respect to $V(R)$.
\end{lem}
\begin{proof}
To prove this lemma we will go through the following steps
\begin{description}
    \item[{Step~1}] Note that if two vertices $u,v$ in $V(R)$ ($R$ being a $3$-regular graph) are not adjacent in $R$ but are consecutive in a \SOET~$U$ with respect to $V(R)$ on $\Lambda(R)$, then the  sub-trails of $U$ described by the maximal sub-words associated to $u$ and $v$ must make a non-zero number of true skips in $\Lambda(R)$
    \item[Step 2] Argue by contradiction that this non-zero number of true skips can never be greater than one.
        This argument leverages \cref{lem:UNUSEDEDGE_1} which states that if a sub-trail of a \SOET~$U$ with respect to $V(R)$ makes true skips at triangle subgraphs $T_{w_1},T_{w_2}$ for $w_1,w_2\in V(R)$ and $w_1,w_2$ are adjacent in $R$ then $w_1, w_2$ must actually be consecutive in the \SOET~$U$ and \cref{lem:UNUSEDEDGE_2} which, in a slightly different initial situation than \cref{lem:UNUSEDEDGE_1}, also concludes that two vertices must be consecutive in a \SOET~$U$ with respect to $V(R)$.
    \item[Step 3] Argue by contradiction that there are only two possible ways for the sub-trails described by the maximal sub-words associated to $u$ and $v$ to make one true skip each.
        This argument also leverages \cref{lem:UNUSEDEDGE_1,lem:UNUSEDEDGE_2}.

    \item[Step 4] Argue that if a triangular-expansion of a $3$-regular graph $R$ allows for a~\SOET~$U$ as in {\bf Step 3}, then it also allows for a \SOET~$U'$ that is a \HAMSOET.
\end{description}

{\bf \underline{Details of step 1}}

    Let $R$ be such that $\Lambda(R)$ allows for a \SOET~with respect to $V(R)$. Let $U$ be any such \SOET.
    If $U$ is also a \HAMSOET~then we are directly done, therefore assume that $U$ is not a \HAMSOET.
    Since $U$ is not a \HAMSOET~there must exist at least two vertices $u,v\in V(R)$ which are consecutive in $U$, but not adjacent in $R$.
    By definition, since $u$ and $v$ are consecutive and $U$ is a \SOET, there must exist exactly two different maximal sub-words $\bs{X},\bs{X}'$ of $m(U)$, associated to $u$ and $v$.
    Since $u$ and $v$ are not adjacent in $R$, the trails described by $\bs{X}$, $\bs{X}'$ must each traverse vertices in at least one (but possibly more) triangle subgraph that is not $T_u$ or $T_v$. Since $\bs{X}$ and $\bs{X}'$ are maximal sub-words, they can not contain any vertex in $V(R)$ as a letter and hence the sub-trails described by $\bs{X},\bs{X}'$ must each  make true skips (see \cref{def:skip}) at at least one triangle subgraph.\\

    Assume w.l.o.g. that the sub-trail described by $\bs{X}$ makes true skips at the triangle subgraphs $T_{s_1},T_{s_2},\dots T_{s_r}$, in this order, and similarly for $\bs{X}'$ and $T_{s'_1},\dots T_{s'_{r'}}$.
    The true skips that the sub-trail described by $\bs{X}$ makes, must be at different triangle subgraphs than the ones for $\bs{X}'$, since there can only be exactly one skip per triangle subgraph, due to \cref{lem:single_skip}.
    The triangle subgraphs $T_{s_1},T_{s_2},\dots T_{s_r},T_{s'_1},\dots T_{s'_{r'}}$ are therefore pairwise different.
    We will call this situation a $rr'$-skip, which is illustrated in \cref{fig:rrp-skip}.
    Note that by assumption, $r$ and $r'$ are both greater than zero.
    We will first show that the case where either $r$ or $r'$ are greater than one can never occur if $U$ is a \SOET.

    \begin{figure}[H]
        \centering
        \includegraphics[width=0.3\textwidth]{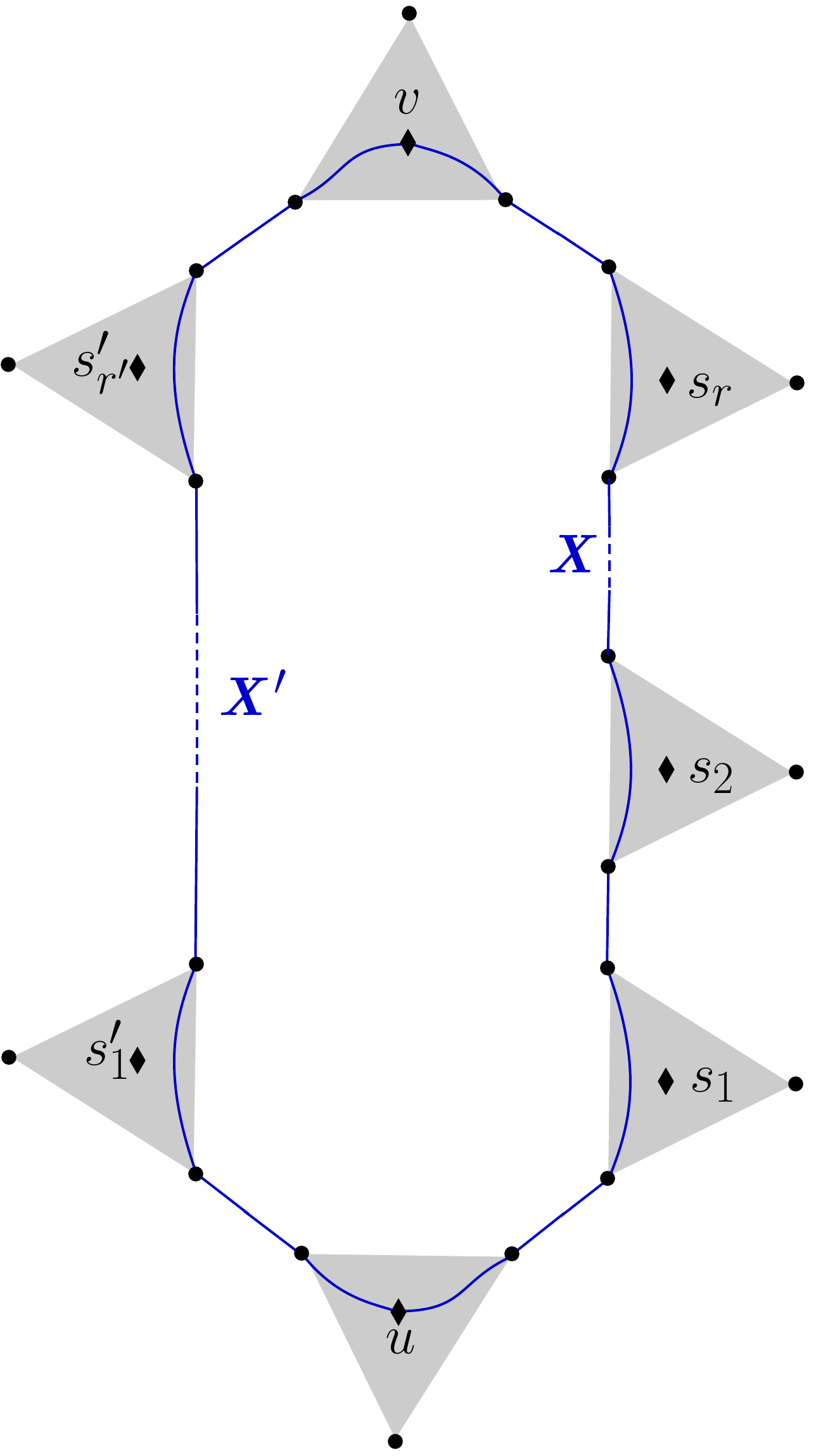}
        \caption{This figure is a visual aid for Step 1 of \cref{lem:SOET2HAMSOET}.
            A simplified visualization of a triangular-expansion $\Lambda(R)$ of a 3-regular graph $R$ is used.
            In the figure triangle subgraphs are shown in gray with only their outer vertices (circles) and vertices in $V(R)$ (diamonds) shown.
            Also shown are the sub-trails, of a SOET $U$, described by maximal sub-words $\bs{X}$ and $\bs{X}'$ associated to vertices $u,v$ making true skips at triangle subgraphs $T_{s_1}, \ldots, T_{s_r}$ and $T_{s'_1}, \ldots, T_{s'_r}$ respectively.
            These sub-trails always begin and end at a vertex in $V(R)$ but their path inside triangle subgraphs is not shown explicitly.
            Dashed lines are used to indicate that the sub-trails also traverse unspecified further parts of the graph $\Lambda(R)$.
        }
        \label{fig:rrp-skip}
    \end{figure}

{\bf\underline{Details of step 2}}

    \noindent As a visual aid for the following argument, refer to \cref{fig:rrskip_explain}.
        We will make an argument by contradiction. Therefore let $r$ be strictly greater than one.
        Consider the sub-trail described by the sub-word $\bs{X}$ defined above.
        By assumption this sub-trail makes true skips at at least two adjacent triangle subgraphs.
        Take these to be $T_{s_1}$ and $T_{s_2}$.
        As we will prove in \cref{lem:UNUSEDEDGE_1}, if two adjacent triangle subgraphs $T_{s_1},T_{s_2}$ contain true skips, then $s_1$ and $s_2$ must be consecutive in the SOET $U$ and that there is some maximal sub-word $\bs{Y}$ of $m(U)$ associated to $s_1,s_2$ that contains the sub-word $s_1^{(s_2)}s_2^{(s_1)}$.\\ 
    
        This implies that (by definition of SOET) there is some other maximal sub-word $\bs{Y}'$ of $m(U)$ associated to $s_1$ and $s_2$. This sub-word $\bs{Y}'$ cannot contain the sub-word $s_1^{(s_2)}s_2^{(s_1)}$, as $s_1^{(s_2)}s_2^{(s_1)}$ already appears in both the maximal sub-word $\bs{X}$ and the maximal sub-word $\bs{Y}$.

        This implies the sub-trail described by $\bs{Y}'$ must make a nonzero number of true skips at triangle subgraphs $T_{\hat{s}_1}, \ldots, T_{\hat{s}_k}$ in that order.
        Now consider the triangle subgraphs $T_{s_1}$ and $T_{\hat{s}_1}$.
        As we will prove below, \cref{lem:UNUSEDEDGE_2} applied to the triangle subgraphs $T_{s_1},T_{\hat{s}_1}$ implies that the vertices $s_1$ and $\hat{s}_1$ must be consecutive in $U$. Call the maximal sub-word of $m(U)$ associated to these vertices $\bs{W}_0$. Moreover, by \cref{lem:UNUSEDEDGE_2} this maximal sub-word contains the sub-word $s_1^{(\hat{s}_1)}\hat{s}_1^{(s_1)}$.
        Similarly we can see that the vertices $s_2$ and $\hat{s}_k$ must be consecutive in $U$.
        Call the maximal sub-word of $m(U)$ associated to these vertices $\bs{W}_{k}$. By \cref{lem:UNUSEDEDGE_2} this maximal sub-word contains the sub-word $s_2^{(\hat{s}_{k+1})}\hat{s}_{k+1}^{(s_2)}$.

        By applying \cref{lem:UNUSEDEDGE_1} again to all triangle subgraph pairs $T_{\hat{s}_i}, T_{\hat{s}_{i+1}}$ for $i \in [k-1]$ we come to the conclusion that $\hat{s}_i,\hat{s}_{i+1}$  for $i \in [k-1]$ must be consecutive in $U$.
        Call the maximal sub-word of $m(U)$ associated to these vertices $\bs{W}_i$.  By \cref{lem:UNUSEDEDGE_2} these maximal sub-words contain the sub-word $\hat{s}_i^{(\hat{s}_{i+1})}\hat{s}_{i+1}^{(\hat{s}_{i})}$.
        This implies that $U$ must traverse the vertices $s_1,\hat{s}_1, \ldots, \hat{s}_k, s_2,s_1, \hat{s}_1, \ldots, \hat{s}_k, s_2$ in order.
        Since by construction $\{s_1,s_2, \hat{s}_1,\ldots,
            \hat{s_k}\} \neq V(R)$ we have that $U$ is not a valid \SOET~on $V(R)$ (this can be seen by noting that e.g. $T_{s'_1}$ already contains a true skip, hence $s'_1$ can not be part of the set).
        Hence by contradiction we must have $r\leq 1$.
        We can make the same argument for $r'$ yielding $r'\leq 1$.

        \begin{figure}[H]
            \centering
            \includegraphics[width=0.5\textwidth]{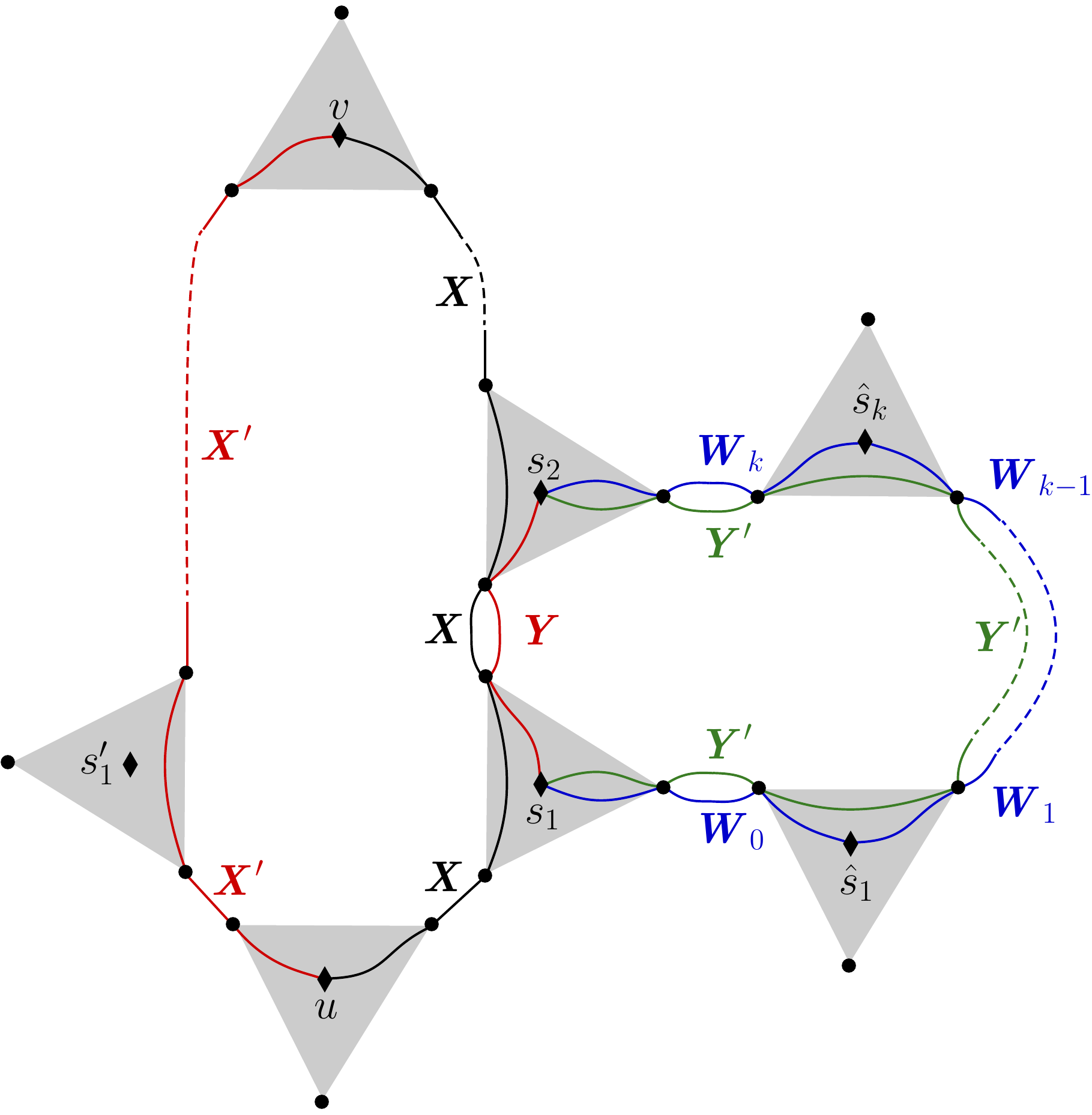}
            \caption{This figure is a visual aid for Step 2 of \cref{lem:SOET2HAMSOET}.
                A simplified visualization of a triangular-expansion $\Lambda(R)$ of a 3-regular graph $R$ is used.
                In the figure triangle subgraphs are shown in gray with only their outer vertices (circles) and the vertices in $V(R)$ (diamonds) shown.
                Also shown are sub-trail of a SOET $U$ labeled by the maximal sub-words that describe them.
                These sub-trails always begin and end at a vertex in $V(R)$ but their path inside triangle subgraphs is not shown explicitly.
                Dashed lines are used to indicate that the sub-trails also traverse unspecified further parts of the graph $\Lambda(R)$.
                In this argument a contradiction is arrived at by first assuming the that the SOET $U$ has a sub-trail described by the maximal sub-word $\bs{X}$ associated to the vertices $u,v$ makes true skips at triangle subgraphs $T_{s_1},T_{s_2}$.
                Using \cref{lem:UNUSEDEDGE_1} it is shown that the maximal sub-word $\bs{Y}$ must exist (associated to vertices $s_1,s_2$).
                This then means that the maximal sub-word $\bs{Y}'$ must exist (also associated to $s_1,s_2$).
                The sub-trail described by this sub-word must make true skips at triangle subgraphs $T_{w_1}, \ldots, T_{w_k}$.
                Using \cref{lem:UNUSEDEDGE_1} and \cref{lem:UNUSEDEDGE_2} it is then concluded that the SOET $U$ must visit the vertices $s_1, w_1, \ldots w_k,s_2,s_1$ consecutively which means $U$ is not a valid SOET (since $
            \{s_1, w_1, \ldots ,w_k, s_2\}\neq V(R)$). This is a contradiction.}
            \label{fig:rrskip_explain}
        \end{figure}

{\bf\underline{Details of step 3}}

    \noindent As a visual aid for the following argument, refer to \cref{fig:11skip_explain}.
        Now let $r = r'= 1$. This means the sub-trails described by the maximal sub-words $\bs{X}$ and $\bs{X}'$ associated to the vertices $u$ and $v$ make exactly one true skip each at triangle subgraphs triangle subgraphs $T_{s_1}, T_{s'_1}$ respectively. We will now argue that there are essentially only two ways that a SOET $U$ with the above properties can exist on $\Lambda(R)$. This argument will again go in steps:
        \begin{description}
        \item[Step 3.1] Argue that the SOET $U$ must have a maximal sub-word associated to the vertex $s_1$ and some other vertex $x$ that describes a trail traversing the edge connecting the triangle subgraphs $T_u$ and $T_{s_{1}}$.
        \item[Step 3.2] Argue that this sub-trail must make a true skip at the triangle subgraph $T_u$. Note that this is equivalent to arguing $x\neq u$. 
        \item[Step 3.3] Apply the same sequence of arguments for the vertices $s_1$ and $v$ and also for the vertices $s'_1$ and $u$ and $s'_1$ and $v$. 
        \item[Step 3.4] Conclude from the fact that the SOET $U$ can only make a single true skip at the triangle subgraphs $T_u$ and $T_v$ that $s_1$ and $s'_1$ must be consecutive in $U$ and moreover that the maximal sub-words $\bs{Z},\bs{Z}'$ associated to $s_1$ and $s'_2$ must make true skips at $T_u$ and $T_v$ respectively.
    \end{description}
    \begin{addmargin}{2em}

        {\bf \underline{Details of step 3.1}}

        Consider the second edge connecting the triangle subgraphs $T_u$ and $T_{s_{1}}$. This is the edge $(u^{(s_1)},s_1^{(u)})$.
        Since the \SOET~$U$ is Eulerian it must traverse this edge.
        Note also that $T_{s_1}$ already contains a true skip (made by the sub-trail described by $\bs{X}$). This means, by \cref{lem:single_skip} that any sub-trail of $U$ crossing the edge  $(u^{(s_1)},s_1^{(u)})$ connecting $T_u,T_{s_1}$ must traverse the vertex $s_1$.
        Let $\bs{Z}$ be the maximal sub-word describing the sub-trail starting at $s_1$ and containing the sub-word $u^{(s_1)}s_1^{(u)}$.
        This maximal sub-word is (by definition) associated to two vertices in $V(R)$.
        One of these vertices is $s_1$ and will label the other one $x$. 

        We will now argue that $x\neq u$ and hence that the sub-trail described by $\bs{Z}$ makes a true skip at $T_u$.
        We do this by contradiction in the following argument.\\

        {\bf \underline{Details of step 3.2}}

            For this part of the argument assume that $x=u$.
            This means that $u$ is consecutive to both $s_1$ and $v$.
            This means there must be some other maximal sub-word $\bs{Z}'$ associated to $s_1$ and $u$. Note that this sub-word cannot contain the sub-word $u^{(s_1)}s_1^{(u)}$ as it is already contained in the maximal sub-words $\bs{Z}$ and $\bs{X}'$.\\

            Now consider the unique third vertex that is adjacent to $u$ in $R$ (the vertex adjacent to $u$ which is not $s_1$ or $s'_1$).
            Let us label this vertex $w$.

            Since $T_{s'_1}$ already contains a true skip (which implies $\bs{Z}'$ cannot connect to $u$ by making a true skip at $T_{s'_1}$) the sub-trail described by the maximal sub-word $\bs{Z}'$ must make a true skip at $w$.
            Now consider the maximal sub-word $\bs{Y}$ of $m(U)$ that that describes a sub-trail traversing the unused edge between $T_{s'_1}$ and $T_{u}$, i.e. $\bs{Y}$ contains the sub-word $u^{(s'_1)}{s'}_1^{(u)}$.
            This sub-trail must be associated to $s'_1$ (since $T_{s'_1}$ already contains a true skip) and can not not be associated to $u$ since $u$ is already consecutive with two vertices in $V(R)$.
            This means the sub-trail described by the maximal sub-word $\bs{Y}$ must make a true skip at the triangle subgraph $T_u$.
            Since the sub-word $u^{(s_1)} s_1^{(u)}$ is already contained in $\bs{Z}$ and $\bs{X}'$ the maximal sub-word $\bs{Y}$ must contain the sub-word $u^{(w)}w^{(u)}$.
            Since $T_w$ already contains a true skip this implies that $w$ must be associated to the maximal sub-word $\bs{Y}$ and hence that $w$ and $s'_1$ are consecutive.
            This means there must be a second maximal sub-word $\bs{Y}'$ associated to $w$ and $s'_1$.

            Since both edges connecting $T_u$ and $T_w$ have already been traversed by, the sub-trail described by the maximal sub-word $\bs{Y}'$ must make true skips on some triangle subgraphs $T_{\hat{s}_1},\ldots, T_{\hat{s}_k}$.

            \begin{figure}[H]
            \begin{subfigure}{0.48\textwidth}
                \includegraphics[width=0.95\textwidth]{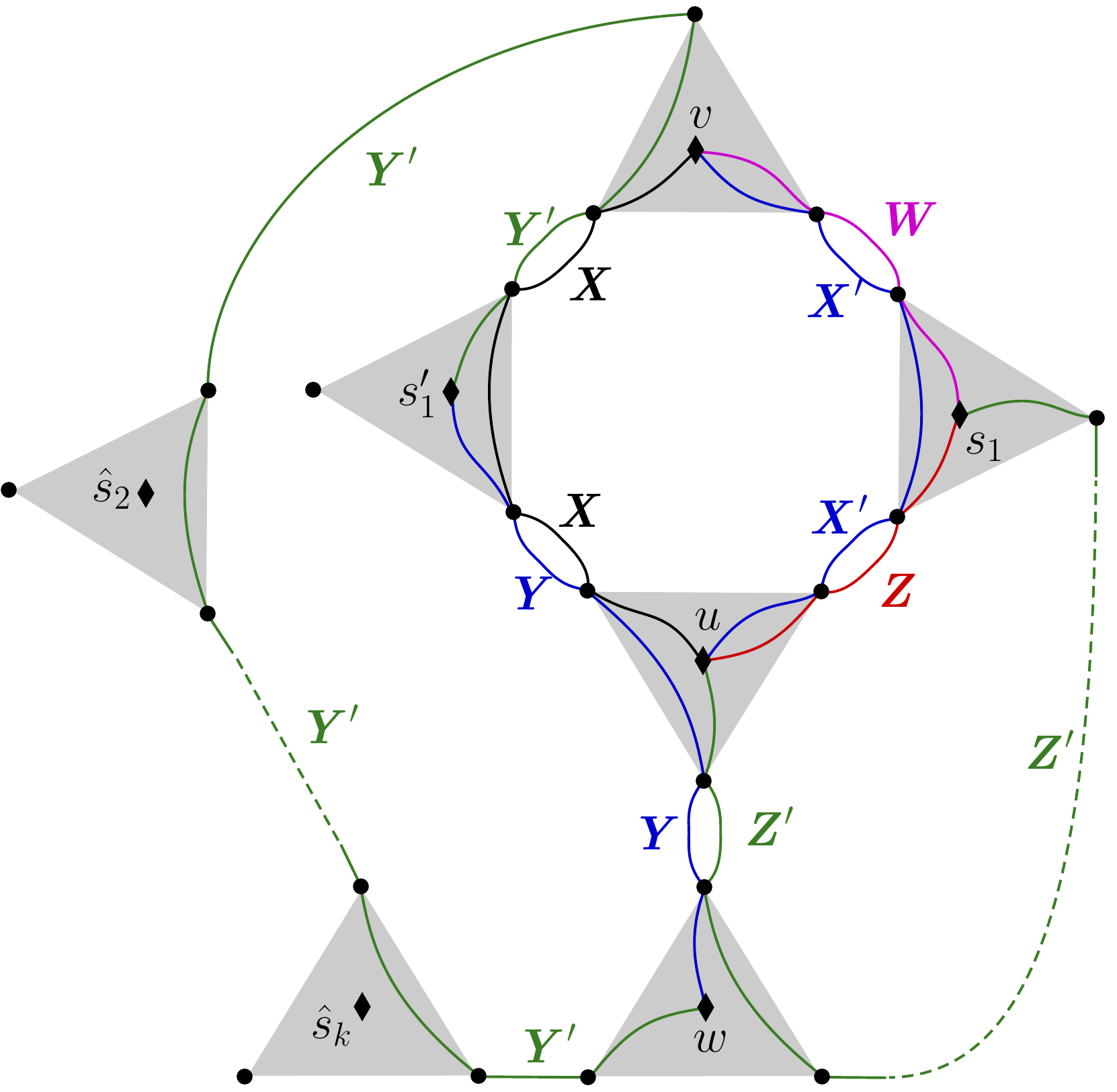}
                \caption{}\label{fig:11skip_v_is_s1}
            \end{subfigure}
            \begin{subfigure}{0.48\textwidth}
                \includegraphics[width=0.95\textwidth]{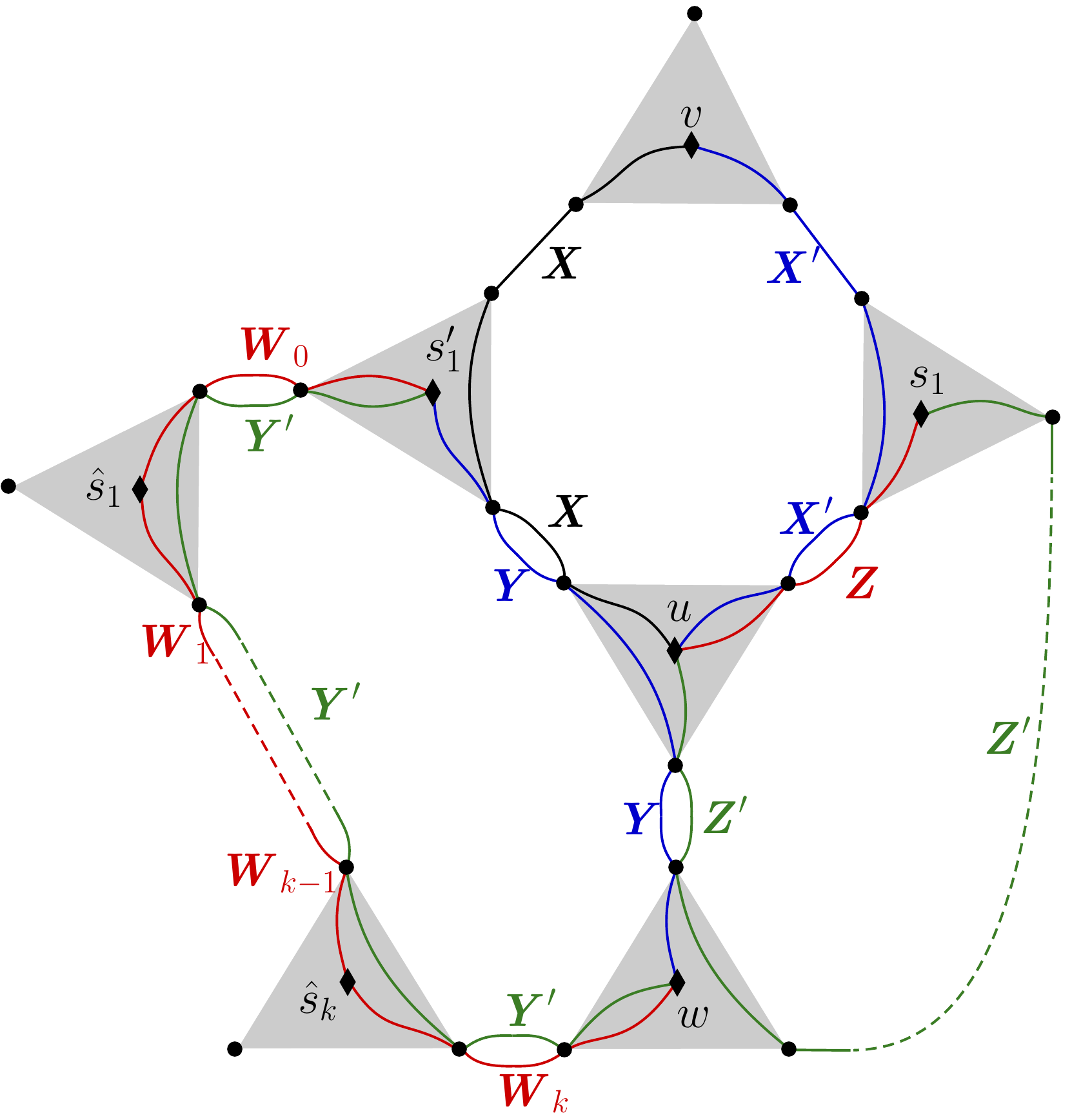}
                \caption{}\label{fig:11skip_v_is_not_s1}
            \end{subfigure} 
                \caption{
                This figure is a visual aid for Steps 3.1 and 3.2 of \cref{lem:SOET2HAMSOET}.
                A simplified visualization of a triangular-expansion $\Lambda(R)$ of a 3-regular graph $R$ is used.
                In the figure triangle subgraphs are shown in gray with only their outer vertices (circles) and vertices in $V(R)$ (diamonds) shown.
                Also shown are sub-trail of a SOET $U$ labeled by the maximal sub-words that describe them.
                These sub-trails always begin and end at a vertex in $V(R)$ but their path inside triangle subgraphs is not shown explicitly.
                Dashed lines are used to indicate that the sub-trails also traverse unspecified further parts of the graph $\Lambda(R)$.
                In Step 3.2 of \cref{lem:SOET2HAMSOET} it is argued that if the subtrail described by the maximal sub-word $\bs{X}'$ associated to vertices $u,v$ makes a true skip at the triangle subgraph $T_{s_1}$ then the sub-trail described by the maximal sub-word $\bs{Z}$ associated to $s_1$ and another vertex $x$ must make a true skip at the triangle subgraph $T_u$ (or equivalently that $x\neq u$).
                This is done by contradiction so it is assumed (as seen in both (a) and (b)) that $x=u$.
                This implies the existence of the maximal sub-word $\bs{Z}'$ which must make a true skip at $T_w$ where $w$ is the unique neighbor of $u$ s.t.
                $w\neq s_1, s_1'$.
                This in turn implies that $s_1'$ and $w$ must be consecutive and connected by a sub-trail described by the maximal sub-word labeled $\bs{Y}$ in both (a) and (b).
                This means another sub-trail connecting $s_1'$ and $w$ must exist, which described by a maximal sub-word $\bs{Y}'$.
                 This sub-word makes true skips at triangle subgraphs $\hat{s}_1, \ldots \hat{s}_k$.
                There are now $2$ options.
                Either, as shown in (a) we have that $v= \hat{s}_1$ or we have, as shown in (b) that $v\neq\hat{s}_1$.
                In the first case (a) \cref{lem:UNUSEDEDGE_2} is applied to conclude that $v$ must be consecutive to $s_1,s_1'$ and $u$ leading to a contradiction.
                In the second case (b) \cref{lem:UNUSEDEDGE_1} and \cref{lem:UNUSEDEDGE_2} are used to show that any SOET $U$ must traverse the vertices $s'_1,w, \hat{s}_k, \ldots ,\hat{s}_1, s_1'$ in order leading to a contradiction.
                }
                \label{fig:11skip_explain}
            \end{figure}
            There are now two possibilities. 
            Either (1) we have that $\hat{s}_1= v$ (see \cref{fig:11skip_v_is_s1}) or (2) that $\hat{s}_1 \neq v$ (see \cref{fig:11skip_v_is_not_s1}).
            We will now consider both of these cases:
            \begin{enumerate}
                \item If $\hat{s}_1 = v$ we can apply \cref{lem:UNUSEDEDGE_2} to the vertices $v$ and $s_1$ to conclude that $\hat{s}_1=v$ implies that $v$ and $s_1$ are consecutive.
                Call the maximal sub-word connecting them $\bs{W}$.
                We now have that the vertices $u,v$ and $s_1$ are pairwise consecutive to each other.
                Since $\{u,v,s_1\}$ is a strict subset of $V(R)$, $U$ cannot be a a \SOET\ which is a contradiction. 

                \item Now assume that $v\neq \hat{s}_1$. By \cref{lem:UNUSEDEDGE_2} we can now conclude that $w$ and $\hat{s}_k$ must be consecutive and that $s'_1$ and $\hat{s}_1\neq v$ must be consecutive. We have to perform one last construction to prove the lemma. This construction is visualized in \cref{fig:11skip_v_is_not_s1}.
                Call the maximal sub-words associated to these vertex pairs $w$ and $\hat{s}_k$ and  $s'_1$ and $\hat{s}_1$ $\bs{W}_k$ and $\bs{W}_0$ respectively. 
                We can again use \cref{lem:UNUSEDEDGE_1} to conclude that $\hat{s}_i$ is consecutive to $\hat{s}_{i+1}$ for all $i \in [k-1]$.
                Call the maximal sub-words associated to the vertices $\hat{s}_i$ and $\hat{s}_{i+1}$ $\bs{W}_i$.
                This implies that $U$ must traverse the vertices $s'_1,\hat{s}_1, \ldots, \hat{s}_k, w,s'_1, \hat{s}_1, \ldots, \hat{s}_k, w$ in order.
                Since by construction $\{s'_1,s_w, \hat{s}_1,\ldots, \hat{s_k}\} \neq V(R)$ we have that $U$ is not a valid SOET on $V(R)$ (this can be seen by noting that e.g. $T_{u}$ already contains a true skip, hence $u$ can not be part of the set).
            \end{enumerate}
            Hence in both cases we arrive at a contradiction. This means by contradiction that $x\neq u$ and thus that the sub-trail described by $\bs{Z}$ makes a true skip at $T_u$. \\

        {\bf \underline{Details of step  3.3}}

        Now similarly to Step 3.1 consider the edge connecting $T_{s_1}$ and $T_v$. We can again argue that there must exist a maximal sub-word $\bs{Z}'$ associated to the vertex $s_1$ and some other vertex $x'$ that describes a sub-trail that traverses this edge. By the same argument as Step 3.2 we can conclude that $x'\neq v$ and thus that $\bs{Z}'$ describes a sub-trail making a true skip at $T_v$.\\

        We can make the same argument for the vertex $s'_1$ establishing the existence of maximal sub-words $\hat{Z},\hat{Z}'$ that describe sub-trails making true skips at $T_u$ and $T_v$ respectively.\\

        {\bf \underline{Details of step 3.4}}

        Note that the sub-trails described by $\bs{Z}$ and $\bs{\hat{Z}}$ make true skips at the triangle subgraph $T_u$. Since $T_u$, by \cref{lem:single_skip} can only contain a single true skip and since $\bs{Z},\bs{\hat{Z}}$ are maximal sub-words we must have that $\bs{Z}\sim \bs{\hat{Z}}$ and thus that the vertices $s_1$ and $s'_1$ are consecutive with respect to the SOET $U$.  However since we must also conclude that $\bs{Z}'\sim \bs{\hat{Z}'}$ we have that $s_1$ and $s'_1$ are consecutive and have two maximal sub-words associated to them. Since there are no further constraints on $U$ imposed by the the fact that the sub-words $\bs{X},\bs{X'}$ associated to the vertices $v,u$ make true skips. Therefore SOETs with this type of behavior are in fact allowed. These SOETs can also be found explicitly, as can be sen in \cref{fig:11skip}. If a SOET $U$ has this type of behavior ($u$ and $v$ are not adjacent in $R$ but are consecutive in the SOET $U$ on $\Lambda(R)$) we say that the SOET $U$ has a \emph{valid} $11$-skip. Next we show that if a SOET $U$ has a valid $11$-skip, and thus is not a HAMSOET, it can always be turned into a HAMSOET by applying a fixed set of local complementation. 

    \end{addmargin}

    {\bf \underline{Details of step 4}}\\

    We now show that a \SOET~with valid $11$-skips can be turned into a \HAMSOET.
    The two possibilities for a \SOET~with a valid $11$-skip are shown in \cref{fig:11skip}. Note that these possibilities have the same `local' structure, the only difference is how the rest of the \SOET~$U$ is connected to the valid $11$-skip.  
    We first show the procedure that should be applied if the 11-skip is of the form in \cref{fig:11skip_1}.
    \begin{figure}[H]
        \centering
        \begin{subfigure}{0.4\textwidth}
            \includegraphics[width=1\textwidth]{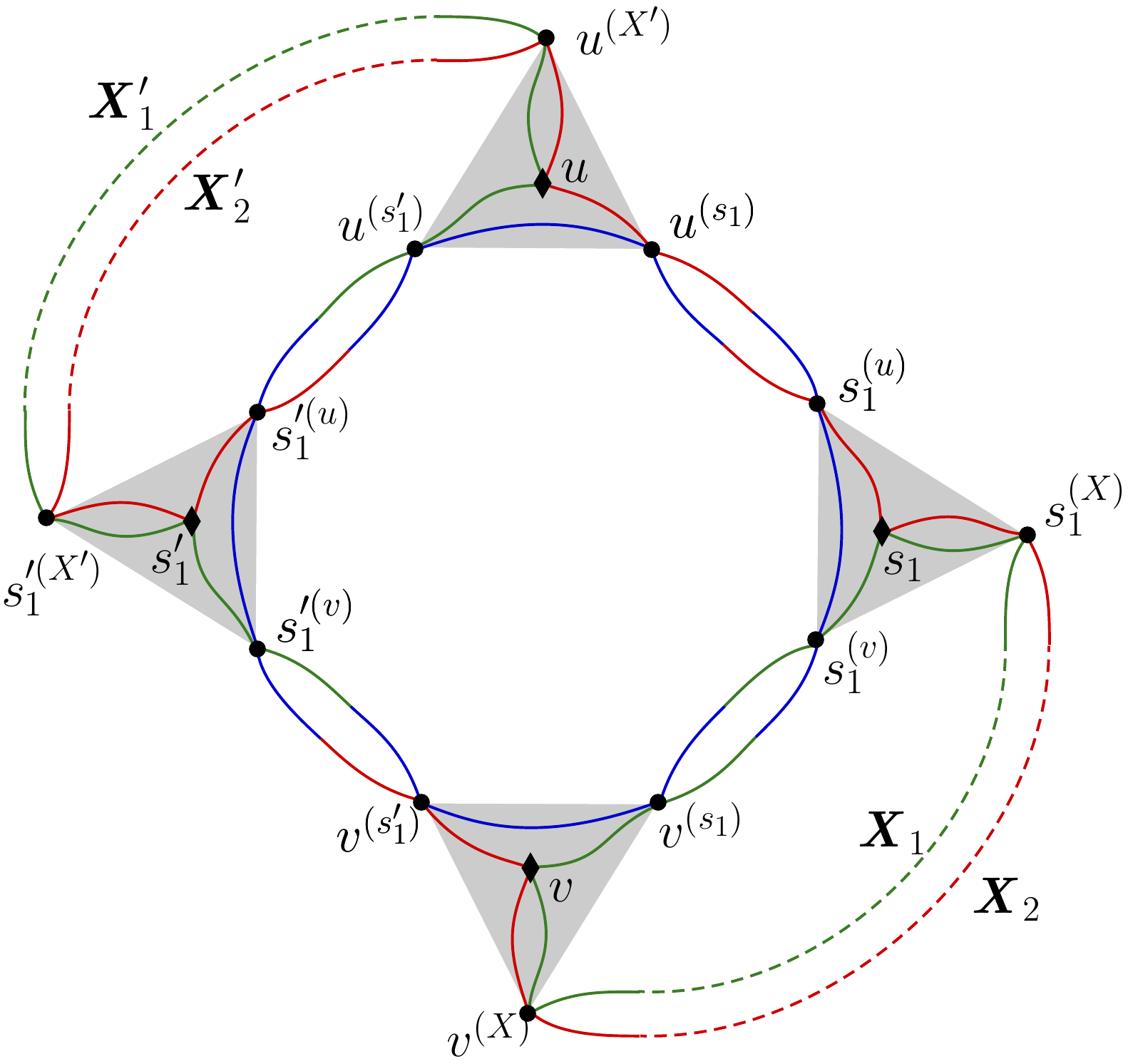}
            \caption{}
            \label{fig:11skip_1}
        \end{subfigure}
        \hspace{1cm}
        \begin{subfigure}{0.4\textwidth}
            \includegraphics[width=1\textwidth]{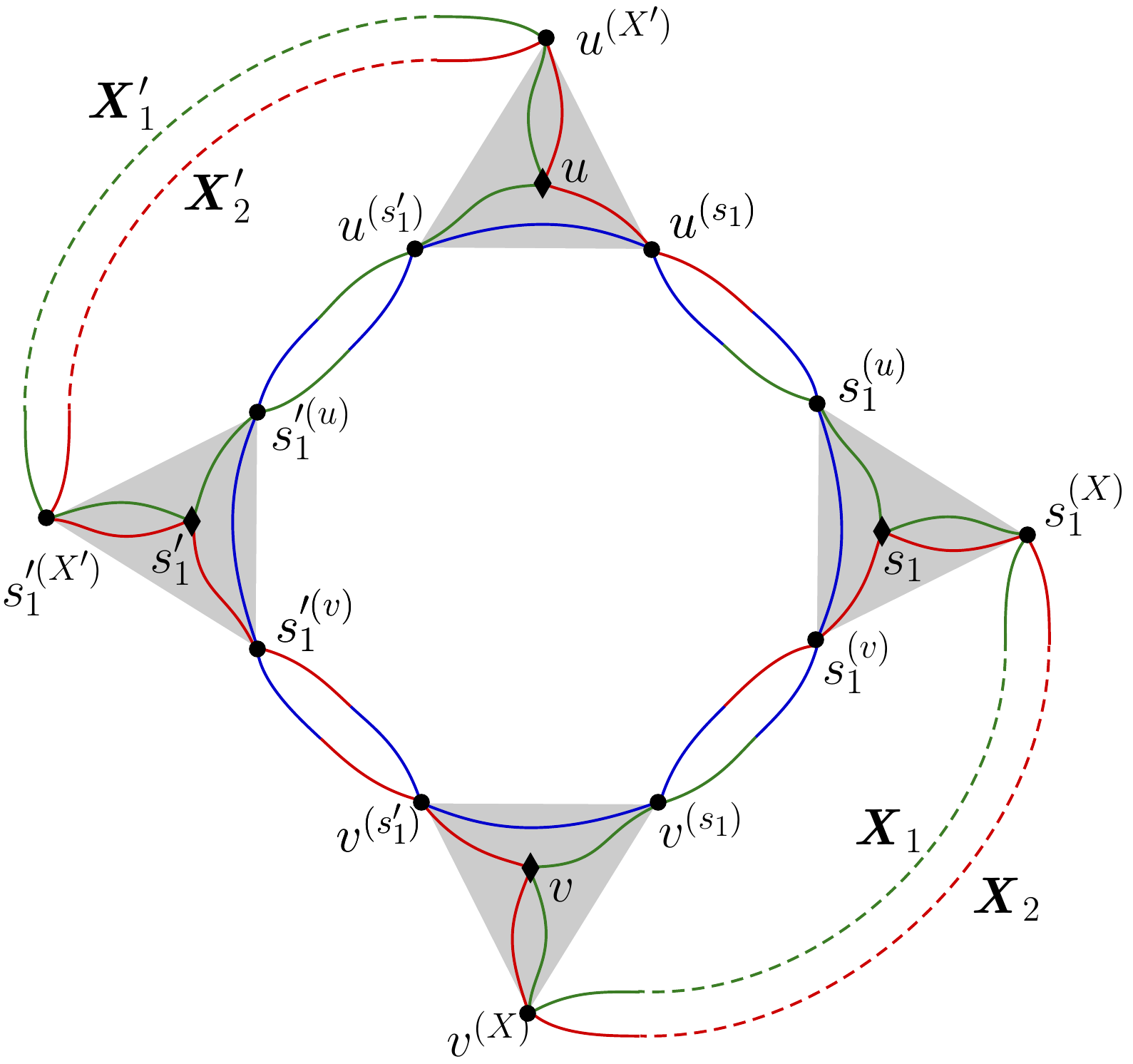}
            \caption{}
            \label{fig:11skip_2}
        \end{subfigure}
        \caption{This figure is a visual aid for Step 4 of \cref{lem:SOET2HAMSOET}.
            A simplified visualization of a triangular-expansion $\Lambda(R)$ of a 3-regular graph $R$ is used.
            In the figure triangle subgraphs are shown in gray with only their outer vertices (circles) and the vertices in $V(R)$ (diamonds) shown.
            Also shown are sub-trail of a SOET $U$ labeled by the maximal sub-words that describe them.
            These sub-trails always begin and end at a vertex in $V(R)$ but their path inside triangle subgraphs is not shown explicitly.
            Dashed lines are used to indicate that the sub-trails also traverse unspecified further parts of the graph $\Lambda(R)$.
            \Cref{fig:11skip_1,fig:11skip_2} show the two valid 11-skips.
            The sub-word $\bs{X}_1$ describes some sub-trail of $U$ connecting $T_v$ and $T_{s_1}$.
            With a slight abuse of notation we denote the endpoints of this trail by $v^{(X)}$ and $s^{(X)}_1$.
            Similarly for $\bs{X}_2$, $\bs{X}'_1$ and $\bs{X}'_2$.
        }
        \label{fig:11skip}
    \end{figure}

    The SOET $U$ in \cref{fig:11skip_1} is of the form
    \begin{equation}
        m(U)=\bs{U}_1 v^{(s_1)} \bs{U}_2 v^{(s_1)} \bs{U}_3 s'^{(u)}_1 \bs{U}_4 s'^{(u)}_1 \bs{U}_5
    \end{equation}
    where
    \begin{align}
        \bs{U}_1 &= u u^{(s_1)} s^{(u)}_1 s^{(v)}_1 \\
        \bs{U}_2 &= v v^{(X)} \bs{X_1} s^{(X)}_1 s_1 s^{(v)}_1 \\
        \bs{U}_3 &= v^{(s'_1)} s'^{(v)}_1 s'_1 s'^{(X')}_1 \bs{X'_1} u^{(X')} u u^{(s'_1)} \\
        \bs{U}_4 &= s'^{(v)}_1 v^{(s'_1)} v v^{(X)} \bs{X}_2 s^{(X)}_1 s_1 s^{(u)}_1 u^{(s_1)} u^{(s'_1)} \\
        \bs{U}_5 &= s'_1 s'^{(X')}_1 \bs{X_2'} u^{(X')}
    \end{align}
    where $\bs{X_1}$ is the word associated to the sub-trail connecting $v^{(X)}$ and $s_1^{(X)}$ as seen in \cref{fig:11skip_1} and similarly for $\bs{X_1'},\bs{X_2},\bs{X_2'}$
    Since our goal is to make this a \HAMSOET~we need to have that pairs of vertices in $V({R})$ are only consecutive w.r.t. $U$ if they are adjacent in ${R}$.
    This can be done by applying $\bar{\tau}$-operations to $U$ at $v^{(s_1)}$ and $s'^{(u)}_1$.
    The Eulerian tour $U'$ after these operations will be described by
    \begin{equation}
        m(U')=m\left(\bar{\tau}_{\left(v^{(s_1)},s'^{(u)}_1\right)}(U)\right)=\bs{U}_1 v^{(s_1)} \widetilde{\bs{U}_2} v^{(s_1)} \bs{U}_3 s'^{(u)}_1 \widetilde{\bs{U}_4} s'^{(u)}_1 \bs{U}_5
    \end{equation}
    where the over-set tilde indicates the mirror-inverting of a sub-word.
    Note that neither $(u,v)$ or $(s_1,s'_1)$ are consecutive anymore, but instead $(u,s_1)$ and $(v,s'_1)$ are now consecutive w.r.t. $U'$.
    To make sure that this procedure works we need to check two things: (1) $U'$ is a SOET and (2) there are no additional consecutive pairs of vertices in $U'$ that are not adjacent in $R$.
    To do this, lets look at the order $U$ and $U'$ traverse the vertices in $V(R)$, i.e. we will look at $m(U)[V(R)]$ and $m(U')[V(R)]$.
    Since $U$ is a SOET we must have that $\bs{X}_1[V(R)]=\bs{X}_2[V(R)]$ and similarly $\bs{X'}_1[V(R)]=\bs{X}'_2[V(R)]$.
    Lets denote these words by $\bs{X}_V=\bs{X}_1[V(R)]$  and $\bs{X}'_V=\bs{X'}_1[V(R)]$.
    We then have that the double occurrence word of $m(U)$ induced by $V(R)$ is
    \begin{equation}
        m(U)[V(R)]=uv\bs{X}_Vs_1s'_1\bs{X}'_Vuv\bs{X}_Vs_1s'_1\bs{X}'_V
    \end{equation}
    and similarly for $U'$ we have
    \begin{equation}
        m(U')[V(R)]=us_1\widetilde{\bs{X}_V}vs'_1\bs{X}'_Vus_1\widetilde{\bs{X}_V}vs'_1\bs{X}'_V
    \end{equation}
    It is therefore clear that the Eulerian tour $U'$ is a \SOET.
    Furthermore the only consecutive pairs of vertices in $U'$ which where not consecutive in $U$ are $(u,s_1)$ and $(v,s'_1)$.
    Since $(u,s_1)$ and $(v,s'_1)$ are edges of $R$ we see that we can iteratively apply this procedure to any valid 11-skip as in \cref{fig:11skip_1} and turn the \SOET~into a \HAMSOET.
    Similarly the \SOET~in \cref{fig:11skip_2} can be turned into a \HAMSOET~by applying $\tau$-operations to the vertices $s^{(u)}_1$ and $v^{(s'_1)}$.
    One can explicitly check this by applying the operations to $U$ in \cref{fig:11skip_2} which is given by
    \begin{equation}
        m(U)=\bs{U_1} s^{(u)}_1 \bs{U_2} s^{(u)}_1 \bs{U_3} v^{(s'_1)} \bs{U_4} v^{(s'_1)} \bs{U_5}
    \end{equation}
    where
    \begin{align}
        \bs{U}_1 &= u u^{(s_1)} \\
        \bs{U}_2 &= s^{(v)}_1 v^{(s_1)} v v^{(X)} \bs{X}_1 s^{(X)}_1 s_1 \\
        \bs{U}_3 &= u^{(s_1)} u^{(s'_1)} s'^{(u)}_1 s'_1 s'^{(X')}_1 \bs{X}'_1 u^{(X')} u u^{(s'_1)} s'^{(u)}_1 s'^{(v)}_1 \\
        \bs{U}_4 &= v v^{(X)} \bs{X}_2 s^{(X)}_1 s_1 s^{(v)}_1 v^{(s_1)} \\
        \bs{U}_5 &= s'^{(v)}_1 s'_1 s'^{(X')}_1 \bs{X}'_2
    \end{align}
    with everything defined similarly to the case of \cref{fig:11skip_2}. Going through a similar argument as above we can show that we can also turn the \SOET~$U$ into a \HAMSOET. This completes the lemma.
\end{proof}

\begin{lem}\label{lem:UNUSEDEDGE_1}
    Let $R$ be a 3-regular graph and $\Lambda(R)$ be its triangular-expansion.
    Also let $u,v$ be adjacent vertices on $R$.
    Let $U$ be a \SOET ~on $\Lambda(R)$ with respect to $V(R)$.
    Let $\bs{X}$ be a maximal sub-word  of $m(U)$ not associated to $u$ and/or $v$ containing $u^{(v)}v^{(u)}$ and describing a sub-trail that makes true skips at $T_u$ and $T_v$.
    Then $u$ and $v$ are consecutive in $U$ and moreover $m(U)$ contains a sub-word of the form
    \begin{equation}
        u \bs{Z}_1 u^{(v)}v^{(u)}\bs{Z}_2 v,\qquad \bs{Z}_1\subset V(T_u),\,  \bs{Z}_2\subset V(T_v),
    \end{equation}
\end{lem}
\begin{proof}
    The situation described in the lemma is described graphically in~\cref{fig:lem_unusededge_1}.
    Because $U$ is a SOET the sub-word $v^{(u)}u^{(v)}$ must be contained exactly twice in $m(U)$.
    Note that at most one of these instances can be contained in the maximal sub-word $\bs{X}$.
    The other instance must be contained in a different maximal sub-word $Z$.
    This maximal sub-word will be associated to two vertices $w_1,w_2$.
    Note that since $v^{(u)}u^{(v)}\in Z$  and $v^{(u)}u^{(v)} \in X$ either we must have that $w_1=u$ or that the sub-trail described by the maximal sub-word $\bs{Z}$ makes a true skip at $T_u$.
    Since $T_u$ already contains a true skip (made by the sub-trail described by $\bs{X}$) we must have that $w_1=u$.
    We can make the same argument for the vertex $v$.
    This means the maximal sub-word $\bs{Z}$ must be associated to $u$ and $v$ and hence that $u,v$ must be consecutive in $U$ and moreover we have that  

    \begin{equation}
         \bs{Z} =  \bs{Z}_1 u^{(v)}v^{(u)}\bs{Z}_2 ,\qquad \bs{Z}_1\subset V(T_u)\setminus\{u\},\;  \bs{Z}_2\subset V(T_v)\setminus\{v\}.
    \end{equation}
\end{proof}

\begin{figure}[H]
            \centering
            \includegraphics[width=0.4\textwidth]{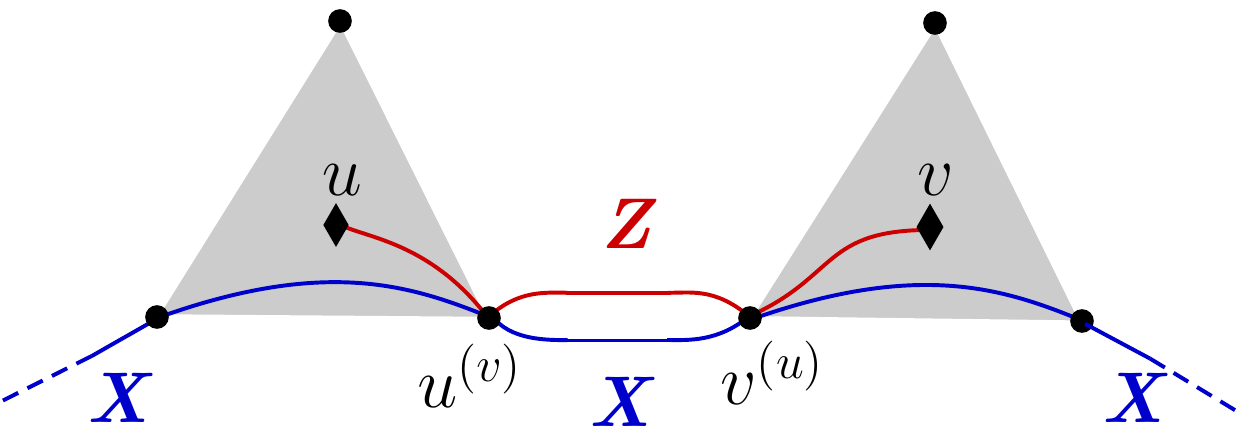}
            \caption{This figure is a graphical aid for \cref{lem:UNUSEDEDGE_1}. Shown are two triangle subgraphs $T_u$, $T_v$ (gray triangles) with outer vertices (circles) and the vertices $u$ and $u$ (diamonds) shown. The lemma starts from assuming the existence of the sub-trail described by the maximal sub-word $\bs{X}$, labeled as such in the figure. From this starting point the existence of the maximal sub-word $\bs{Z}$ associated to the vertices $u$ and $v$ is derived.}
            \label{fig:lem_unusededge_1}
\end{figure}

\begin{lem}\label{lem:UNUSEDEDGE_2}
    Let $R$ be a 3-regular graph and $\Lambda(R)$ be its triangular-expansion.
    Also let $u,v$ be adjacent vertices on $R$.
    Also take $x_1,x_2$ to be the vertices adjacent to $u$ in $R$ such that $x_1 \neq v,x_2 \neq v$.
    Let $U$ be a SOET on $\Lambda(R)$ with respect to $V(R)$.
    Let $\bs{Y}$ be a maximal sub-word of $m(U)$ not associated to $u$ and/or $v$ containing $u^{(x_1)}x_1^{(u)}$ and  $u^{(x_2)}x_2^{(u)}$ and describing a sub-trail making a true skip at $T_u$.
    Also let $\bs{X}$ be a maximal sub-word associated to $u$ and a vertex $x_3\neq v$ that describes a sub-trail making a true skip at $v$.
    Then $u,v$ are consecutive and moreover $m(U)$ contains a sub-word of the form
    \begin{equation}
        u \bs{Z}_1 u^{(v)}v^{(u)}\bs{Z}_2 v,\qquad \bs{Z}_1\subset V(T_u),\,  \bs{Z}_2\subset V(T_v),
    \end{equation}
\end{lem}
\begin{proof}
    The situation described in the lemma is described graphically in~\cref{fig:lem_unusededge_2}.
    Because $U$ is a SOET the sub-word $v^{(u)}u^{(v)}$ must be contained exactly twice in $m(U)$.
    Note that at most one of these instances can be contained in the maximal sub-word $\bs{Y}$ and none can be contained in the maximal sub-word $\bs{X}$.
    This means there must be a maximal sub-word $\bs{Z}$ of $m(U)$ (different from $\bs{X}$ and $\bs{X}$) containing $v^{(u)}u^{(v)}$.
    This maximal sub-word must again be associated with two vertices $x,\hat{x}$.
    If these vertices are not $u,v$ then the sub-trail described by $\bs{Z}$ must make true skips at $T_u,T_v$ or both.
    Since both of these triangle subgraphs already contain true skips this is not possible and hence $\bs{Z}$ must be associated to $u$ and $v$ which means they are consecutive.
    Moreover, by construction of $\bs{Z}$ we have 
    \begin{equation}
         \bs{Z} =  \bs{Z}_1 u^{(v)}v^{(u)}\bs{Z}_2,\qquad \bs{Z}_1\subset V(T_u)\setminus\{u\},\;  \bs{Z}_2\subset V(T_v)\setminus\{v\}.
    \end{equation}
\end{proof}

\begin{figure}[H]
            \centering
            \includegraphics[width=0.4\textwidth]{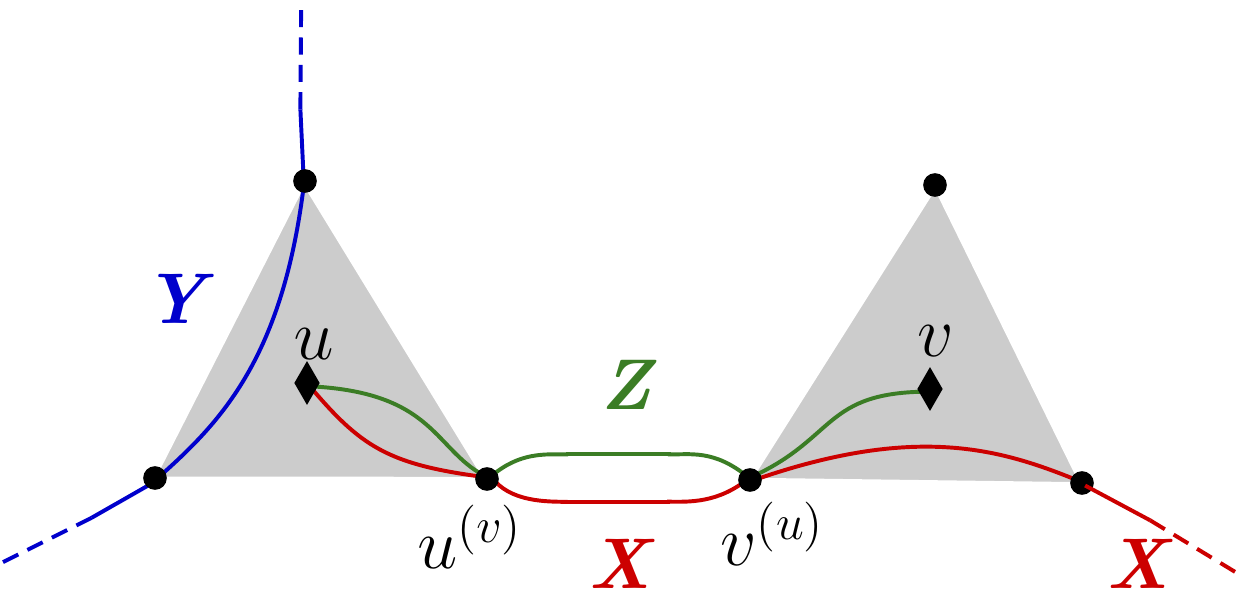}
            \caption{This figure is a graphical aid for \cref{lem:UNUSEDEDGE_2}. Shown are two triangle subgraphs $T_u$, $T_v$ (gray triangles) with outer vertices (circles) and the vertices $u$ and $v$ (diamonds) shown. The lemma starts from assuming the existence of the sub-trails described by the maximal sub-words $\bs{Y}$ and $\bs{X}$, labeled as such in the figure. From this starting point the existence of the maximal sub-word $\bs{Z}$ associated to the vertices $u$ and $v$ is derived. }
            \label{fig:lem_unusededge_2}
\end{figure}

\section{Algorithms}\label{eq:algorithms}
This section is concerned with providing algorithms for various versions of the decision problems \SVM\ with input $(G,V')$ and \VM\ with input $(G,G')$.
We begin by describing an efficient algorithm for \SVM\ whenever the input graph $G$ is distance-hereditary.
We prove that this algorithm always terminates and gives correct results.
We also analyze its runtime and show that it is $\mathcal{O}(\abs{V'}\abs{V(G)}^3)$.
Next we describe an algorithm for \SVM\ whenever the input graph $G$ is a circle graph.
We prove that this algorithm is fixed-parameter tractable in the size of the input vertex-set $V'$.
Finally we prove for the decision problem \VM\ with arbitrary connected input graphs $G$ and $G'$, that whenever $|V(G')|\leq3$ and $V(G')\subset V(G)$ we have that  $G'<G$ and provide an efficient algorithm for finding the sequence of local complementations and vertex-deletions that takes $G$ to $G'$.

\subsection{Star graph as vertex-minor of a distance-hereditary graph}\label{sec:DH_alg}
In this section we present an efficient algorithm for deciding whether a star graph on a given set of vertices $V'$ is a vertex-minor of a given distance-hereditary graph $G$.
Throughout this section we assume that the graph $G$ is distance-hereditary and that $V'$ is a subset of its vertices.
The algorithm presented in this section will return a sequence of vertices $\bs{v}$ in $V(G)$, such that $\tau_{\bs{v}}(G)[V']=S_{V'}$ if such a sequence exists and raise an error-flag otherwise, indicating that $S_{V'}$ is not a vertex-minor of $G$.
We first present the algorithm in \cref{sec:thealg}, analyze its runtime in \cref{sec:runtime} and prove that it is correct in \cref{sec:theproof}.
An implementation in SAGE~\cite{sage} of the algorithm can be found at~\cite{git}.

\subsubsection{The algorithm}\label{sec:thealg}
We first give a rough sketch of the idea behind the algorithm.
Remember that the task of the algorithm is to find a sequence of local complementations $\tau_{\bs{v}}$ such that the induced subgraph of $\tau_{\bs{v}}(G)$ on the vertices $V'$ is a star graph.


The algorithm starts by choosing a one vertex $c$ in $V'$ which will become the center of the star graph on $V'$. It then proceeds by different picking vertices $v \in V'$  and making them adjacent to $c$ by performing local complementations. After every vertex that is made adjacent the algorithm will check if the induced subgraph on the neighborhood of $c$ is a star graph. If it is not it will attempt to turn it into a star graph by local complementations. If it fails at doing so it will raise an error and if it succeeds it will pick another vertex in $V'$ and repeat the procedure until all vertices in $V'$ are in the neighborhood of $c$. We will often call this process of making a vertex $v$ adjacent to $c$ 'adding' the vertex $v$ to the star graph. To understand when the algorithm might fail we now zoom in on the situation where all but one vertex  of $V'$ has been added to the neighborhood of $c$. Let us call this vertex $f$. At this point in the algorithm the induced subgraph $G[V'\{f\}]$ is already a star graph (be previous successful iterations of this procedure).\\

The task is now to turn $G[V']$ into a star graph by making $f$ adjacent to the center $c$ of $G[V'\setminus \{f\}]$ but to no other vertex of $V'$, and at the same time not change any edges in $G[V'\setminus \{f\}]$.
This will be done in two steps, which are explained further below:
\begin{enumerate}
    \item Make $f$ and $c$ adjacent, without changing any edges in $G[V'\setminus\{f\}]$.
        The star graph $S_{V'}$ is then a subgraph of the graph, but not necessarily an induced subgraph, since $f$ could be also be adjacent to other vertices in $V'$ than $c$.
        We will call these edges between $f$ and vertices in $V'\setminus \{f\}$ \emph{bad edges}.
        This first step is the task of \cref{alg:Sn_vminor_2} below.
        Interestingly, this step always succeeds if the graph is connected, even if the graph is not distance-hereditary.
    \item Remove the bad edges, without changing any other edges between vertices in $V'$.
        The removal of the bad edges is the task of \cref{alg:Sn_vminor_1} below.
        \Cref{alg:Sn_vminor_1} tries to remove the bad edges by checking a few cases.
        Thus, one of the main results of this section is to prove that these cases provide a necessary condition for $S_{V'}$ being a vertex-minor of $G$.
\end{enumerate}

We will now describe the two main steps above of the algorithm in more detail.
Let's denote the vertices as above and furthermore the current leaves in $G[V'\setminus\{f\}]$ as $V'\setminus \{c,f\}=\{l_1,\dots,l_k\}$.

\noindent\underline{Details of step 1:}
\begin{addmargin}[1em]{0em}
    The vertices $c$ and $f$ are made adjacent by performing local complementations along the shortest path $P$ from $c$ and $f$, see \cref{fig:c_leaf1}.
    The operations along the path $P$ will in fact be either pivots, i.e. $\rho_{(u,v)}=\tau_v\circ\tau_u\circ\tau_v$, or single local complementations depending on the situation.
    The reason for this is to not remove edges between $c$ and the $l_i$'s.
    The details of the operations along the path $P$ are given in \cref{alg:Sn_vminor_2} together with the proof in section \cref{sec:theproof}.
    \begin{figure}[H]
        \centering
        \begin{subfigure}{0.5\textwidth}
            \includegraphics[width=\textwidth]{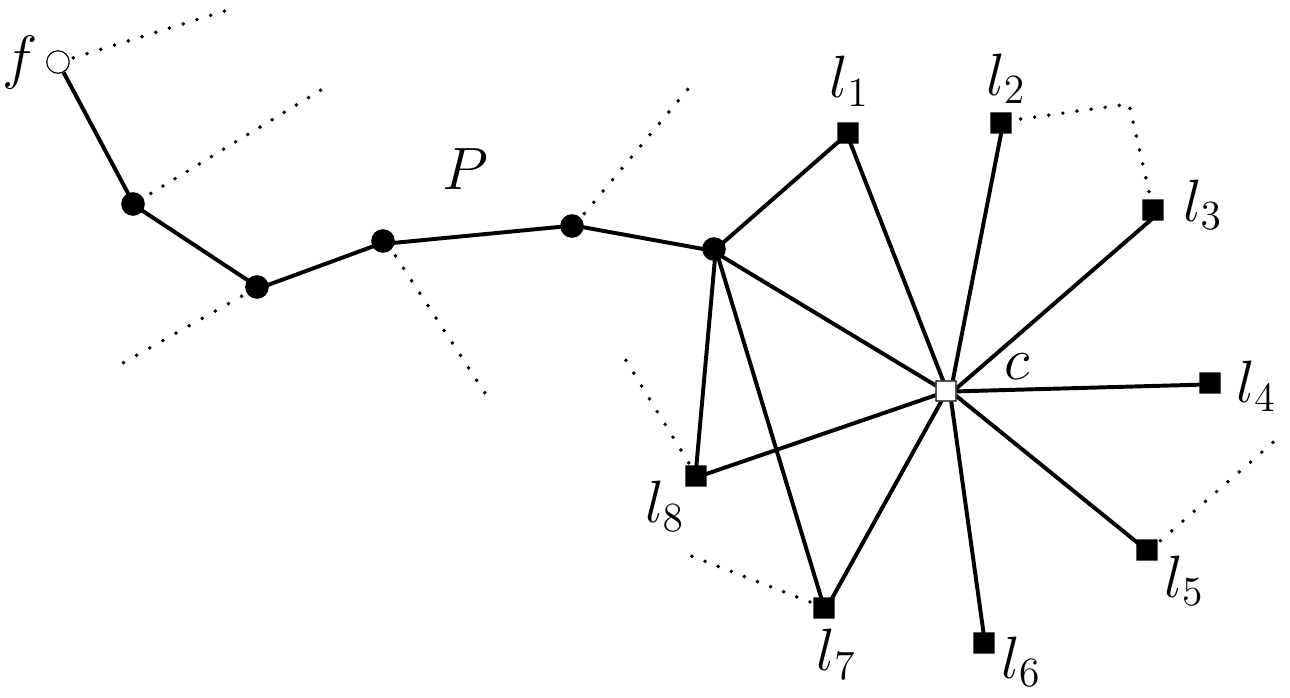}
            \caption{}
            \label{fig:c_leaf1}
        \end{subfigure}
        ~
        \begin{subfigure}{0.3\textwidth}
            \includegraphics[width=\textwidth]{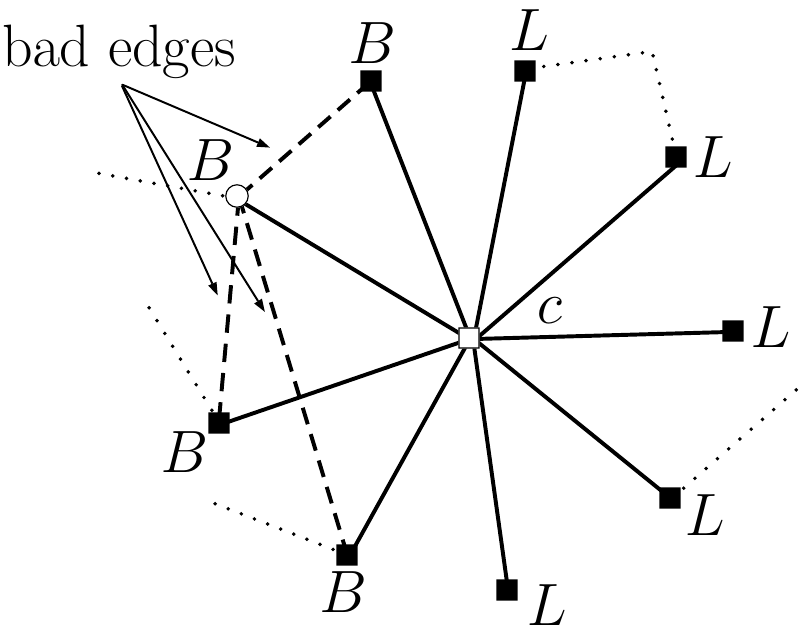}
            \caption{}
            \label{fig:c_leaf2}
        \end{subfigure}
        \caption{Visualization of how a vertex $f$ is added to the neighborhood of $c$. 
            The original graph is shown in \cref{fig:c_leaf1} where the vertices $\{c\}\cup \{l_i\}_i$ (squares) already form a star graph and the dotted lines are arbitrary edges to the rest of the graph.
    The new vertex $f$ (white circle) is made adjacent to the center $c$ (white square) by performing pivots and local complementations along the shortest path $P$ (black circles).
    After the pivots and local complementations the new vertex $f$ is adjacent to the center $c$ of the star graph but also to some leaves $B$ by \textit{bad edges} (dashed lines), see \cref{fig:c_leaf2}.}
        \label{fig:c_leaf}
    \end{figure}
\end{addmargin}
As mentioned above, by making $f$ and $c$ adjacent we might have also added bad edges between $f$ and some of the vertices $\{l_i\}_i$, see \cref{fig:c_leaf2}.
Let's denote the set of vertices which are incident to a bad edge by $B$ and the set of vertices not incident to a bad edge, apart from $c$, by $L=(V'\setminus\{c\})\setminus B$.
We call such a graph as the induced subgraph on  $V'$ a star-star graph, see \cref{def:SS}.

Next, we must remove the bad edges in order to turn $G[V']$ into a star graph.
Let $G$ now be the graph with $f$ and $c$ adjacent but with possibly some bad edges.

\noindent\underline{Details of step 2:}
\begin{addmargin}[1em]{0em}
    In this step we will remove the bad edges, if we can.
    A situation where bad edges can be removed, as we will show, is when there exists a vertex $u\notin V'$, which is adjacent to all vertices in $B$ but not to any vertex in $L$. The existence of such a vertex $u$ is thus a sufficient condition for the removal of bad edges. 
    When $G$ is a distance hereditary graph, it turns out that this condition is also necessary, that is if no such vertex $u$ exists, then the star graph on $V'$ is not a vertex-minor of $G$, and we can stop the algorithm. This is shown in detail in \cref{sec:theproof}.
    For this statement to hold $L$ cannot be empty, but this can always be achieved by performing a local complementation at $c$ first if needed, which is done in line 14 in \cref{alg:Sn_vminor_1}.
    Assume that there indeed exist such a vertex $u$, i.e.
    \begin{equation}
        (L\neq\emptyset)\;\land\;(u\notin V')\;\land\;(B\subseteq N_u)\;\land\;(L\cap N_u=\emptyset)
    \end{equation}
    see \cref{fig:good_com}.
    Now $u$ can be adjacent to $c$ or not.
    Let's consider these cases separately:
    \begin{itemize}
        \item Case 1 $u$ and $c$ are not adjacent:

            Remember that $f$ is the center of the induced star graph $G[B]$.
            If a local complementation is performed at $u$, the bad edges are removed but new ones will be created between the vertices in $B\setminus\{f\}$.
            These new bad edges will then form a complete graph on $B\setminus\{f\}$ and we call such a graph on the vertices $V'\setminus \{f\}$ a complete-star graph, see \cref{def:KS}.


            \begin{figure}[H]
                \centering
                \begin{subfigure}{0.3\textwidth}
                    \includegraphics[width=\textwidth]{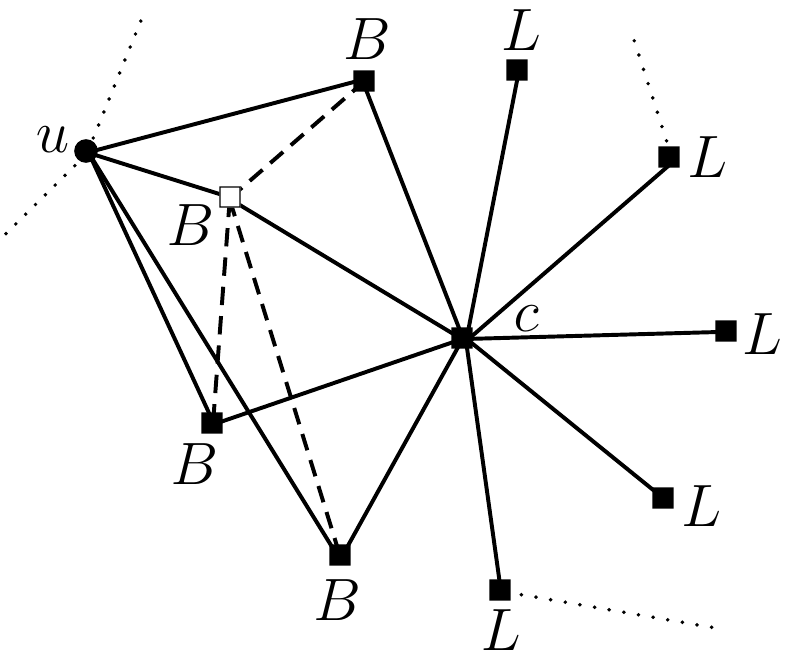}
                    \caption{}
                    \label{fig:good_com1}
                \end{subfigure}
                ~
                \begin{subfigure}{0.3\textwidth}
                    \includegraphics[width=\textwidth]{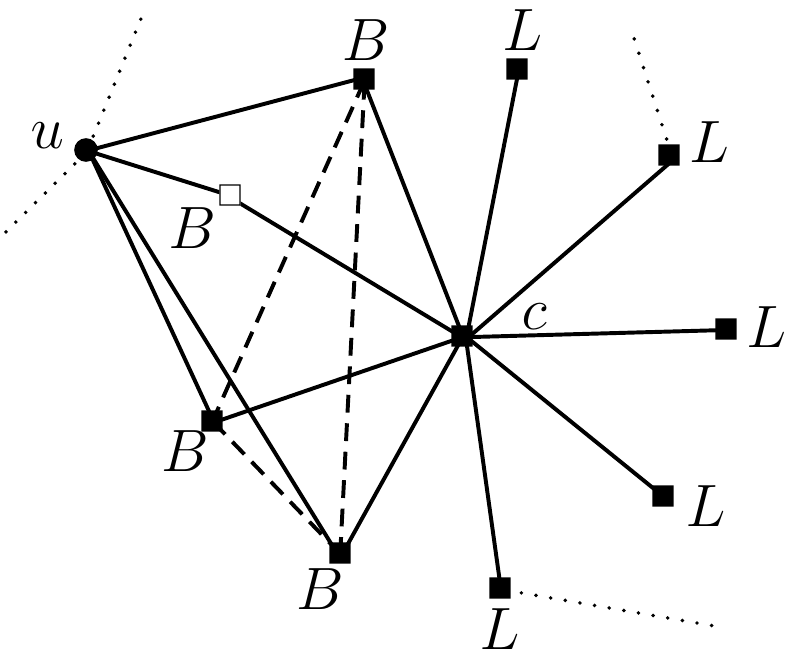}
                    \caption{}
                    \label{fig:good_com2}
                \end{subfigure}
                \caption{Visualization of how \textit{bad edges} are removed.
                    The original graph is shown in \cref{fig:good_com1} where the vertices $B\cup L\cup \{c\}$ (squares) are the desired vertices of the star graph, the dashed lines are the bad edges and vertex $u$ (black circle) is as in \cref{eq:cond}.
                    \Cref{fig:good_com2} shows the graph after performing a local complementation on $u$ which produces a new leaf (white square) and makes the bad edges form a complete graph. This complete graph of bad edges can then be removed by finding a vertex $u'$ that is adjacent to all vertices in $B$ (and to none in $L$) and performing a local complementation at $u'$.
                }
                \label{fig:good_com}
            \end{figure}

            Performing the same step again, i.e. doing a local complementation at another vertex adjacent to all vertices in $B\setminus\{f\}$, will remove all bad edges.
            We have then produced the star graph on $V'$ in two steps.

        \item Case 2: $u$ and $c$ are adjacent:

            In this case, if a local complementation is performed at $u$, some edges between $c$ and vertices in $L$ will be removed, which is not desired.
            We can solve this by finding another vertex $h$ adjacent to both $u$ and $c$ but not to any other vertex in $V'$, by which we can remove the edge $(u,c)$, see \cref{fig:h_helper}.
            In the following section we show that if there is no vertex $h$ of this form, the star graph is not a vertex-minor of $G$ and we can stop the algorithm.
    \end{itemize}
        To prove that the algorithm is correct we need to show that cases checked by \cref{alg:Sn_vminor_2} to remove the bad edges actually provides a necessary condition for $S_{V'}$ being a vertex-minor of $G$.
        To be precise, we will show that a necessary\footnote{This condition is not sufficient in itself, however \cref{thm:conds} provide a necessary and sufficient condition.} condition for the star graph on $V'$ being a vertex-minor of $G$ is
        \begin{equation}\label{eq:cond}
            \mathcal{P}(B,L,c)=\exists u\in V\setminus V':\qty(B\subseteq N_u\,\land\,L\cap N_u=\emptyset\,\land\,\qty((u,c)\notin E\,\lor\,\exists h:\qty(h\in N_u\cap N_c\setminus\bigcup_{x\in V'\setminus\{c\}}\!N_x))),
        \end{equation}
        where $V'=B\cup L\cup\{c\}$ and $L$ is assumed to be nonempty.
        It is important to note here that this condition is only valid if the graph is in the correct form, i.e. the induced subgraph on $V'$ form a star-star graph or a complete-star graph.

        \begin{figure}[H]
            \centering
            \includegraphics[width=0.3\textwidth]{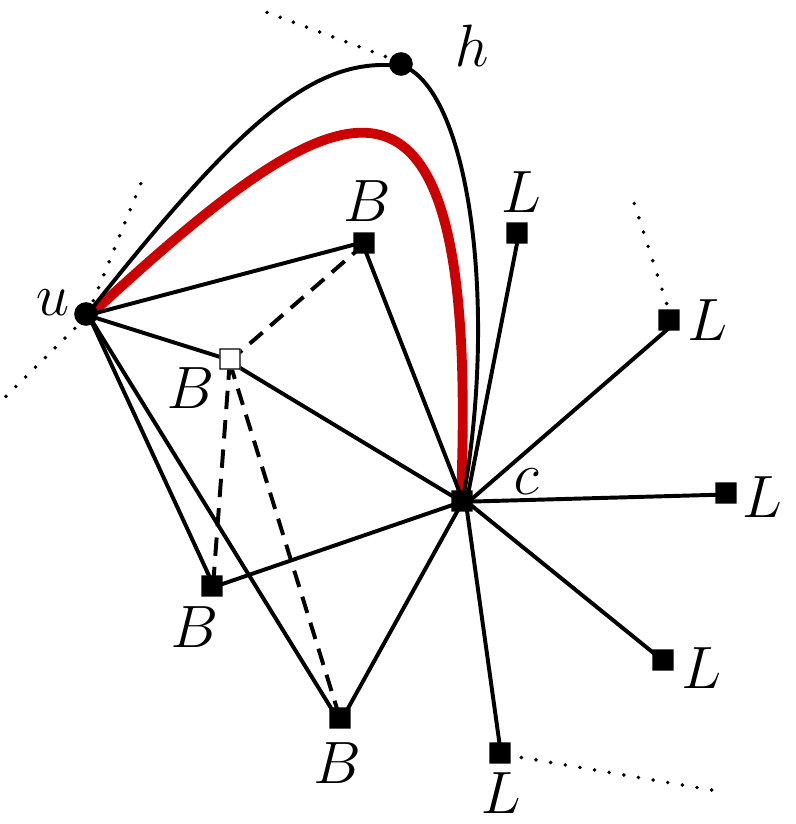}
            \caption{A visualization of the case where a vertex $u$ used to remove the bad edges is also adjacent to $c$ (red thick edge).
                The vertex $h$ can be used to remove the edge $(u,c)$ by applying a local complementation at $h$.
                Since $h$ is not adjacent to any other vertex in $V'$, no edges in the induced subgraph on $V'$ are changed by this local complementation.}
            \label{fig:h_helper}
        \end{figure}
\end{addmargin}

We formally state that \cref{eq:cond} is a necessary condition for the star graph on $V'$ being a vertex-minor of $G$ in \cref{thm:conds}, which we prove in \cref{sec:theproof}.
The theorem uses the notion of star-star and complete-star graphs which we formally define as:

\begin{mydef}[Star-star graph]\label{def:SS}
    A graph $G=(V,E)$ is called a star-star graph if there exists two subsets $B$ and $L$ and a vertex $c$, such that $\{B,L,\{c\}\}$ form a partition of $V$ and $\abs{B}>1$.
    Furthermore $N_l=\{c\}\;\forall l\in L$ and $c\in N_b\;\forall b\in B$.
    Finally $G[B]=S_B$.
    Such a graph is denoted $SS_{(B,L,c)}$.
\end{mydef}

\begin{mydef}[Complete-star graph]\label{def:KS}
    A graph $G=(V,E)$ is called a complete-star graph if there exists two subsets $B$ and $L$ and a vertex $c$, such that $\{B,L,\{c\}\}$ form a partition of $V$ and $\abs{B}>1$.
    Furthermore $N_l=\{c\},\;\forall l\in L$ and $c\in N_b,\;\forall b\in B$.
    Finally $G[B]=K_B$.
    Such a graph is denoted $KS_{(B,L,c)}$.
    Note that if $\abs{B}=2$, $G$ is also a star-star graph.
\end{mydef}


\begin{thm}\label{thm:conds}
   Let $G$ be a distance-hereditary graph on the vertices $V$ and let $V'$ be a subset of $V$.
   Furthermore, let $V'=B\cup L\cup\{c\}$ be a partition of $V'$ and let $S_{V'}$ be a star graph on the vertices $V'$.
   Then the following statements hold
   \begin{itemize}
       \item If $G[V']=SS_{(B,L,c)}$ is a star-star graph and $\abs{B}=2$, then
           \begin{equation}\label{eq:cond_small_SS}
              \mathcal{P}(B,L,c)\quad\Leftrightarrow\quad S_{V'}<G.
          \end{equation}
       \item If $G[V']=SS_{(B,L,c)}$ is a star-star graph then
           \begin{equation}\label{eq:cond_not_SS}
              \neg\mathcal{P}(B,L,c)\quad\Rightarrow\quad S_{V'}\nless G.
          \end{equation}
      \item If $G[V']=SS_{(B,L,c)}$ is a star-star graph and $f$ is the center of the star graph $G[B]$, then
           \begin{equation}\label{eq:cond_SS}
              \mathcal{P}(B,L,c)\quad\Leftrightarrow\quad KS_{(B\setminus\{f\},L\cup\{f\},c)}<G.
          \end{equation}
       \item If $G[V']=KS_{(B,L,c)}$ is a complete-star graph then
           \begin{equation}\label{eq:cond_KS}
              \mathcal{P}(B,L,c)\quad\Leftrightarrow\quad S_{V'}<G.
          \end{equation}
   \end{itemize}
\end{thm}
\Cref{thm:conds} implicitly gives a necessary and sufficient condition for when $S_{V'}$ is a vertex-minor of $G$, if $G[V']$ is a star-star graph.
More precisely, if the induced subgraph on $V'$ is a star-star graph and $\mathcal{P}(B,L,c)$ is true then we know that local complementations can be performed to turn the induced subgraph on $V'\setminus \{f\}$ into a complete-star graph, see \cref{eq:cond_SS}.
Then, if $\mathcal{P}(B\setminus \{f\},L\cup\{f\},c)$ is again true, then a star graph can be created on $V'$ by performing further local complementations, see \cref{eq:cond_KS}.
If in any of these two steps, $\mathcal{P}(B,L,c)$ or $\mathcal{P}(B\setminus\{f\},L\cup\{f\},c)$ is false, then $S_{V'}$ is not a vertex-minor of $G$, see \cref{eq:cond_not_SS,eq:cond_KS}.
In \cref{sec:theproof} we prove these statements.



\begin{algorithm}[H]
    \caption{Producing $S_{V'}$ from a distance-hereditary graph $G$
            }
            \label{alg:Sn_vminor_1}
    \begin{algorithmic}[1]
        \State \textbf{INPUT}:\hspace{1.2em} A graph $G$ and a subset of vertices $V'\subseteq V(G)$.
        \State \textbf{OUTPUT}: A sequence $\bs{v}$ such that $\tau_{\bs{v}}(G)[V']=S_{V'}$,$\phantom{\textproc{ERROR}}$if $S_{V'}<G$.
        \State \phantom{\textbf{OUTPUT}:} $\textproc{ERROR}$,$\phantom{\text{A sequence $\bs{v}$ such that $\tau_{\bs{v}}(G)[V']=S_{V'}$}}$if $S_{V'}\nless G$.
        \State \hrulefill
        \State
        \If{$\abs{V'}=1$}
            \State \textbf{Return} $()$
            \State \textproc{QUIT}
        \EndIf
        \State Find a $\bs{v}$ such that $\tau_{\bs{v}}(G)$ contain the star graph on $V'$ as a subgraph by calling \cref{alg:Sn_vminor_2}.
        \State Let $c$ be a vertex in $V'$, adjacent to all other in $V'$ (except itself).
        \For{$i$ in $\{0,1\}$} \Comment{Two iterations are always needed if there is more than one bad edge}
            \State Let $B$ be the vertices incident to a bad edge. \Comment{These are the vertices in $\tau_m(G)[V'\setminus\{c\}]$ of degree 1 or higher}
            \State Let $L=V'\setminus (\{c\}\cup B)$.
            \If{$B=\emptyset$} \Comment{If already $S_{V'}$, only for $i=0$}
                \State \textbf{Return} $\bs{v}$
                \State \textproc{QUIT}
            \Else
                \If{$B=V'\setminus\{c\}$} \Comment{I.e. if $L=\emptyset$}
                    \State Set $\bs{v}=\bs{v}\Vert (c)$
                    \State \textproc{BREAK}
                \EndIf
                \State Let $U$ be the set $U=\{u\in V(G)\setminus V':B\subseteq N_u\land L\nsubseteq N_u$\} \Comment{Candidates for the $u$ in \cref{eq:cond}}
                \If{$U=\emptyset$}
                \State \textbf{Raise} \textproc{ERROR}($S_{V'}$ is not a vertex-minor of $G$) \Comment{Actually not needed, only for clarity}
                \EndIf
                \State Set $found$=\textproc{False}
                \For{$u$ \textbf{in} $U$}
                    \If{$(u,c)\notin E(\tau_{\bs{v}}(G))$}
                        \State Set $\bs{v}=\bs{v}\Vert (u)$
                        \State Set $found$=\textproc{True} \Comment{Found a $u$ satisfying \cref{eq:cond}}
                        \State \textproc{Break}
                    \ElsIf{$\exists h:\qty(h\in N_u\cap N_c\setminus\bigcup_{x\in V'\{u,c\}}N_x)$}
                        \State Set $\bs{v}=\bs{v}\Vert (h,u)$
                        \State Set $found$=\textproc{True} \Comment{Found a $u$ and $h$ satisfying \cref{eq:cond}}
                        \State \textproc{Break}
                    \EndIf
                \EndFor
                \If{$\neg found$} \Comment{I.e. condition \cref{eq:cond} is false}
                    \State \textbf{Raise} \textproc{ERROR}($S_{V'}$ is not a vertex-minor of $G$)
                \EndIf
            \EndIf
        \EndFor
        \State \textbf{Return} $\bs{v}$
        \State \textproc{QUIT}
    \end{algorithmic}
\end{algorithm}

\begin{algorithm}[H]
    \caption{Find a $\bs{v}$ such that $\tau_{\bs{v}}(G)$ contain the star graph on $V'$ as a subgraph}\label{alg:Sn_vminor_2}
    \begin{algorithmic}[1]
        \State \textbf{INPUT}:\hspace{1.2em} A graph $G$ and a subset of vertices $V'\subseteq V(G)$.
        \State \textbf{OUTPUT}: A sequence $\bs{v}$ such that $\tau_{\bs{v}}(G)[V']=SS_{(B,L,c)}$, where $(B,L,\{c\})$ is a partition of $V'$.
        \State \hrulefill
        \State
        \State Pick an arbitrary vertex from $V'$ and denote this $f$
        \State Find a $\bs{v}$ such that $\tau_{\bs{v}}(G)[V'\setminus\{f\}]=S_{V'}$ and denote the center $c$ by calling \cref{alg:Sn_vminor_1}
        \State Find a shortest path $P=(p_0=f,p_1,\dots,p_k,p_{k+1}=c)$ between $f$ and $c$.
        \For{$i$ in $(1,\dots,k)$}
            \If{$f$ is adjacent to any vertex in $V'\setminus\{c\}$ in the graph $\tau_{\bs{v}}(G)$}
                \State Pick an arbitrary vertex in $N_{f}(\tau_{\bs{v}}(G))\cap V'\setminus\{c\}$ and denote this $v$
                \State Set $\bs{v}=\bs{v}\Vert v$
            \Else
                \State Set $\bs{v}=\bs{v}\Vert (f,p_i,f)$
            \EndIf
        \EndFor
        \State \textbf{Return} $\bs{v}$
        \State \textproc{QUIT}
    \end{algorithmic}
\end{algorithm}

\subsubsection{Runtime of the algorithm}\label{sec:runtime}
The algorithm described in the previous section checks if a star graph with vertex set $V'$ is a vertex-minor of a distance-hereditary graph $G$.
Here we show that the runtime of this algorithm is $\mathcal{O}(\abs{V'}\abs{V(G)}^3)$.
We will represent subsets of a base-set as unsorted binary lists\footnote{It is possible to represent the sets in different ways, by for example (un)sorted lists containing the vertices as entries. However most reasonable data structures will not affect the total runtime of the algorithm but can reduce the memory used.}, where $1$ indicates that an element in the base-set is in the represented set and $0$ that an element in the base-set is not in the represented set.
This will be the case both for sets of vertices and sets of edges.
The base-set for sets of vertices will be the set of vertices $V(G)$ of the input graph $G$ and the base-set for edges-sets will be $V(G)\times V(G)$.
Thus, we assume that the input graph $G$ is given as an unsorted binary list, of length $\abs{V(G)}^2$, indicating which edges are in $E(G)$.
This allows us to check if an edge $(u,v)$ is in the graph or not in constant time.
Furthermore, we assume that the input-set $V'$ is also represented as an unsorted binary list, of length $\abs{V(G)}$, indicating which of the vertices of $G$ are in $V'$.
We also assume that the size of $\abs{V'}$ is given together with its representation, which allows us to faster create representations of subsets of $V'$.

Sets used internally by the algorithm ($B$, $L$ and $U$) will also be represented as unsorted binary lists together with the size of the sets.
The sizes of the sets will be updated accordingly whenever an element is added.
Note that $B$ and $L$ are subsets of $V'$ and will therefore be represented as unsorted binary lists, of length $\abs{V'}$, indicating which elements of $V'$ are in these sets.
However, $U$ is not a subsets of $V'$ and will therefore be represented by an unsorted binary list of length $\abs{V}$.
Thus, given a vertex $v$, checking if $v$ is in a set of vertices $V$ can be done in constant time and adding a vertex to a set can be done in constant time (flipping the bit at the corresponding position).
Furthermore, iterating over elements in a set can be done in linear time with respect to the base-set, i.e. $\mathcal{O}(\abs{V(G)})$ for $V'$ and $U$ and $\mathcal{O}(\abs{V'})$ for $B$ and $L$.

As described, the full algorithm starts by calling \cref{alg:Sn_vminor_1}, which in turn calls \cref{alg:Sn_vminor_2}, which again calls \cref{alg:Sn_vminor_1} and so on.
We will see that the computation that dominates the runtime is updating the graph $\tau_{\bs{v}}(G)$ whenever $\bs{v}$ is concatenated, as in line 13 of \cref{alg:Sn_vminor_1} and line 6 of \cref{alg:Sn_vminor_2}.
We will assume that both \cref{alg:Sn_vminor_1} and \cref{alg:Sn_vminor_2} have access to a common graph which they can update to $\tau_{\bs{v}}(G)$, whenever $\bs{v}$ is concatenated, to prevent this from being done for the whole sequence $\bs{v}$ every time. Note that $\tau_{\bs{v}}(G)$ takes up the same amount of space regardless of $\bs{v}$. 
Each local complementation in the sequence can be performed in time $\mathcal{O}(\abs{V(G)}^2)$~\cite{Bouchet1991}.
Since \cref{alg:Sn_vminor_1} and \cref{alg:Sn_vminor_2} increase the length of $\bs{v}$ by $\mathcal{O}(1)$ and $\mathcal{O}(\abs{V(G)})$ respectively each call, the runtime to update the graph $\tau_{\bs{v}}(G)$ is $\mathcal{O}(\abs{V(G)}^3)$.
We will now show that all other parts of both \cref{alg:Sn_vminor_1} and \cref{alg:Sn_vminor_2} takes time less than $\mathcal{O}(\abs{V(G)}^3)$, which will imply that the total runtime is $\mathcal{O}(\abs{V'}\abs{V(G)}^3)$ since \cref{alg:Sn_vminor_2} is called $\mathcal{O}(\abs{V'})$ times\footnote{Note that \cref{alg:Sn_vminor_2} calls \cref{alg:Sn_vminor_1} with the set $V'\setminus \{f\}$, thereby reducing the size of $V'$ in each recursive call.}.
Let's start by going through the runtime of \cref{alg:Sn_vminor_1} line by line:
\begin{itemize}
    \item Line 6 (and 15, 24): Checking if there is only one element (or none) in $V'$ (in $B$, in $U$) can be done in constant time, since we keep track of the sizes of these sets.
    \item Line 11: Finding a vertex $c\in V'$ adjacent to all vertices in $V'$ (except itself) can be done in time $\mathcal{O}(\abs{V'}^2)$ by for each vertex $v$ in $V'$ checking if $\tau_{\bs{v}}(G)$ contains all edges in the set $\{(v,w):w\in V'\setminus\{v\}\}$. Let $c$ be the first such vertex $v$.
    \item Line 13 and 14: Constructing the sets $B$ and $L$ can be done in time $\mathcal{O}(\abs{V'}^2)$ by checking, for each vertex $v$ in $V'\setminus\{c\}$, if $\tau_{\bs{v}}(G)$ contains at least one edge from the set $\{(v,w):w\in V'\setminus\{c\}$. If this is the case, $v$ will be added to the array representing $B$, otherwise $v$ will be added to $L$.
    \item Line 19: Checking if $B=V'\setminus\{c\}$ can be done in time $\mathcal{O}(\abs{V'})$ by checking if all entries of the list representing $B$ are $1$, except at position $c$.
    \item Line 23: Constructing the set $U$ can be done in time $\mathcal{O}(\abs{V'}\abs{V(G)})$ by checking, for each $u$ in $V(G)\setminus V'$, that $\tau_{\bs{v}}(G)$ contains all edges in the set $\{(u,w):w\in B\}$ and no edges in the set $\{(u,w):w\in L\}$. If this is the case, $u$ will be added to the array representing $U$.
    \item Line 28-38: The body of this \emph{for}-loop will be executed $\mathcal{O}(\abs{V(G)})$ since there are at most $\mathcal{O}(\abs{V(G)})$ elements in $U$.
        \begin{itemize}
            \item Line 29: Checking if $(u,c)$ is an edge in $\tau_{\bs{v}}(G)$ can be done in constant time.
            \item Line 33: Finding a vertex $h$ which is adjacent to both $u$ and $c$ but to no other vertex in $V'$ can be done in time $\mathcal{O}(\abs{V'}\abs{V(G)})$ (or determining that there is none), by first finding the neighbors of $u$, i.e all the vertices $h$ such that $(u,h)$ is an edge in $\tau_{\bs{v}}(G)$ and then, for each neighbor $h$ of $u$, checking if $h$ is also adjacent to $c$ but to no other vertex in $V'$. This is done by checking if $(h,c)$ is an edge in $\tau_{\bs{v}}(G)$ and that no element of the set $\{(h,w):w\in V'\setminus\{c\}\}$ is.
        \end{itemize}
\end{itemize}
Thus, the total runtime of \cref{alg:Sn_vminor_1}, except for the recursive call to \cref{alg:Sn_vminor_2} in line 10, is $\mathcal{O}(\abs{V'}\abs{V(G)}^2)$ (from line 33 in the \emph{for}-loop.).

The runtime of each command in \cref{alg:Sn_vminor_2} is:
\begin{itemize}
    \item Line 5: Picking the vertex $f$ can be done in constant time (pick the first entry).
    \item Line 7: Finding a shortest path between $f$ and $c$ can be done in time $\mathcal{O}(\abs{V(G)}^2)$ by using Dijkstra's algorithm~\cite{Dijkstra1959}.
    \item Line 8-15: The body of this \emph{for}-loop will be executed $\mathcal{O}(\abs{V(G)})$ since the shortest path $P$ is necessarily shorter than the number of vertices in $\tau_{\bs{v}}(G)$.
        \begin{itemize}
            \item Line 9: Checking if $f$ is adjacent to any vertex in $V'\setminus\{c\}$ can be done in time $\mathcal{O}(\abs{V'})$ by checking if any of the edges $\{(f,w):w\in V'\setminus\{c\}$ are in $\tau_{\bs{v}}(G)$.
            \item Line 10: The column with entry $1$ in line 9 can be used for $v$ here and thus only adds a constant time to the runtime.
        \end{itemize}
\end{itemize}
Thus, the total runtime of \cref{alg:Sn_vminor_1}, except for the recursive call to \cref{alg:Sn_vminor_2} in line 6, is $\mathcal{O}(\abs{V(G)}^2)$ (from line 7 in the \emph{for}-loop).

To further substantiate the efficiency of the algorithm we give actual run-times for an implementation of the algorithm in \cref{fig:runtimes}.

\subsubsection{Proof that the algorithm is correct}\label{sec:theproof}
In this section we prove that the algorithm presented in the previous section works, i.e. it gives a sequence of local complementations ${\bs{v}}$ such that $\tau_{\bs{v}}(G)[V']=S_{V'}$, given a distance-hereditary graph $G$, if such a sequence exists.
It is relatively easy to show that the algorithm gives the desired results when it does not return an error, which we show in \cref{sec:theproof1}.
The hard part is to prove that, when the algorithm gives an error-flag it is in fact not possible to produce the star graph, i.e. the star graph is not a vertex-minor of G, which is done in \cref{sec:theproof2}.
The notation will be the same as in the previous section, $c$ is a vertex in $V'$ and is adjacent to the rest of the vertices in $V'$.
The vertices in $G[V'\setminus\{c\}]$ with degree greater than $0$, i.e. the vertices incident on some bad edge, are denoted as the set $B$.

\paragraph{Algorithm succeeds}\label{sec:theproof1}
In this section we show that if \cref{alg:Sn_vminor_1} returns a sequence ${\bs{v}}$, i.e. does not give an error-flag, then $\tau_{\bs{v}}(G)[V']=S_{V'}$.
We start by showing that \cref{alg:Sn_vminor_2} always succeeds and gives the desired result, assuming that \cref{alg:Sn_vminor_1} works.
After performing a pivot $\rho_{(v,u)}$, i.e. $\tau_u\circ\tau_v\circ\tau_u$, any neighbor of $v$ will become a neighbor of $u$, except $u$ itself.
This means that after the first pivot in line 9 in \cref{alg:Sn_vminor_2}, i.e. $\rho_{(p_1,f)}$, $f$ and $p_2$ will be adjacent.
We want to inductively show that this implies that after performing pivots along the whole path, $f$ and $c$ are adjacent.
To do this we only need to make sure that a pivot does not remove edges in the later part of the path.
More precisely, the pivot $\rho_{(p_i,f)}$ should not remove an edge $(p_j,p_{j+1})$ for $j>i$.
The fact that later edges in the path are not removed follows from the properties of the pivot and that the path is a shortest path.
Apart from edges incident on $u$ or $v$, a pivot $\rho_{(v,u)}$ can only flip edges in the set $N_v\times N_u$.
This shows that the pivot $\rho_{(p_i,f)}$ cannot remove an edge $(p_j,p_{j+1})$ since neither $p_j$ or $p_{j+1}$ is equal to $f$ or $p_i$ or is adjacent to $f$.
If $p_j$ or $p_{j+1}$ would be adjacent to $f$, then this would not be a shortest path.
We also need to make sure that we do not remove the edges from $E(S_{V'\setminus\{f\}})=\{(c,v):v\in V'\setminus\{c,f\}\}$, when doing pivots along the path.
By the same argument above we have that the pivot $\rho_{(p_i,f)}$ can only remove an edge in $E(S_{V'\setminus\{f\}})$ if $f$ is adjacent to a vertex in $V'\setminus\{c\}$.
This is the reason for the if-statement in line 5 in \cref{alg:Sn_vminor_2}, where we then just perform a local complementation on the corresponding vertex in $v\in V'\setminus\{c\}$ which will make $f$ and $c$ adjacent.
Performing the local complementation on such a vertex $v$ will not remove edges in $E(S_{V'\setminus\{f\}})$, since $v$ is a leaf in the induced subgraph on $V'$.
Note that there are only two cases where \textit{bad edges} are created.
If $f$ and $c$ are made adjacent by a local complementation on a vertex $v\in V'\setminus\{c\}$, as in line 7, the bad edge $(f,v)$ will be created.
On the other hand, if this is not the case but the last vertex $p_k$ is adjacent to some vertices $U\subseteq(V'\setminus\{c,f\})$, then the bad edges $\{(f,u)\}_{u\in U}$ will be created.
In both of these cases $\tau_{\bs{v}}(G)[V']$ will be a star-star graph, see \cref{def:SS}.
Note that $f$ can also be adjacent to some vertices in $V'\setminus\{f\}$, even before we perform the local complementations, but these edges will still form a star graph with $f$ as the center.
If one wants to minimize the number of local complementations and use local complementation instead of pivots along the path, this is in fact possible.
The only place where a pivot is needed instead of a local complementation is towards the end of the path, when $p_i$ is adjacent to a vertex in $V'\setminus\{c,f\}$ not on the path.

What is left to show is that if \cref{alg:Sn_vminor_1} succeeds and returns a $\bs{v}$, then $\tau_{\bs{v}}(G)=S_{V'}$.
This is easy to see, since if we perform local complementations on such vertices we are looking for, i.e. $u$ and possibly $h$ in \cref{eq:cond}, we will remove the bad edges and produce the star graph.
If $\abs{B}>2$ this has to be done twice, as captured by the loop over $i$ in \cref{alg:Sn_vminor_1}.
The reason for this is that, when doing a local complementation on such a $u$ we complement the induced subgraph $G[B]$.
Since $G[B]$ is a star graph, the induced subgraph after the local complementation will be a complete graph plus a single disconnected vertex which was the center of $G[B]$.
Performing the step once more will then complement the complete graph, without the disconnected vertex, and all bad edges are thus removed.

Note that we have nowhere in this section used the assumption that the graph is distance-hereditary.
This implies that if the algorithm succeeds we know that $\tau_{\bs{v}}(G)=S_{V'}$, independently of whether $G$ is distance-hereditary, in fact even independently of the rank-width of $G$.
Furthermore, since \cref{alg:Sn_vminor_2} always succeeds to make $\tau_{\bs{v}}(G)[V']$ connected and from the fact that any connected graph on two or three vertices is either a star graph or a complete graph, this implies that a star graph on any subset of size two or three is a vertex-minor of $G$, if the vertices are connected in $G$, which we make use of in \cref{sec:small_star}.
On the other hand, if the algorithm stops and gives an error-flag, then we do not know in general if $S_{V'}$ is a vertex-minor of $G$ or not.
In the next section we show that if the graph is distance-hereditary and the algorithm gives an error-flag we actually do know that $S_{V'}$ is not a vertex-minor of $G$.

\paragraph{Algorithm gives error}\label{sec:theproof2}
In this section we prove that if \cref{alg:Sn_vminor_1} gives an error-flag, i.e. if $\mathcal{P}(B,L,c)$ in \cref{eq:cond} is false, then the star graph is not a vertex-minor of the input graph.
At the steps in the algorithm where the error-flag can be raised, we know that the induced subgraph on $V'$ is either a star-star graph (\cref{def:SS}) or a complete-star (\cref{def:KS}) graph as shown in \cref{sec:theproof1}. The proof will follow the following sequence of steps.
\begin{enumerate}
    \item Prove for any distance-hereditary graph $G$ that if $\mathcal{P}(B,L,c)$ is false and $G[V']$ is a star-star graph (or complete-star graph) where $\abs{V'}=4$ then $S_{V'}$ is not a vertex-minor of $G$. This is done in \cref{thm:alg4}. 
    \item Use the case proven in step 1 to argue that if $\mathcal{P}(B,L,c)$ is false and $G[V']$ is a star-star graph where $\abs{V'}>4$ then $S_{V'}$ is not a vertex-minor of $G$. This is done in \cref{thm:algSS}.
    \item Use the case proven in step 1 to argue that if $\mathcal{P}(B,L,c)$ is false and $G[V']$ is a complete-star graph where $\abs{V'}>4$ then $S_{V'}$ is not a vertex-minor of $G$. This is done in \cref{thm:algKS}.
\end{enumerate}

\paragraph{Proof for a star-star (complete-star) graph of size $4$}

We will first show in \cref{thm:alg4} that if $\mathcal{P}(B,L,c)$ is false and $\abs{V'}=4$, then $S_{V'}$ is not a vertex-minor of $G$.
This will then allows us to prove the statement for the cases where $\abs{V'}\geq 4$.

\begin{thm}\label{thm:alg4}
   Let's assume that $G$ is a distance-hereditary graph with the following induced subgraph (which is both a star-star and a complete-star-graph)
   \begin{equation}
       G[V'=\{1,2,3,4\}]= \raisebox{-0.03\textwidth}{\includegraphics[width=0.1\textwidth]{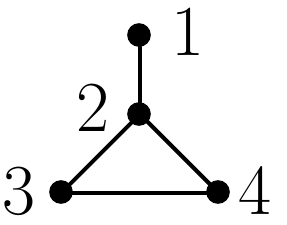}}.
   \end{equation}
   Furthermore assume that $\mathcal{P}(B,L,c)$ in \cref{eq:cond} is false, where $B=\{3,4\}$, $L=\{1\}$ and $c=2$.
   Then $S_{V'}\nless G$.
\end{thm}
\begin{proof}
    We will prove this by first showing that if $\mathcal{P}(B,L,c)$ is false and $\abs{V(G)}>4$, then there exist a removable leaf or twin.\footnote{As in section \cref{sec:reductions}, removable means a vertex not belonging to the target vertices $V'$}
    This then implies the we can actually delete removable leaves and twins, i.e. vertices in $T(G)\setminus V'$, until there is only the vertices in $V'$ left.
    Note that, if $\mathcal{P}(B,L,c)$ is false, then it is also false for any graph reached by deleting vertices in $V\setminus V'$.
    From \cref{thm:reductions}, i.e. the fact that deletion of removable twins or leafs does not change the property of whether a graph on a $V'$ is a vertex-minor and that $S_{V'}\neq_{LC}G[V']$, the theorem follows.
    More visually, $\mathcal{P}(B,L,c)$ is false if there exist no $u,h\in V\setminus V'$ such that
    \begin{equation}
        G[V'\cup\{u\}]=\raisebox{-0.03\textwidth}{\includegraphics[width=0.1\textwidth]{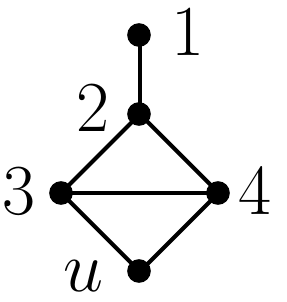}}\quad\lor\quad G[V'\cup\{u,h\}]=\raisebox{-0.03\textwidth}{\includegraphics[width=0.13\textwidth]{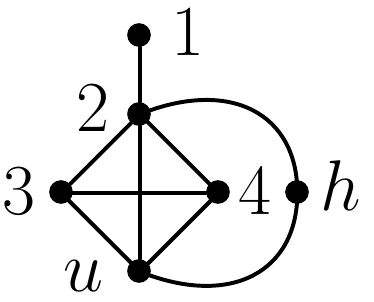}}.
    \end{equation}
    There are only two ways for $\mathcal{P}(B,L,c)$ to be false; either $(N_3\cap N_4)\setminus\{2\}=\emptyset$ or
    \begin{equation}\label{eq:cond2}
        (N_3\cap N_4)\setminus\{2\}\neq\emptyset\quad\land\quad\forall u\in(N_3\cap N_4)\setminus(N_1\cup\{2\}):\Big((u,2)\in E(G)\land (N_u\cap N_2)\setminus\big(N_1\cup N_3\cup N_4\big)=\emptyset\Big).
    \end{equation}
    We will consider these two cases separately and prove that if either is true, then $T(G)\setminus V'\neq\emptyset$.
    In most cases below we will do this by showing that one of the four vertices in $V'$ is not in $T(G)$ which shows that $T(G)\setminus V'\neq\emptyset$ since $T(G)\geq4$ by \cref{thm:T_size}.

    \noindent\underline{Case 1:}
    \begin{addmargin}[1em]{0em}
        To prove that if $(N_3\cap N_4)\setminus\{2\}=\emptyset$ and $\abs{V(G)}>4$, then there exists a removable leaf or twin, we consider the following cases.
        \begin{itemize}
            \item Assume that $\abs{(N_3\cup N_4)\setminus\{3,4\}}>1$.
                Since, by assumption $N_3\cap N_4=\{2\}$, $3$ and $4$ does not form a twin-pair.
                Also, neither $3$ nor $4$ is not a twin, since the twin-partner would have to be a common neighbor of $3$ and $4$.
                Furthermore, neither $3$ or $4$ is a leaf.
                Thus, the only way for $3\in T(G)$, is if $3$ is an axil, requiring some vertex not in $V'$ being a leaf.\footnote{The same for 4.}
                Finally, if $3\notin T(G)$, then there exist a vertex in $T(G)$ which is not in $V'$, since by \cref{thm:T_size} we know that $\abs{T(G)}\geq 4$.
            \item Assume that $(N_3\cup N_4)\setminus\{3,4\}=\{2\}$.
                \begin{itemize}
                    \item Assume that $\abs{N_1}>1$.
                        Then $1$ is not a leaf and $2$ is not an axil.
                        Furthermore, since nothing else is connected to $3$ and $4$, $2$ is not a twin.
                        Thus, the only way for $2\in T(G)$, is if $2$ is an axil, requiring some vertex not in $V'$ being a leaf.
                        Finally, if $2\notin T(G)$, then since $\abs{T(G)}\geq 4$ by \cref{thm:T_size}, there must exist a vertex in $T(G)$ which is not in $V'$.
                    \item Assume that $\abs{N_1}=1$.
                        Then $2$ is necessarily a cut-vertex and $G\setminus 2$ will contain a connected component with no vertices in $V'$, since $\abs{V(G)}>4$,.
                        Thus, there exist a vertex in $T(G)$ which is not in $V'$ by \cref{cor:CC}.
                \end{itemize}
        \end{itemize}
    \end{addmargin}

    \noindent\underline{Case 2:}
    \begin{addmargin}[1em]{0em}
        To prove that if \cref{eq:cond2} is true then there exists a removable leaf or twin, we consider the following cases.
        \begin{itemize}
            \item Assume that $\abs{N_1}>1$.
                Then if $2$ is an axil, the corresponding leaf cannot be in $V'$, since $1$ is not a leaf.
                Furthermore, since $2$ is not a leaf, if $2\in T(G)$ then $2$ has a twin-partner not in $V'$, which is then also in $T(G)$.
                On the other hand if $2\notin T(G)$, then there exist a vertex in $T(G)$ which is not in $V'$, by \cref{thm:T_size}.
            \item Assume that $N_1=\{2\}$.
                \begin{itemize}
                    \item Assume that $\abs{(N_3\cup N_4)\setminus\{3,4\}}>\abs{N_3\cap N_4}$.
                        In this case, $3$ and $4$ does not form a twin-pair.
                        Furthermore, neither $3$ or $4$ is a leaf.
                        Thus the only only way for $3(\text{or }4)\in T(G)$, is if $3(4)$ is an axil or a twin, requiring some vertex not in $V'$ being a leaf or a twin.
                        Finally if $3(4)\notin T(G)$, then there exist a vertex in $T(G)$ which is not in $V'$, since by \cref{thm:T_size} we know that $\abs{T(G)}\geq 4$.
                    \item Assume that $(N_3\cup N_4)\setminus\{3,4\}=N_3\cap N_4$.
                        We will for this case show that $\abs{T(G)\setminus V'}>0$ by assuming that $T(G)=V'$ and arriving at a contradiction.
                        Since this implies that $T(G)\neq V'$ and from the fact that $\abs{T(G)}\geq 4$, we know that $\abs{T(G)\setminus V'}>0$.
                    Consider the induced subgraph $G\setminus 4$ which is also distance-hereditary.\footnote{Induced subgraphs of a distance-hereditary graph are distance-hereditary.}
                        From \cref{thm:T_size} we know that $\abs{T(G\setminus 4)}\geq 4$, since $(N_3\cap N_4)\setminus \{2\}\neq\emptyset$ and therefore $\abs{G\setminus 4}\geq 4$.
                        Thus, there is a vertex $v\notin V'$ such that $v\in T(G\setminus 4)$ but $v\notin T(G)$.
                        Note that by the assumption from \cref{eq:cond2}, any neighbor of $4$, except $2$, is also a neighbor of both $2$ and $3$.
                        The removal of $4$ cannot therefore create a new leaf in $V\setminus V'$.
                        The only option left is if there are two vertices $v,v'\in V(G)\setminus\{4\}$, such that $v,v'$ form a twin-pair in $G\setminus4$ but not in $G$.
                        If $v$ and $v'$ are such vertices, it must be the case that $4$ is adjacent to exactly one of $v$ and $v'$.
                        Assume without loss of generality that $4$ is adjacent to $v'$ but not to $v$.
                        The neighborhoods of these vertices are then
                        \begin{equation}
                            N_v=N_{v'}\setminus\{4\}\quad\land\quad 4\in N_{v'}.
                        \end{equation}
                        Note that the vertices adjacent to $4$ are $(N_3\cap N_4)\cup\{3\}$.
                        Firstly, $v'$ cannot be in $N_3\cap N_4$, since $v$ is then necessarily adjacent to $3$ but not to $4$ which contradicts the assumption that $(N_3\cup N_4)\setminus\{3,4\}=N_3\cap N_4$.
                        Secondly, $v'$ cannot be $3$, since $v$ is then necessarily a neighbor of $2$ and all vertices in $N_3\cap N_4$, contradicting the second part of \cref{eq:cond2}.
                \end{itemize}
        \end{itemize}
    \end{addmargin}
\end{proof}

\paragraph{Proof for star-star graphs}

We are now able to prove the same statement as in \cref{thm:alg4} but for $\abs{V'}\geq4$.
The case when $G[V']$ is a star-star graph is given in \cref{thm:algSS}.

\begin{thm}\label{thm:algSS}
    Let's assume that $G$ is a distance-hereditary graph and $V'$ is a subset $V'\subseteq V(G)$ such that the induced subgraph $G[V']$ is a star-star graph $SS_{(B',L',c')}$.
    Furthermore assume that $\mathcal{P}(B',L',c')$ is false, then $S_{V'}\nless G$.
\end{thm}
\begin{proof}
    Pick an edge in $(b_1,b_2)\in G[B']$, which exist since $\abs{B'}>1$ and $G[B']$ is a star graph.
    We will prove this by first showing that
    \begin{equation}\label{eq:Pimpl}
    \neg\mathcal{P}(B',L',c')\quad\Rightarrow\quad\exists l\in L':\qty(\neg\mathcal{P}(\{b_1,b_2\},\{l\},c')).\footnote{Remember that $L'\neq\emptyset$, by definition of a star-star graph.}
    \end{equation}
    Then from \cref{thm:alg4} we know that $S_{\{l,c',b_1,b_2\}}\nless G$ for some $l\in L'$ and the corollary follows, because if $S_{\{l,c',b_1,b_2\}}$ is not a vertex-minor of $G$ then neither is $S_{V'}$, since $S_{\{l,c',b_1,b_2\}}<S_{V'}$.
    To show that \cref{eq:Pimpl} is true, we instead show the contrapositive statement, i.e.
    \begin{equation}\label{eq:contrapos1}
        \forall l\in L':\qty(\mathcal{P}(\{b_1,b_2\},\{l\},c'))\quad\Rightarrow\quad\mathcal{P}(B',L',c').
    \end{equation}
    Let $\mathcal{Q}(u,B,L,c)$ be the expression on $\mathcal{P}(B,L,c)$ such that
    \begin{equation}
        \mathcal{P}(B,L,c)=\exists u\in V\setminus V':\mathcal{Q}(u,B,L,c)
    \end{equation}
    For each $l\in L'$, let $u_l$ be a vertex in $V\setminus \{l,c',b_1,b_2\}$ such that $\mathcal{Q}(u_l,\{b_1,b_2\},\{l\},c')$ is true.
    We will now show that for at least one of these $u_l$, $u_l\in V\setminus V'$ and $\mathcal{Q}(u_l,B',L',c')$ is true.
    To show that for at least one $u_l$, $u_l\in V\setminus V'$ and $\mathcal{Q}(u_l,B',L',c')$ is true we will go trough the following steps:
    \begin{enumerate}
        \item Show that
            \begin{equation}\label{eq:SSstep1}
                \forall l\in L':\qty(u_l\notin V')
            \end{equation}
        \item Show that
            \begin{equation}\label{eq:SSstep2}
                \forall l\in L':\qty(B'\subseteq N_{u_l})
            \end{equation}
        \item Show that
            \begin{equation}\label{eq:SSstep3}
                \exists l\in L':\qty(L'\cap N_{u_l}=\emptyset)
            \end{equation}
        \item Fix $l\in L'$ to be such that the corresponding expression in \cref{eq:SSstep3} is true.
        \item Show that
            \begin{equation}\label{eq:SSstep5}
                (u_l,c')\in E(G)\;\lor\;\exists h:\qty(h\in N_{u_l}\cap N_{c'}\setminus\bigcup_{x\in V'\setminus\{c'\}}N_x)
            \end{equation}
    \end{enumerate}
    If all the statements in the above steps are shown to be true we know that there exist a $u_l$ in $V\setminus V'$ such that $\mathcal{Q}(u_l,B',L',c')$ and therefore $\mathcal{P}(B',L',c')$ is true.
    It is important to note that even if we consider the statements in the above steps separately we know that there exist at least one $u_l$ that simultaneously satisfy all.
    To see this, note that we only use a existential quantifier in step 3 and in step 5 we consider a $u_l$ which satisfies the corresponding property in step 3.
    Let's now consider the steps 1 through 5 one by one.\footnote{Step 4 is trivial.}
    We will in these steps often claim that certain small graphs are not distance-hereditary.
    Verifying this can be done by hand or using our code supplied at~\cite{git}.

    \noindent\underline{Step 1:}
    \begin{addmargin}[1em]{0em}
        Firstly, since $\mathcal{Q}(u_l,\{b_1,b_2\},\{l\},c')$ is true we know that $u_l\neq c'$.
        Similarly, we know that $u_l\in(N_{b_1}\cap N_{b_2})\setminus\{c'\}$, thus $u_l$ is not in $B'$, since $G[B']$ is a star graph.
        Furthermore, since $b$ and $l$ are not adjacent $\forall b\in B'$ and $\forall l\in L'$, $u_l$ is not in $L'$.
        We therefore know that $u_l\notin V'$.
        \hfill$\diamond$
    \end{addmargin}

    \noindent\underline{Step 2:}
    \begin{addmargin}[1em]{0em}
        To see that $B'\subseteq N_{u_l}$, assume that first that this is not the case, i.e. $\exists \tilde{b}\in B':\tilde{b}\notin N_{u_l}$.
        We will now show that this contradicts the distance-hereditary property.
        Note that $\tilde{b}$ is not the center of the star graph $G[B]$, because either $b_1$ or $b_2$ is the center, since $(b_1,b_2)$ is an edge in $G[B]$.
        Let's assume without loss of generality that $b_1$ is the center of $G[B]$.
        Furthermore, let's consider the cases where $u_l$ and $c'$ are adjacent or not separately.
        \begin{itemize}
            \item Assume that $u_l$ and $c'$ are not adjacent and consider the following induced subgraph
            \begin{equation}\label{eq:step2_1}
                G[\{c',\tilde{b},b_1,b_2,u_l\}]=\raisebox{-0.05\textwidth}{\includegraphics[scale=0.5]{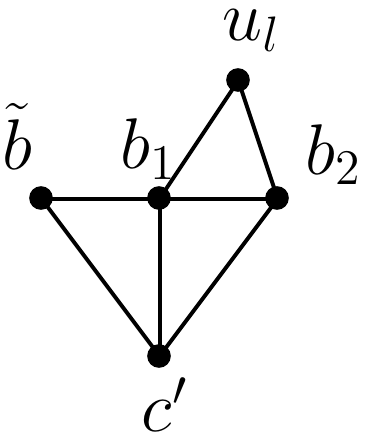}}.
            \end{equation}
            The graph in \cref{eq:step2_1} is not distance-hereditary, since the distance between for example $\tilde{b}$ and $u_l$ increase if $b_1$ is removed.
            This is therefore a contradiction to the assumption that $G$ is distance-hereditary.

            \item Assume that $u_l$ and $c'$ are adjacent.
            We then know that there exist a vertex $h_l$ which is adjacent to $u_l$ and $c'$, since $\mathcal{Q}(u_l,\{b_1,b_2\},\{l\},c')$ is true.
            First, let's show that $h_l$ and $\tilde{b}$ are not adjacent.
            In fact we will show that $h_l$ is not adjacent to any vertex in $B$, which will be useful in step 5.

            \begin{addmargin}[1em]{0em}
                Assume the opposite, i.e. $h_l$ is adjacent to some vertex $\hat{b}\in B'\setminus\{b_1,b_2\}$.
                We already know that $h_l$ is not adjacent to $b_1$ or $b_2$, since $\mathcal{Q}(u_l,\{b_1,b_2\},\{l\},c')$ is true.
                Consider therefore the following induced subgraph
                \begin{equation}\label{eq:step2_2}
                    G[\{c',\hat{b},b_1,b_2,h_l\}]=\raisebox{-0.04\textwidth}{\includegraphics[scale=0.5]{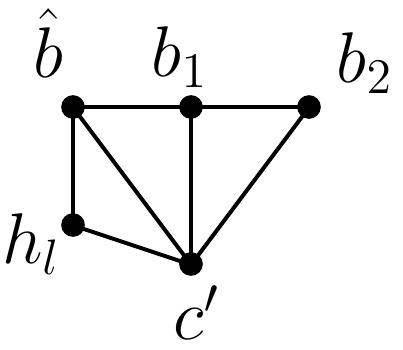}}
                \end{equation}
                which is not distance-hereditary and we therefore know that $N_{h_l}\cap B'=\emptyset$.
            \end{addmargin}


            Consider now on the other hand the following induced subgraph, with the knowledge that $h_l$ is not adjacent to $\tilde{b}$
            \begin{equation}\label{eq:step2_3}
                G[\{c',\tilde{b},b_1,u_l,h_l\}]=\raisebox{-0.06\textwidth}{\includegraphics[scale=0.5]{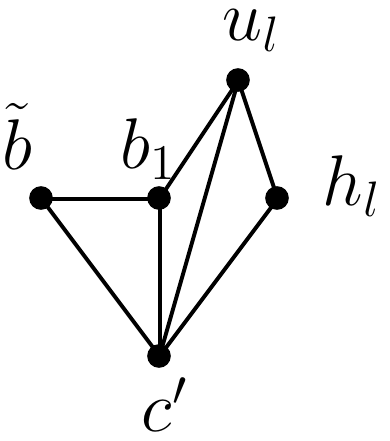}}
            \end{equation}
            which is also not distance-hereditary.
        \end{itemize}
        Since in all cases we arrived at a non-distance-hereditary graph we know that $B'\subseteq N_{u_l}$.
        \hfill$\diamond$
    \end{addmargin}

    \noindent\underline{Step 3:}
    \begin{addmargin}[1em]{0em}
        We show that at least for one of the $u_l$, $L'\cap N_{u_l}=\emptyset$.
        We will do this by contradiction, assume therefore that $\forall l\in L':\qty(L'\cap N_{u_l}\neq\emptyset)$.
        Since $\mathcal{Q}(u_l,\{b_1,b_2\},\{l\},c')$ is true, we know that $\{l\}\cap N_{u_l}=\emptyset$ for all $l\in L'$.
        Consider now the graph $G[L'\cup\{u_l\}_l]$.
        From \cref{thm:LU} we know that there exist $l_1,l_2\in L'$ and $u_{l_3},u_{l_4}\in V\setminus V'$, such that $u_{l_3}$ is adjacent to $l_1$ but not to $l_2$ and $u_{l_4}$ is adjacent to $l_2$ but not to $l_1$.\footnote{Note that for example $l_1$ and $l_3$ could be the same vertex, but not necessarily.}
        But this is in contradiction with that the graph is distance-hereditary.
        To see this, consider the induced subgraph
        \begin{equation}\label{eq:step3}
            G[\{c',b_1,l_1,l_2,u_{l_3},u_{l_4}\}]=\raisebox{-0.06\textwidth}{\includegraphics[scale=0.5]{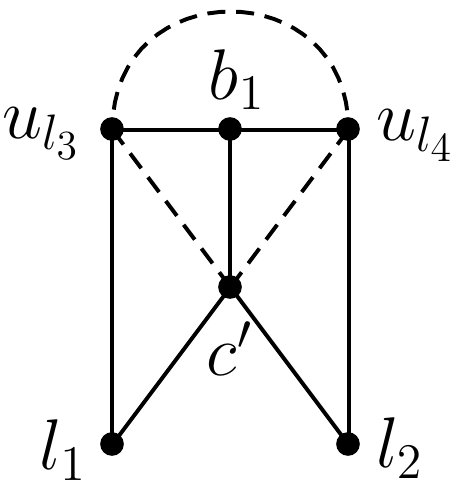}}
        \end{equation}
        which is not distance-hereditary, independently if the edges $(u_{l_3},u_{l_4})$, $(c',u_{l_3})$ and $(c',u_{l_4})$ are individually present or not.
        Since this contradicts the distance-hereditary property we know that $\exists l\in L':(L'\cap N_{u_l}=\emptyset)$.
        \hfill$\diamond$
    \end{addmargin}

    \noindent\underline{Step 5:}
    \begin{addmargin}[1em]{0em}
        Let's assume that $u_l$ is then a vertex such that $B\subseteq N_{u_l}$ and $L'\cap N_{u_l}=\emptyset$.
        If $u_l$ is not adjacent to $c'$, then clearly $\mathcal{Q}(u_l,B',L',c')$.
        On the other hand if $u_l$ and $c'$ are adjacent we know that there exist a $h_l$ in $N_{u_l}\cap N_{c'}\setminus\bigcup_{x\in \{l,b_1,b_2\}}N_x$.
        We thus need to show that $h_l$ is not adjacent to any vertex in $V'$, other than $c'$.
        Firstly, $h_l$ cannot be adjacent to a vertex in $L'$, since this would violate the distance-hereditary property.
        To see this, assume that $h_l$ is adjacent to $\tilde{l}\in L'$ and consider the following induced subgraph
        \begin{equation}\label{eq:step5}
            G[\{c',b_1,\tilde{l},u_l,h_l\}]=\raisebox{-0.06\textwidth}{\includegraphics[scale=0.5]{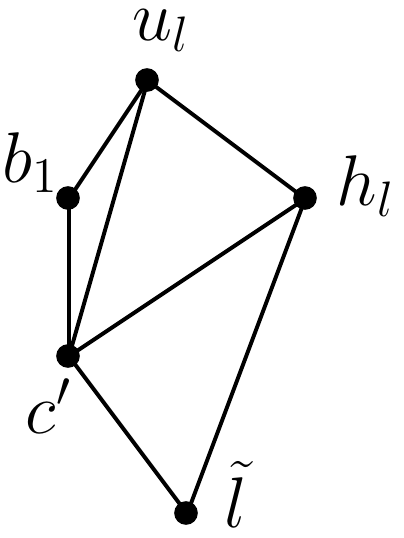}}
        \end{equation}
        which is not distance-hereditary.
        This is a contradiction with the distance-hereditary property and we therefore know that $N_{h_l}\cap L'=\emptyset$.
        As we already shown in step 2, $h_l$ is also not adjacent to any vertex in $B'$.
        Thus, $h_l$ is not adjacent to any vertex in $V'=B'\cup L'\cup \{c'\}$.
        \hfill$\diamond$
    \end{addmargin}
    We have therefore shown that \cref{eq:contrapos1} is true which implies that \cref{eq:Pimpl} is true.
    Finally, as we described in the beginning of the proof, this implies that if $\mathcal{P}(B',L',c')$ is false then $S_{V'}\nless G$.
\end{proof}
\newpage
\begin{thm}\label{thm:LU}
    Assume $G$ is a graph on the vertices $U\cup L$ such that $U\cap L=\emptyset$ and $U\neq\emptyset$.
    Furthermore, assume that for each $l$ in $L$, there is at least one vertex in $U$ not adjacent to $l$ and for each $u$ in $U$, there is at least one vertex in $L$ adjacent to $u$, i.e. $G$ satisfies the following expression
    \begin{equation}\label{eq:LUcond}
        \mathcal{R}(U,L)=\forall l\in L:\Big(\exists u\in U:u\notin N_l\Big)\;\land\;\forall u\in U:\Big(\exists l\in L:l\in N_u\Big)
    \end{equation}
    Then there exist two vertices $u_1$ and $u_2$ in $U$ and two vertices $l_1$ and $l_2$ in $L$ such that $u_1$ is adjacent to $l_1$ but not to $l_2$ and $u_2$ is adjacent to $l_2$ but not to $l_1$. In other words the induced subgraph is of the following form
    \begin{equation}\label{eq:LUind}
        G[\{u_1,u_2,l_1,l_2\}]=\raisebox{-0.03\textwidth}{\includegraphics[width=0.1\textwidth]{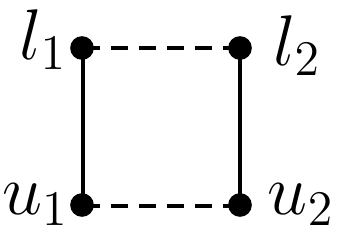}}.
    \end{equation}
    where the dashed edges are individually either present or not.
\end{thm}
\begin{proof}
    We will first show that $\abs{L}\geq 2$ and $\abs{U}\geq 2$.
    Pick an element $u_1\in U$, which exists since $U\neq\emptyset$, by assumption there is a $l_1\in L$ which is adjacent to $u_1$.
    Furthermore there exist a $u_2\in U$ which is not adjacent to $l_1$, thus $u_1\neq u_2$.
    Finally, by assumption there is a $l_2\in L$ which is adjacent to $u_2$, thus $l_1\neq l_2$.
    Note, that this does not yet prove the theorem, since $u_1$ and $l_2$ might be adjacent.

    We will first prove the theorem for $\abs{L}=2$ and then use this to prove the general case.

    \noindent\underline{$\abs{L}=2$ :}
    \begin{addmargin}[1em]{0em}
        Let's denote the vertices in $L$ by $l_1$ and $l_2$.
        We first show by contradiction that both vertices in $L$ must have at least one neighbor in $U$.
        Assume that $l_1$ does not have a neighbor in $U$, then all vertices in $U$ must be adjacent to $l_2$ by the second part of \cref{eq:LUcond}, but then the first part of \cref{eq:LUcond} is false.
        Thus $l_1$ has at least one neighbor in $U$ and by symmetry the same is true for $l_2$.
        Now choose such a neighbor of $l_1$ in $U$ and denote this $u_1$.
        We now show by contradiction that there exist another vertex $u_2\in U$ which is adjacent to $l_2$ but not to $l_1$.
        Assume that this is not the case, i.e. all vertices in $U\setminus\{u_1\}$ are adjacent to $l_1$ or not adjacent to $l_2$.
        If a vertex in $U$ is not adjacent to $l_2$ then it is necessarily adjacent to $l_1$, thus by assumption all vertices in $U\setminus\{u_1\}$ are adjacent to $l_1$.
        This is in contradiction with the first part of \cref{eq:LUcond} and the theorem for $\abs{L}=2$ follows.
    \end{addmargin}

    \noindent\underline{$\abs{L}>2$ :}
    \begin{addmargin}[1em]{0em}
        We will show that the following is true: (1) $G$ has an induced subgraph as in \cref{eq:LUind} or (2) there exist a $l\in L$ such that $G\setminus l$ satisfy $\mathcal{R}(U,L\setminus\{l\})$.
        The theorem then follows since if (1) is true the theorem follows directly and if (2) is true we can make the same argument for $G\setminus l$ for some $l\in L$ and then possibly for $(G\setminus l)\setminus l'$ etc., which at some point will give the case $\abs{L}=2$, which we have proven above.
        Note that if the graph reached by deleting vertices from $G$, has the graph in \cref{eq:LUind} as an induced subgraph, then so does $G$.

        To prove that (1) or (2) is true, we show that if (2) is false then (1) is necessarily true.
        Therefore, assume now that (2) is false, which means that for every choice of $l$, $G\setminus l$ does not satisfy $\mathcal{R}(U,L\setminus\{l\})$.
        The only possibility for this to happen, i.e. the deletion of $l$ makes the graph not satisfy \cref{eq:LUcond}, is if the deletion of $l$ makes some $u\in U$ not adjacent to any vertex in $L$.
        It is easy to see that this can only happen if $\exists u\in U:(L\cap N_u=\{l\})$.
        Since this should be true for all $l\in L$, we have that if (2) is false, the following is true,
        \begin{equation}\label{eq:2false}
            \forall l\in L:\Big(\exists u\in U:(L\cap N_u=\{l\})\Big).
        \end{equation}
        But \cref{eq:2false} implies that (1) is true.
        To see this pick two different vertices $l_1$ and $l_2$ in $L$.
        From \cref{eq:2false} we know that there exist a vertex $u_1\in U$ such that $L\cap N_{u_1}=\{l_1\}$ and similarly a $u_2$ for $l_2$.
        Note that $u_1\neq u_2$ since $N_{u_1}\neq N_{u_2}$.
        Furthermore, since $L\cap N_{u_1}=\{l_1\}$ and $L\cap N_{u_2}=\{l_2\}$ the induced subgraph $G[\{u_1,u_2,l_1,l_2\}]$ is as in \cref{eq:LUind}.
    \end{addmargin}
\end{proof}

\paragraph{Proof for a complete-star graph}
Here we prove that if $\mathcal{P}(B',L',c')$ is false, then $S_{V'}\nless G$ for the case where $G[V']$ is a complete-star graph. We have the following theorem.
\begin{thm}\label{thm:algKS}
    Let's assume that $G$ is a distance-hereditary graph and $V'$ is a subset $V'\subseteq V(G)$ such that the induced subgraph $G[V']$ is a complete-star graph $KS_{(B',L',c')}$.
    Furthermore assume that $\mathcal{P}(B',L',c')$ is false, then $S_{V'}\nless G$.
\end{thm}
\begin{proof}
    We will prove this by induction on the size of $B'$.
    The base-case, $\abs{B'}=2$, is true due to \cref{thm:algSS}, since for $\abs{B'}=2$, the graph $G$ is also a star-star graph.
    Let's now assume that \cref{thm:algKS} is true for $\abs{B'}=k$.
    We will now prove that the theorem is true for $\abs{B'}=k+1$ by showing that
    \begin{equation}\label{eq:Pimpl2}
        \neg\mathcal{P}(B',L',c')\quad\Rightarrow\quad\exists b\in B':\neg\mathcal{P}(B'\setminus\{b\},L',c').
    \end{equation}
    where $\abs{B'}=k+1\geq 3$.
    Then from the induction hypothesis we know that $S_{V'\setminus\{b\}}\nless G$ for some $b\in B'$ and the corollary follows, because if $S_{V'\setminus\{b\}}$ is not a vertex-minor of $G$ then neither is $S_{V'}$, since $S_{V'\setminus\{b\}}<S_{V'}$.
    To show that \cref{eq:Pimpl2} is true, we instead show the contrapositive statement, i.e.
    \begin{equation}\label{eq:contrapos2}
        \forall b\in B':\mathcal{P}(B'\setminus\{b\},L',c')\quad\Rightarrow\quad\mathcal{P}(B',L',c').
    \end{equation}
    Let's therefore assume that $\mathcal{P}(B'\setminus\{b\},L',c')$ is true for all $b$ in $B'$.
    Let $\mathcal{Q}(u,B',L',c')$ be the expression on $\mathcal{P}(B',L',c')$ such that
    \begin{equation}
        \mathcal{P}(B',L',c')=\exists u\in V\setminus V':\mathcal{Q}(u,B',L',c')
    \end{equation}
    For each $b\in B'$, let $u_b$ be a vertex in $V\setminus (V'\setminus\{b\})$ such that $\mathcal{Q}(u_b,B'\setminus\{b\},L',c')$ is true.
    We now need to show that for at least one of these $u_b$, $u_b\in V\setminus V'$ and $\mathcal{Q}(u_b,B',L',c')$ is true.
    Note that $L'\cap N_{u_b}=\emptyset$ for all $b$, since $\mathcal{Q}(u_b,B'\setminus\{b\},L',c')$ for all $b$.
    Thus, to show that for at least on $u_b$, $u_b\in V\setminus V'$ and $\mathcal{Q}(u_b,B',L',c')$ is true we will go trough the following steps:
    \begin{enumerate}
        \item Show that there can maximally be one $\tilde{b}\in B'$ such that $u_{\tilde{b}}\in B'$, i.e.
            \begin{equation}\label{eq:KSstep1}
                \exists \tilde{b}\in B':\Big(\forall b\in B':\big(b=\tilde{b}\;\lor\;u_b\notin B'\big)\Big)
            \end{equation}
        \item Fix $\tilde{b}\in B'$ to be such that the corresponding expression in \cref{eq:KSstep1} is true.\footnote{Note that this does not imply that $u_{\tilde{b}}\in B'$.}
        \item Show that there can maximally be one $\hat{b}\in B'\setminus\{\tilde{b}\}$ such that $B'\nsubseteq N_{u_{\hat{b}}}$, i.e.
            \begin{equation}\label{eq:KSstep3}
                \exists \hat{b}\in B'\setminus\{\tilde{b}\}:\Big(\forall b\in B'\setminus\{\tilde{b}\}:\big(b=\hat{b}\;\lor\;B'\subseteq N_{u_b}\big)\Big)
            \end{equation}
        \item Fix $\hat{b}\in B'$ to be such that the corresponding expression in \cref{eq:KSstep3} is true.
        \item Use step 1 to 4 to show that there exist a $b\in B'$ such that $u_b\notin B'$, $B'\subseteq N_{u_b}$ and
            \begin{equation}\label{eq:KSstep5}
                (u_b,c')\notin E(G)\;\lor\;\exists h_b:\qty(h_b\in N_{u_b}\cap N_{c'}\setminus\bigcup_{x\in V'\setminus\{c'\}}N_x)
            \end{equation}
            i.e. $\mathcal{Q}(u_b,B',L',c')$ is true.
    \end{enumerate}
    Let us now consider the steps 1 through 5 one by one.\footnote{Step 2 and 4 are trivial.}
    We will in these steps often claim that certain small graphs are not distance-hereditary.
    Verifying this can be done by hand or using our code supplied at~\cite{git}.

    \noindent\underline{Step 1:}
    \begin{addmargin}[1em]{0em}
        Here we show that \cref{eq:KSstep1} is true.
        Firstly, if for all $b\in B'$ we have that $u_b\notin B'$, then \cref{eq:KSstep1} is clearly true, since $\tilde{b}$ can then be chosen as any element in $B'$.\footnote{Remember that $\abs{B'}\geq 3$.}
        We now show by contradiction that there cannot exist two different vertices $\tilde{b}_1,\tilde{b}_2\in B'$, such that $u_{\tilde{b}_1}\in B'$ and $u_{\tilde{b}_2}\in B'$.
        Thus, let's assume that such vertices $\tilde{b}_i$ for $i\in \{1,2\}$, does exist.
        Note that, since $\mathcal{Q}(u_{\tilde{b}_i},B'\setminus\{\tilde{b}_i\},L',c')$ is true, we know that $u_{\tilde{b}_i}$ is adjacent to all vertices in $B'\setminus\{\tilde{b}_i\}$ and therefore $u_{\tilde{b}_i}=\tilde{b}_i$.
        Since $u_{\tilde{b}_i}=\tilde{b}_i$, we know that $u_{\tilde{b}_1}\neq u_{\tilde{b}_1}$.
        Furthermore, from the fact that $u_{\tilde{b}_i}\in B'$, we know that $u_{\tilde{b}_i}$ is adjacent to $c'$ and thus, since $\mathcal{Q}(u_{\tilde{b}_i},B'\setminus\{\tilde{b}_i\},L',c')$ is true, there exist a vertex $h_i$ such that
        \begin{equation}
            h_i\in N_{u_{\tilde{b}_i}}\cap N_{c'}\setminus\bigcup_{x\in V'\setminus\{c',\tilde{b}_i\}}N_x.
        \end{equation}
        The vertices $h_i$ are necessarily different, since $h_1$ is adjacent to $b_1=u_{\tilde{b}_1}$ but not to $b_2$ and vice versa for $h_2$.
        Now consider the following induced subgraph
        \begin{equation}\label{eq:sub1}
            G[\{c',\tilde{b}_1,\tilde{b}_2,b,h_1,h_2\}]=\raisebox{-0.06\textwidth}{\includegraphics[scale=0.5]{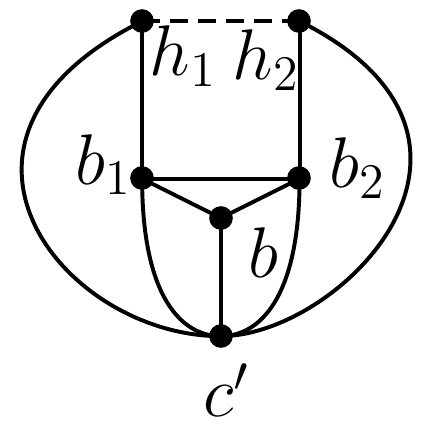}}
        \end{equation}
        where $b$ is a vertex in $B'\setminus\{b_1,b_2\}$, which exists since $\abs{B'}\geq 3$ by assumption.
        The graph in \cref{eq:sub1} is not distance-hereditary, independently if the edge $(h_1,h_2)$ is present or not.
        Since this contradicts the fact that $G$ is distance-hereditary, we know that \cref{eq:KSstep1} is true.
        \hfill$\diamond$
    \end{addmargin}

    \noindent\underline{Step 3:}
    \begin{addmargin}[1em]{0em}
        Here we show that \cref{eq:KSstep3} is true.
        Firstly, if for all $b\in B'\setminus\{\tilde{b}\}$ we have that $B'\subseteq N_{u_b}$, then \cref{eq:KSstep3} is clearly true, since $\hat{b}$ can then be chosen as any element in $B'\setminus\{\tilde{b}\}$.\footnote{Remember that $\abs{B'}\geq 3$.}
        We now show by contradiction that there cannot exist two different vertices $\hat{b}_1,\hat{b}_2\in B'$, such that $B'\nsubseteq N{u_{\hat{b}_1}}$ and $B'\nsubseteq N{u_{\hat{b}_2}}$.
        Thus, let's assume that such vertices $\hat{b}_i$ for $i\in\{1,2\}$, does exist.
        Let's for the remainder of this step denote $u_{\hat{b}_i}$ as $u^{(i)}$.
        From the previous steps we know that $u^{(i)}\notin B'$ and furthermore $(B'\setminus\{\hat{b}_i\})\subseteq N_{u^{(i)}}$ since $\mathcal{Q}(u^{(i)},B'\setminus\{\hat{b}_i\},L',c')$ is true.
        Thus, by assumption, we have that $\hat{b}_i\notin N_{u^{(i)}}$.
        The vertices $u^{(i)}$ are then necessarily different, i.e. $u^{(1)}\neq u^{(2)}$, since for example $\hat{b}_1$ is a neighbor of $u^{(2)}$ but not of $u^{(1)}$.
        We will now show that this contradicts the fact that $G$ is distance-hereditary, by considering the following cases:
        \begin{itemize}
            \item Assume that $u^{(1)}$ is not adjacent to $u^{(2)}$ and consider the following induced subgraph
                \begin{equation}\label{eq:sub2}
                    G[\{b,\hat{b}_1,\hat{b}_2,u^{(1)},u^{(2)}\}]=\raisebox{-0.03\textwidth}{\includegraphics[scale=0.5]{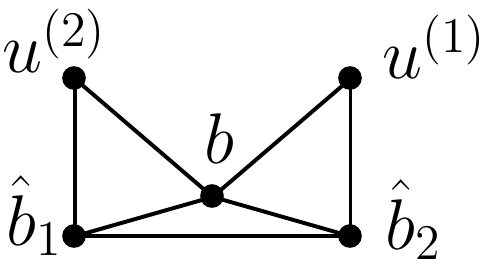}}
                \end{equation}
                where $b$ is a vertex in $B'\setminus\{\hat{b}_1,\hat{b}_2\}$, which exists since $\abs{B'}\geq 3$ by assumption.
                The graph in \cref{eq:sub2} is not distance-hereditary.
            \item Assume that $u^{(1)}$ is adjacent to $u^{(2)}$.
                \begin{itemize}
                    \item Assume that neither $u^{(1)}$ or $u^{(2)}$ is adjacent to $c'$ and consider the following induced subgraph
                        \begin{equation}\label{eq:sub3}
                            G[\{c',\hat{b}_1,\hat{b}_2,u^{(1)},u^{(2)}\}]=\raisebox{-0.05\textwidth}{\includegraphics[scale=0.5]{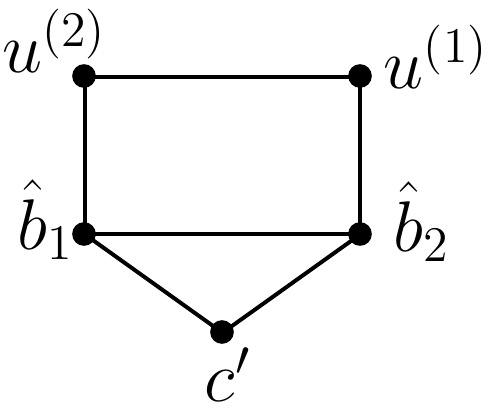}}
                        \end{equation}
                        which is not distance-hereditary.
                    \item Assume that exactly one of $u^{(1)}$ and $u^{(2)}$ is adjacent to $c'$ and let's assume without loss of generality that it is $u^{(1)}$ that is adjacent to $c'$.
                        Since $\mathcal{Q}(u^{(1)},B'\setminus\{\hat{b}_1\},L',c')$ is true and $u^{(1)}$ is adjacent to $c'$, we know that there exist a vertex $h_1$ such that
                        \begin{equation}
                            h_1\in N_{u^{(1)}}\cap N_{c'}\setminus\bigcup_{x\in V'\setminus\{c',b_1\}}N_x.
                        \end{equation}
                        Consider now the following induced subgraph
                        \begin{equation}\label{eq:sub4}
                            G[\{c',b,\hat{b}_2,u^{(1)},u^{(2)},h_1\}]=\raisebox{-0.06\textwidth}{\includegraphics[scale=0.5]{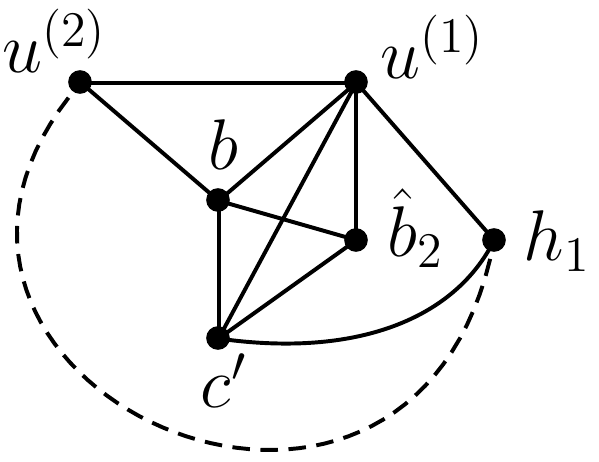}}
                        \end{equation}
                        where $b$ is a vertex in $B'\setminus\{\hat{b}_1,\hat{b}_2\}$, which exists since $\abs{B'}\geq 3$ by assumption.
                        The graph in \cref{eq:sub4} is not distance-hereditary, independently if the edge $(h_1,u^{(2)})$ is present or not.
                    \item Assume that both $u^{(1)}$ and $u^{(2)}$ is adjacent to $c'$.
                        Since $\mathcal{Q}(u^{(i)},B'\setminus\{\hat{b}_i\},L',c')$ is true and $u^{(i)}$ is adjacent to $c'$, we know that there exist a vertex $h_i$ such that
                        \begin{equation}
                            h_i\in N_{u^{(i)}}\cap N_{c'}\setminus\bigcup_{x\in V'\setminus\{c',b_i\}}N_x.
                        \end{equation}
                        \begin{itemize}
                            \item Assume that $h_1=h_2$, which implies that $h_1$ is not adjacent to $\hat{b}_1$ or $\hat{b}_2$ and consider the following induced subgraph
                                \begin{equation}\label{eq:sub5}
                                    G[\{\hat{b}_1,\hat{b}_2,u^{(1)},u^{(2)},h_1\}]=\raisebox{-0.04\textwidth}{\includegraphics[scale=0.5]{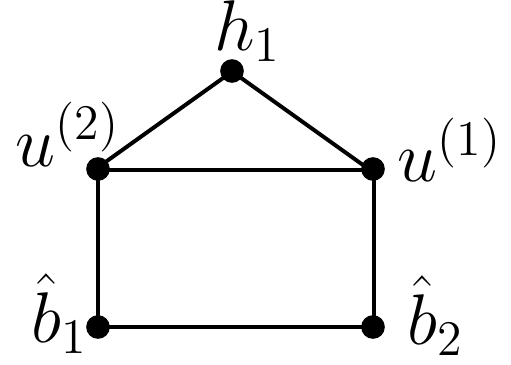}}
                                \end{equation}
                                which is not distance-hereditary.
                            \item Assume that $h_1\neq h_2$ and consider the following induced subgraph
                                \begin{equation}\label{eq:sub6}
                                    G[\{c',b,\hat{b}_1,\hat{b}_2,u^{(1)},u^{(2)},h_1,h_2\}]=\raisebox{-0.06\textwidth}{\includegraphics[scale=0.5]{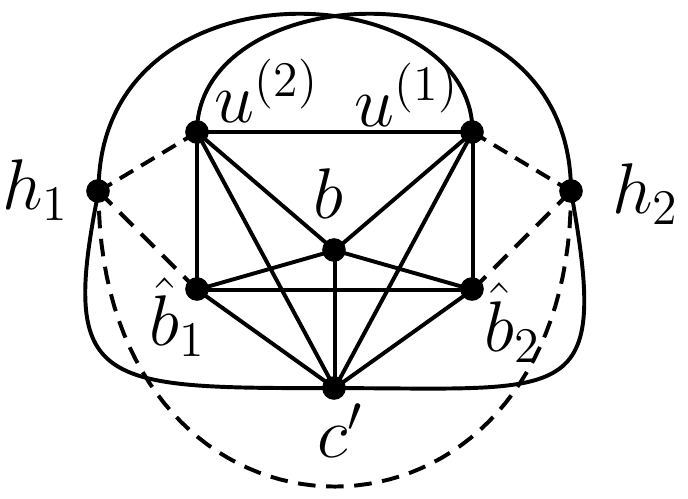}}
                                \end{equation}
                                where $b$ is a vertex in $B'\setminus\{\hat{b}_1,\hat{b}_2\}$, which exists since $\abs{B'}\geq 3$ by assumption.
                                The graph in \cref{eq:sub6} is not distance-hereditary, independently if the edges
                                \begin{equation}
                                    (h_1,h_2),\;(h_1,u^{(2)}),\;(h_2,\hat{b}_2),\;(h_1,u^{(1)}),\;(h_1,\hat{b}_2)
                                \end{equation}
                                are individually present or not.
                                To make this statement more transparent, we also provide the adjacency matrix of the graph in \cref{eq:sub6}.
                                The graph in \cref{eq:sub6} has adjacency matrix
                                \begin{equation}\label{eq:adjx1x5}
                                    \Gamma=\begin{blockarray}{ccccccccc}
                                        & c' & b & \hat{b}_1 & \hat{b}_2 & u^{(1)} & u^{(2)} & h_1 & h_2 \\
                                        \begin{block}{c(cccccccc)}
                                            c'        & 0   & 1   & 1   & 1   & 1   & 1   & 1   & 1   \\
                                            b         & 1   & 0   & 1   & 1   & 1   & 1   & 0   & 0   \\
                                            \hat{b}_1 & 1   & 1   & 0   & 1   & 0   & 1   & x_1 & 0   \\
                                            \hat{b}_2 & 1   & 1   & 1   & 0   & 1   & 0   & 0   & x_2 \\
                                            u^{(1)}   & 1   & 1   & 0   & 1   & 0   & 1   & 1   & x_3 \\
                                            u^{(2)}   & 1   & 1   & 1   & 0   & 1   & 0   & x_4 & 1   \\
                                            h_1       & 1   & 0   & x_1 & 0   & 1   & x_4 & 0   & x_5 \\
                                            h_2       & 1   & 0   & 0   & x_2 & x_3 & 1   & x_5 & 0   \\
                                        \end{block}
                                   \end{blockarray}
                                \end{equation}
                                where $x_1,\dots,x_5\in\{0,1\}$.
                                By explicit computation one can check that for any assignment of the variables $x_1,\dots,x_5$, the graph with adjacency matrix as in \cref{eq:adjx1x5} is not distance-hereditary.
                        \end{itemize}
                \end{itemize}
        \end{itemize}
        Since in all cases we arrived at a contradiction of the fact that $G$ is distance-hereditary, we know that \cref{eq:KSstep3} is true.
        \hfill$\diamond$
    \end{addmargin}

    \noindent\underline{Step 5:}
    \begin{addmargin}[1em]{0em}
        Here we show that there exist a $b\in B'$ such that $u_b\notin B'$, $B'\subseteq N_{u_b}$ and
        \begin{equation}\label{eq:KSstep5_2}
            (u_b,c')\notin E(G)\;\lor\;\exists h_b:\qty(h_b\in N_{u_b}\cap N_{c'}\setminus\bigcup_{x\in V'\setminus\{c'\}}N_x).
        \end{equation}
        Note that $B'\subseteq N_{u_b}$ implies $u_b\notin B'$, thus we can focus on the first property.
        We prove the statement by contradiction and assume therefore that there exist no such $b$, i.e. there exist no $b\in B'$ such that $B'\subseteq N_{u_b}$ and for which \cref{eq:KSstep5_2} is true.
        Let's first introduce the set $\mathcal{B}$ of vertices in $B'$ which satisfy the first of these properties, i.e.
        \begin{equation}
            \mathcal{B}=\{b\in B':\qty(B'\subseteq N_{u_b})\}.
        \end{equation}
        From the previous steps we know that $\mathcal{B}$ is not empty.
        Furthermore, from our assumption we must have that for all $b\in\mathcal{B}$, \cref{eq:KSstep5_2} is false, i.e
        \begin{equation}\label{eq:helper1}
            \forall b\in\mathcal{B}:\Bigg((u_b,c')\in E(G)\;\land\;\forall h_b:\qty(h_b\notin N_{u_b}\cap N_{c'}\setminus\bigcup_{x\in V'\setminus\{c'\}}N_x)\Bigg).
        \end{equation}
        Let $b$ now be a fixed element of $\mathcal{B}$.
        Since $u_b$ is adjacent to $c'$ and from the fact that $\mathcal{Q}(u_b,B'\setminus\{b\},L',c')$ is true, we know that there exist a $h_b$ such that
        \begin{equation}\label{eq:helper2}
            h_b\in N_{u_b}\cap N_{c'}\setminus\bigcup_{x\in V'\setminus\{c',b\}}N_x.
        \end{equation}
        Note that, \cref{eq:helper1} together with \cref{eq:helper2} implies that $h_b$ is adjacent to $b$ but to no other vertex in $B'$.
        We will now show that this leads to a contradiction by considering a vertex $b_1\in B'\setminus\{b\}$, such that $u_{b_1}\notin B'$, which we showed exists in step 1.
        Furthermore, let $b_2$ a vertex in $B'\setminus\{b,b_1\}$, which is necessarily adjacent to $u_{b_1}$, since $\mathcal{Q}(u_{b_1},B'\setminus\{b_1\},L',c')$ is true.
        Let's consider the following cases:
        \begin{itemize}
            \item Assume that $u_{b_1}$ is not adjacent to $c'$.
                By assumption we then have that $B'\nsubseteq N_{u_{b_1}}$, which implies that $u_{b_1}$ is not adjacent to $b_1$.
                Consider now the following induced subgraph
                \begin{equation}\label{eq:sub7}
                    G[\{c',b,b_1,b_2,u_{b_1},h_b\}]=\raisebox{-0.06\textwidth}{\includegraphics[scale=0.5]{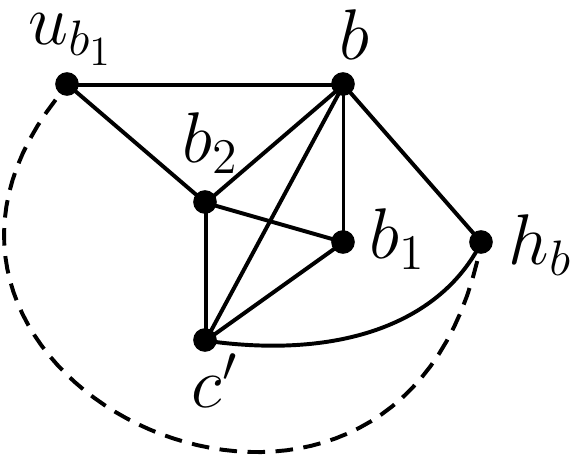}}
                \end{equation}
                which is not distance-hereditary, independently if the edge $(h_b,u_{b_1})$ is present or not.
            \item Assume that $u_{b_1}$ is adjacent to $c'$.
                \begin{itemize}
                    \item Assume that $B'\subseteq N_{u_{b_1}}$.
                        By the same argument as for $b$ and $u_b$, we know that there exist a vertex $h_{b_1}$ which is adjacent to $u_{b_1}$, $c'$ and $b_1$ but not to any other vertex in $B'$.
                        Consider now the following induced subgraph
                        \begin{equation}\label{eq:sub8}
                            G[\{c',b,b_1,b_2,h_b,h_{b_1}\}]=\raisebox{-0.06\textwidth}{\includegraphics[scale=0.5]{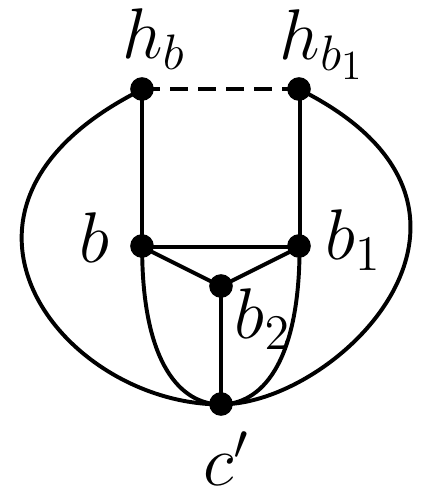}}
                        \end{equation}
                        which is not distance-hereditary, independently if the edge $(h_b,h_{b_1})$ is present or not.
                    \item Assume that $B\nsubseteq N_{u_{b_1}}$, which implies that $b_1$ is not adjacent to $u_{b_1}$ since $\mathcal{Q}(u_{b_1},B'\setminus\{b_1\},L',c')$ is true.
                        Furthermore, we know that there is a vertex $h_{b_1}$ which is adjacent to $u_{b_1}$ and $c'$ and possibly to $b_1$ but no other vertex in $B'$.
                        Consider now the following induced subgraph
                        \begin{equation}\label{eq:sub9}
                            G[\{c',b,b_1,b_2,u_b,u_{b_1},h_b,h_{b_1}\}]=\raisebox{-0.06\textwidth}{\includegraphics[scale=0.5]{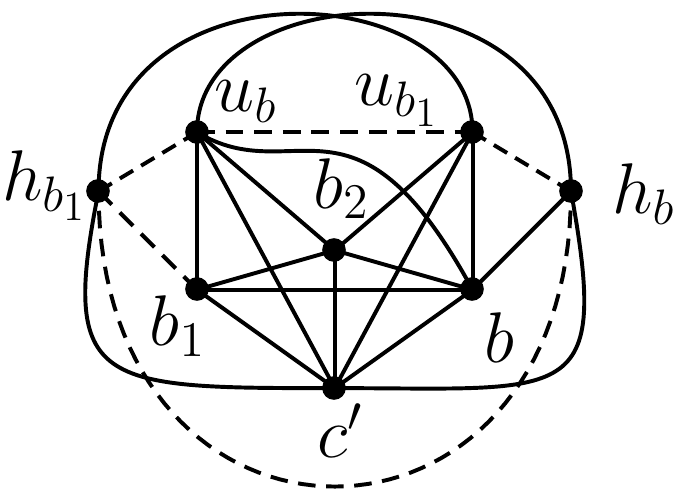}}
                        \end{equation}
                        which is not distance-hereditary, independently if the edges
                        \begin{equation}
                            (h_b,h_{b_1}),\;(h_b,u_{b_1}),\;(h_{b_1},b_1),\;(h_{b_1},u_b),\;(u_b,u_{b_1})
                        \end{equation}
                        are individually present or not.
                        To make this statement more transparent, we also provide the adjacency matrix of the graph in \cref{eq:sub9}.
                        The graph in \cref{eq:sub9} has adjacency matrix
                        \begin{equation}\label{eq:adjx1x5_2}
                            \Gamma=\begin{blockarray}{ccccccccc}
                                & c' & b & b_1 & b_2 & u_b & u_{b_1} & h_b & h_{b_1} \\
                                \begin{block}{c(cccccccc)}
                                    c'      & 0   & 1   & 1   & 1   & 1   & 1   & 1   & 1   \\
                                    b       & 1   & 0   & 1   & 1   & 1   & 1   & 1   & 0   \\
                                    b_1     & 1   & 1   & 0   & 1   & 1   & 0   & 0   & x_1 \\
                                    b_2     & 1   & 1   & 1   & 0   & 1   & 1   & 0   & 0   \\
                                    u_b     & 1   & 1   & 1   & 1   & 0   & x_2 & 1   & x_3 \\
                                    u_{b_1} & 1   & 1   & 0   & 1   & x_2 & 0   & x_4 & 1   \\
                                    h_b     & 1   & 1   & 0   & 0   & 1   & x_4 & 0   & x_5 \\
                                    h_{b_1} & 1   & 0   & x_1 & 0   & x_3 & 1   & x_5 & 0   \\
                                \end{block}
                           \end{blockarray}
                        \end{equation}
                        where $x_1,\dots,x_5\in\mathbb{F}_2$.
                        By explicit computation one can check that for any assignment of the variables $x_1,\dots,x_5$, the graph with adjacency matrix as in \cref{eq:adjx1x5_2} is not distance-hereditary.
                \end{itemize}
        \end{itemize}
        Since in all cases we arrived at a contradiction of the fact that $G$ is distance-hereditary, we know that there exist a $b\in B'$ such that $u_b\notin B'$, $B'\subseteq N_{u_b}$ and such that \cref{eq:KSstep5_2} is true.
        \hfill$\diamond$
    \end{addmargin}
    We have therefore shown that \cref{eq:contrapos2} is true which implies that \cref{eq:Pimpl2} is true.
    Finally, as we described in the beginning of the proof, this implies, by induction, that if $\mathcal{P}(B',L',c')$ is false then $S_{V'}\nless G$.
\end{proof}

\subsection{Fixed-parameter tractable algorithm for unbounded rank-width}\label{sec:FPT}
In this section we will show that the star vertex-minor problem (\SVM) is fixed-parameter tractable for circle graphs, in terms of the size of the considered star graph.
More specifically, we will show that there exists an efficient algorithm to decide if $S_{V'}$ is a vertex-minor of $G$, given that $G$ is a circle graph and the subset $V'\subseteq V(G)$ is of size $k$.
We will do this by showing that we can map this problem in polynomial time to deciding whether the $4$-regular multi-graph that defines $G$ has a~\SOET~on the vertices $V'$.
This is done in \cref{sssec:ksvm_to_ksoet}.
We will then give an algorithm that decides whether a 4-regular multi-graph has a \SOET~on a subset $V'$ of its vertices where $V'=k$ is fixed.
This is done in \cref{sssec:ksoet_in_p}.
We begin by formally stating the decisions problems considered.\\

 We first define the problem \kSVM.

\begin{prm}[\kSVM]\label{prob:ksvm}
Let $G$ be a graph and let $V'$ be a subset of $V(G)$ with $|V'|=k$. Decide whether $S_{V'}$ is a vertex-minor of $G$.
\end{prm}
We also define the problem \kSOET.

\begin{prm}[\kSOET]\label{prob:ksoet}
Let $F$ be a 4-regular multi-graph and let $V'$ be a subset of $V(F)$ with $|V'|=k$. Decide whether $F$ allows for a \SOET~with respect to $V'$
\end{prm}

\begin{thm}\label{thm:ksvm_in_p}
\kSVM, restricted to circle graphs, is in $\mathbb{P}$.
\end{thm}
\begin{proof}
This will follow from \cref{thm:ksvm_to_ksoet} which provides an efficient mapping of every circle graph instance of $\kSVM$ to a corresponding instance of $\kSOET$. By \cref{cor:ksoet_in_p} $\kSOET$ is in $\Poly$ and hence so is $\kSVM$. An efficient algorithm for $\kSOET$ is given in \cref{alg:kSOET}.
\end{proof}
This theorem has the following corollary

\begin{cor}\label{cor:svm_is_fpt}
\SVM~is fixed-parameter tractable in the size of the input vertex-set $V'$ if the input graph $G$ is a circle graph.
\end{cor}
The existence of this fixed-parameter tractable algorithm is theoretically interesting but it is not likely to be of practical use. This is so because while the algorithm is efficient in the size of the input graph $G$ it suffers from a hidden constant that is of size $O(k!\cdot (f(k)^{\mathcal{O}(k f(k))}))$, where $f(k)=2^{2^{\mathcal{O}(k^2)}}$ making its practical implementation unlikely.

\subsubsection{Mapping \kSVM~to~\kSOET}\label{sssec:ksvm_to_ksoet}

In this section we will prove that there exists an efficient mapping from instances of $\kSVM$ that are circle graphs to $\kSOET$. This is formalized in the following theorem, the proof of which also provides a prescription of the algorithm that defines the mapping. 

\begin{thm}\label{thm:ksvm_to_ksoet}
Let $(G,V')$ be an instance of $\kSVM$ and let $G$ be a circle graph. These is an efficient mapping from this instance to an instance of $\kSOET$ and moreover the instance $(G,V')$ is a \emph{yes}-instance of $\kSVM$ if and only if its image under the mapping is a \emph{yes}-instance of $\kSOET$.
\end{thm}
\begin{proof}
We will prove this by providing an explicit mapping.
An instance $(G,V')$ of \kSVM, where $G$ is a circle graph and $V'$ a vertex set, can be mapped to an instance of $\kSOET$ by the following two steps:
\begin{itemize}
    \item Find a double occurrence word $\bs{X}$ with letters in $V(G)$ such that $G=\mathcal{A}(\bs{X})$.
        This can be done in time $\mathcal{O}(\abs{V(G)}^2)$ by using Spinrad's algorithm~\cite{Spinrad1994}.
    \item Construct a 4-regular multi-graph $F$, such that $\bs{X}=m(U)$ for some Eulerian tour $U$ on $F$.
        As shown in~\cite{Bouchet1994}, this can be done in the following way:
        \begin{itemize}
            \item Let $C_{\bs{X}}$ be a cycle graph with the vertices labeled as the consecutive letters of $\bs{X}$.
            \item Contract every pair of vertices with the same label, while keeping all the edges. Note that this step can create multi-edges or self-loops.
                This step can be done in time $\mathcal{O}(\abs{V(G)})$ by adding the corresponding rows of the adjacency matrix and deleting one row and one column.
        \end{itemize}
        The graph obtained from these steps is then a 4-regular multi-graph $F$ with a Eulerian tour $U$, such that $\bs{X}=m(U)$.
    Therefore, constructing the 4-regular multi-graph $F$, given $\bs{X}$, can also be done in time $\mathcal{O}(\abs{V(G)}^2)$.
\end{itemize}
One can see that the above algorithm runs in $\mathcal{O}(\abs{V(G)}^2)$ and given a circle graph $G$ outputs a $4$-regular graph $F$ that has a Eulerian tour $U$ such that $\mathcal{A}(U)=G$.
From \cref{cor:reg_to_circle} we then know that $S_{V'}<G$ if and only if $F$ allows for a \SOET\ with respect to $V'$.
Using the above two steps we see that any circle graph instance of \kSVM\  can be mapped to \kSOET\ in time $\mathcal{O}(\abs{V(G)}^2)$.
\end{proof}

Given this mapping, the next logical step is to find an efficient algorithm for $\kSOET$. This is done in the next section.

\subsubsection{\kSOET~is in \Poly}\label{sssec:ksoet_in_p}

In this section we will prove that $\kSOET$ is in \Poly. We will do this by explicitly writing down an efficient algorithm. This algorithm will make use of an algorithm which solves a well known graph problem we call $\DPP{k'}$. $\DPP{k'}$, for $k'$-Disjoint Path Problem is formally defined as follows.

\begin{prm}[\DPP{k'}]\label{prob:kdpp}
Let $H$ be a multi-graph and let $L = \{(v_1,v'_1),\ldots,(v_{k'},v'_{k'})\} $ be a set of two-tuples of vertices of $H$ such that $|L|=k'$. Decide whether there exist $k'$ edge-disjoint paths $P_i$ on $H$ that start at $v_i$ and end at $v'_i$ for $i \in [k']$.
\end{prm}
Robertson and Seymour proved that there exist an efficient algorithm for solving \DPP{k'}, if $k'$ is fixed.
Their proof of correctness is based on 23 papers named Graph minors. I , $\dots$, Graph minors. XXIII \cite{RobertsonFullSeries}.
In \cite{Kawarabayashi2015} an improved algorithm for \DPP{k'} is given, with a much smaller hidden constant describing the scaling in terms of $k'$.
This hidden constant is still quite large as running time for the algorithm in \cite{Kawarabayashi2015} is $(f(k')^{\mathcal{O}(k' f(k'))})n^{\mathcal{O}(1)}$, where $n$ is the number of vertices in the graph and $f(k')=2^{2^{\mathcal{O}(k'^2)}}$\\

We will discuss an algorithm that solves $\kSOET$ efficiently. It will use an algorithm for $\DPP{k'}$ as a subroutine, calling it a constant number of times.\\

The algorithm for \kSOET~is sketched as follows.
Let $F$ be a $4$-regular multi-graph and let $V'\subset V(F)$ be of size $k$.
Recall that a \SOET~is a Eulerian path that visits the vertices in $V'$ in some order.
A necessary condition for a \SOET~to exist is that there are, for some ordering of the vertices in $V'$ two edge-disjoint trails from the first vertex in $V'$ to the second vertex in $V'$ and from the second to the third and so on.
If one is given such a collection of paths it is not hard to see that one can connect these paths to each other to form a tour and moreover extend this total tour to be Eulerian.
Hence we could use the \DPP{k'} algorithm described above to to find such a tour by finding all the edge disjoint trails that connect the vertices in $V'$.\\

There are two problems with this.
The first problem is that a \SOET~with respect to $V'$ can exist with respect to any possible ordering of the set $V'$.
This means that in order to find this tour we might have to apply \DPP{k} to all $k!$ different orderings.
This is a large overhead but acceptable since we are only looking for an algorithm that is efficient for fixed $k$.
The second problem is that the \DPP{k'} algorithm expects $k'$ pairs of vertices and requires these pairs to be different.\\

This means we cannot input the same vertex-pair twice to get two edge disjoint paths.
We can resolve this at some overhead by running the \DPP{k'} algorithm on a modified multi-graph $H$.
This multi-graph, which we call a \SOET~splitting of $F$, is created by taking each vertex $v_i$ in $V$ and splitting it into two vertices $v_i^{(a)}, v_i^{(b)}$ such that the four edges $e_1,\ldots e_4$ originally incident on $v_i$ are now pairwise incident on $v_i^{(a)}, v_i^{(b)}$.\\

As an example we could have for instance that $e_1,e_2$ are incident on $v_i^{(a)}$ and $e_3,e_4$ are incident on $v_i^{(b)}$.
We must consider all possible choices of pairings here (of which there are $6$) and since there are $k$ vertices on which to perform this procedure (each vertex in $V'$) there are $6^k$ such \SOET~splittings of $F$ (which are not $4$-regular anymore).
This means we must call the \DPP{2k} subroutine $6^k\cdot k!$ times to account for all possible orderings of the \SOET.\\

Let's make this a little more rigorous.
We begin by defining the notion of a \SOET-splitting of a $4$-regular multi-graph $F$ with respect to a set $V'$.

\begin{mydef}[\SOET-splitting]
Let $F$ be a $4$-regular multi-graph and let $V'$ be a subset of it's vertices. A \SOET-splitting $H$ of $F$ with respect to $V$ is a multi-graph created from $F$ y performing the following operation on all vertices in $V$
\begin{equation}\label{eq:soet_split}
    \raisebox{-0.035\textwidth}{\includegraphics[scale=0.5]{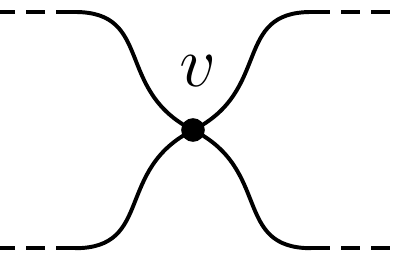}}\;\mapsto\;\raisebox{-0.035\textwidth}{\includegraphics[scale=0.5]{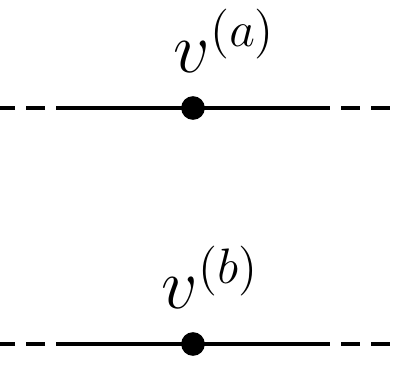}}
\end{equation}
We label the two vertices that originate from a vertex $v\in V'$ as $v^{(a)}$ and $v^{(b)}$. Note that we have not specified how to connect the edges that were originally incident on $v$ to $v^{(a)}$ and $v^{(b)}$. There are six possible ways to do this for each $v\in V'$ and each choice leads to an a priori distinct \SOET-splitting multi-graph $H$.
\end{mydef}
We also define the subroutine $\DPP{k'}$
\begin{algorithm}[H]
    \caption{\DPP{k'}}\label{alg:DPP}
    \begin{algorithmic}[1]
        \State \textbf{INPUT}:\hspace{1.2em} A multi-graph $H$ and a set of $2$-tuples $\{(v_1, \hat{v}_1), \ldots (v_{k'}, \hat{v}_{k'}) \}$
        \State \textbf{OUTPUT}: TRUE if there exist edge-disjoint paths in $H$ connecting $v_i,\hat{v}_i$ for all $i\in [k']$.
        \State \phantom{\textbf{OUTPUT}:} FALSE otherwise
    \end{algorithmic}
\end{algorithm}
Note that we have only specified the inputs and outputs of this subroutine. For details on the inner workings of this algorithm see~\cite{Kawarabayashi2015}.
Using this subroutine we can write down an algorithm for $\kSOET$.

\begin{algorithm}[H]
    \caption{$\kSOET$}\label{alg:kSOET}
    \begin{algorithmic}[1]
        \State \textbf{INPUT}:\hspace{1.2em} A $4$-regular multi-graph $F$ and a set $V' =\{v_1, \ldots,v_k\}$ such that $V'\subset V(F)$
        \State \textbf{OUTPUT}: TRUE if there exist a \SOET\ on $F$ with respect to $V'$
        \State \phantom{\textbf{OUTPUT}:} FALSE otherwise
        \State \hrulefill
        \State
        \State Generate a list $L$ of all possible \SOET-splittings of $F$ with respect to $V'$
        \State \For{all graphs $H$ in $L$}
            \State \For{all permutations $\pi$ of the set $[1:k]$}
                \State Run \DPP{2k} on the set $\{(v^{(a)}_{\pi(1)},v^{(a)}_{\pi(2)}), \ldots ,(v^{(a)}_{\pi(k)},v^{(a)}_{\pi(1)}), (v^{(b)}_{\pi(1)},v^{(b)}_{\pi(2)}), \ldots ,(v^{(b)}_{\pi(k)},v^{(b)}_{\pi(1)})\}$
                \State \If{\DPP{2k} returns TRUE}
                    \State \textbf{RETURN} TRUE
                    \State \textbf{QUIT}
                \EndIf
            \EndFor
        \EndFor        
        \State \textbf{RETURN} FALSE
        \State \textbf{QUIT}            

    \end{algorithmic}
\end{algorithm}
We now prove that this algorithm returns TRUE if and only if the multi-graph $F$ allows for a \SOET~with respect to the vertex-subset $V'$. We begin by proving that if the algorithm returns TRUE the multi-graph $F$ allows for a \SOET~with respect to the vertex-subset $V'$. We have the following theorem.

\begin{thm}\label{thm:kdpp_to_soet}
    Let $F$ be a connected $4$-regular multi-graph and let $V'\subset V(F)$ be a subset of its vertices with $|V'|=k$. If \cref{alg:kSOET} returns TRUE then $F$ allows for a \SOET~with respect to the vertex set $V'$.
\end{thm}
\begin{proof}
    Let $F$ be a connected $4$-regular multi-graph and let $V'\subset V(F)$ be a subset of its vertices with $|V'|=k$.
    If \cref{alg:kSOET} returns true this means there exists a \SOET-splitting $H$ of $F$ and a permutation $\pi$ of the set $[k]$ such that there exist $2k$ edge-disjoint paths in $H$ that connect the vertices $v^{(a)}_{\pi(1)}$ and $v^{(a)}_{\pi(2)}$, $v^{(a)}_{\pi(2)}$ and $v^{(a)}_{\pi(3)}$ and so forth.
    Undoing the \SOET-splitting operation that defines $F$ we can see that these edge-disjoint paths can be attached to one another to form a closed trail\footnote{Note that a path is also a trail.} (a tour) $U$ on $F$ that visits all vertices of $V'$ twice in the order $v_{\pi(1)}, \ldots ,v_{\pi(k)}$.
    This is not yet a Eulerian tour however.
    To construct a Eulerian tour out of the tour $U$ we consider the multi-graph $\hat{H}$ obtained from $F$ by deleting all vertices in $V$ (looking at the induced subgraph $F[V\setminus V']$) and subsequently removing all edges in the tour $U$ from the remaining multi-graph.
    The resulting multi-graph will consist of multiple connected components.
    These connected components will either consist of a single vertex or will be multi-graphs with vertices of degree two and four.
    This means all these connected components are Eulerian.
    Moreover each of these connected components will contain a vertex that is also a vertex in the tour $U$.
    Consider on each of these connected components then a Eulerian tour.
    These tours will thus form tours on the original multi-graph $F$ as well and since all such tours have at least one vertex in common with the tour $U$ we can extend $U$ to a Eulerian tour on the multi-graph $F$ by for each connected component cutting $U$ at such a vertex and inserting the tour originating from the connected components of $G$.
    This means that $U$ can be turned into a Eulerian tour and hence there exists a \SOET~on the multi-graph $F$.
    This proves the theorem.
\end{proof}

Next we prove the converse statement.

\begin{thm}\label{thm:soet_to_kdpp}
Let $F$ be a connected $4$-regular multi-graph and let $V'\subset V(F)$ be a subset of its vertices with $|V'|=k$. If $F$ allows for a \SOET~with respect to $V'$ then \cref{alg:kSOET} will return TRUE.
\end{thm}
\begin{proof}
Let $U$ be a $\SOET$~on $F$ with respect to $V'$ and without loss of generality assume that $U$ traverses the vertices of $V'$ in the order $v_1,v_2, \ldots,v_k$.
Hence for the vertices $v_1,v_2$ there are $2$ sub-trails $U_1^{(a)},U^{(b)}_1$ of $U$ that start at $v_1$ and end at $v_2$.
The sub-trail $U_1^{(a)}$ (or $U_1^{(b)}$) might not be a path, but one can easily pick a subset of the edges in $U_1^{(a)}$ which gives a path, for example as the shortest path $P_1^{(a)}$ between $v_1$ and $v_2$ in the subgraph of $F$ given by the vertices and edges of $U_1^{(a)}$.
Similarly for $U_1^{(b)}$ and $P_1^{(b)}$.
We can make the same argument for the vertices $v_2,v_3$ and so forth.
By the definition of \SOET~splittings there must thus exist a \SOET~splitting $G$ of $F$ with respect to $V'$ such that the path $P_1^{(a)}$ starts at $v^{(a)}_1$ and ends at $v^{(a)}_2$ and such that the path $P_1^{(b)}$ starts at $v^{(b)}_1$ and ends at $v^{(b)}_2$.
We can make the same argument for the vertices $v_2,v_3$ and so forth.
Hence there must exist a \SOET-splitting $G$ of $F$ with respect to $V'$ such that the algorithm \DPP{2k}$(G,S)$ with
$S = \{(v^{(a)}_{1},v^{(a)}_{2}), \ldots ,(v^{(a)}_{k},v^{(a)}_{1}), (v^{(b)}_{1},v^{(b)}_{2}), \ldots ,(v^{(b)}_{k},v^{(b)}_{1})\}$ returns TRUE.
Since \cref{alg:kSOET} runs over all possible orderings of $V'$ and all possible \SOET~splittings this particular call to \DPP{2k}~will always happen and hence \cref{alg:kSOET} will return TRUE.
This proves the theorem.
\end{proof}
This leads to the following corollary.

\begin{cor}\label{cor:ksoet_in_p}
$\kSOET$~is in $\mathbb{P}$.
\end{cor}
\begin{proof}
By $\cref{thm:soet_to_kdpp}$ and $\cref{thm:kdpp_to_soet}$, $\cref{alg:kSOET}$ returns TRUE if and only if the tuple $(F,V')$, with $F$ a $4$-regular multi-graph and $V'\subset V(F)$ a subset of its vertices with $|V'|=k$, is a YES instance of \kSOET.
 Moreover the function $\DPP{k'}$ runs in polynomial time in the size of the input graph and we have for all \SOET-splittings $G$ of $F$ with respect to $V'$ that $|V(G)|=|V(F)|+k$ and $|E(G)| = |E(F)|$. \Cref{alg:kSOET} hence performs a constant number of function calls (constant in the size of $F$, not in $k$) of $\DPP{k'}$ with an input multi-graph of size $O(|V(F)|,|E(F)|)$. Hence \Cref{alg:kSOET} runs in polynomial time with respect to $|V(F)|,|E(F)|$ and thus we have that \kSOET~is in $\mathbb{P}$.
\end{proof}
This corollary then leads to \cref{thm:ksvm_in_p}.

\section{Connected vertex-minor on three vertices or less}\label{sec:small_star}
In this section we show that a fixed connected graph with vertices $V'$ is a vertex-minor of a connected graph $G$ if $V'\subseteq V(G)$ and $\abs{V'}\leq 3$.
Furthermore we provide an algorithm that reduces the number of operations one need to transform $G$ to the desired graph.

\begin{thm}\label{thm:small_star}
    Let $G$ be a connected graph and $G'$ be a connected graph with vertices $V'$, such that $\abs{V'}\leq 3$.
    Then we have that
    \begin{equation}
        G'<G\quad\Leftrightarrow\quad V'\subseteq V(G)
    \end{equation}
\end{thm}
\begin{proof}
    Note that if $G'<G$ then clearly $V'\subseteq V(G)$.
    Assume therefore that $V'\subseteq V(G)$.
    Let's denote the vertices in $V(G)\setminus V'$ as $U=(v_1,v_2,\dots,v_{n-k}$, where $n=\abs{V(G)}$ and $k=\abs{G'}$.
    We will now show that $G'<G$ by finding a sequence of operations $P=P_{v_{n-k}}\circ\dots\circ P_{v_1}$, where $P_{v_i}\in\{X_v,Y_v,Z_v\}$, such that each intermediate graph $G^{(i)}=P_{v_i}\circ\dots\circ P_{v_1}(G)$ is connected.
    Since any connected graph with vertices $V'$, for $\abs{V'}\leq 3$, is either a star graph or a complete graph, which are LC-equivalent, this shows that $G'<G$.
    Let's therefore consider such an intermediate graph $G^{(i)}$ for some $i\in[n-k]$ and the next vertex $v_{i+1}$.
    We will now show that $G^{(i)}\setminus v_{i+1}$ or $\tau_{v_{i+1}}(G^{(i)})\setminus v_{i+1}$ is connected.
    This will be done by considering the following two cases:

    \noindent\underline{Assume that $G^{(i)}\setminus v_{i+1}$ is not connected:}
    \begin{addmargin}[1em]{0em}
        This means that $v_{i+1}$ is a cut vertex.
        Let $G_1$ be a connected component in the graph $G^{(i)}\setminus v_{i+1}$ and $G_2$ the rest of the vertices, see \cref{eq:cut_vertex1,eq:cut_vertex2} for an illustration.
        Since $v_{i+1}$ is a cut vertex, we know that no vertex in $N_{v_{i+1}}\cap V(G_2)$ is adjacent to any vertex in $N_{v_{i+1}}\cap V(G_1)$, in the graph $G^{(i)}$.
        Thus, all vertices in $N_{v_{i+1}}\cap V(G_2)$ are adjacent to all vertices in $N_{v_{i+1}}\cap V(G_1)$ in the graph $\tau_{v_{i+1}}(\hat{G})\setminus {v_{i+1}}$, showing that this graph is connected, see \cref{eq:cut_vertex3}.
        \begin{align}
            G^{(i)}\;&=\;\raisebox{-0.04\textwidth}{\includegraphics[scale=0.5]{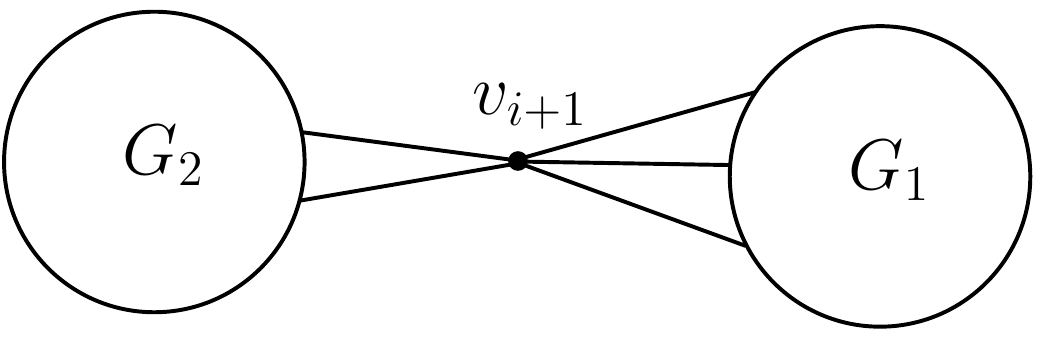}}\label{eq:cut_vertex1}\\
            G^{(i)}\setminus v_{i+1}\;&=\;\raisebox{-0.04\textwidth}{\includegraphics[scale=0.5]{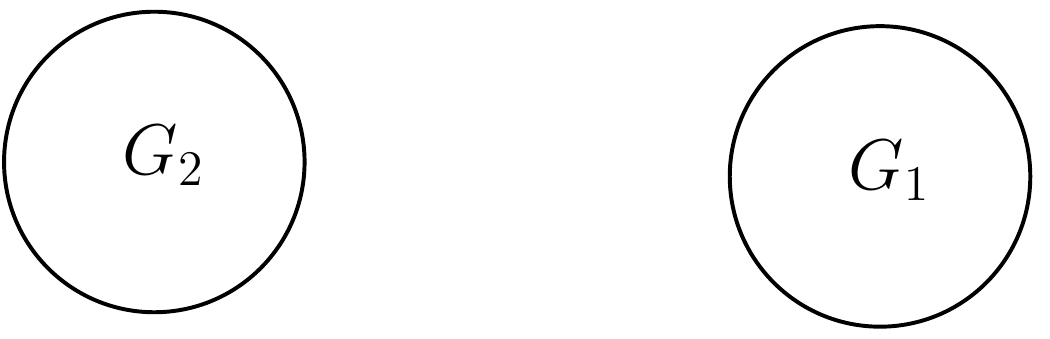}}\label{eq:cut_vertex2}\\
            \tau_{v_{i+1}}(G^{(i)})\setminus v_{i+1}\;&=\;\raisebox{-0.04\textwidth}{\includegraphics[scale=0.5]{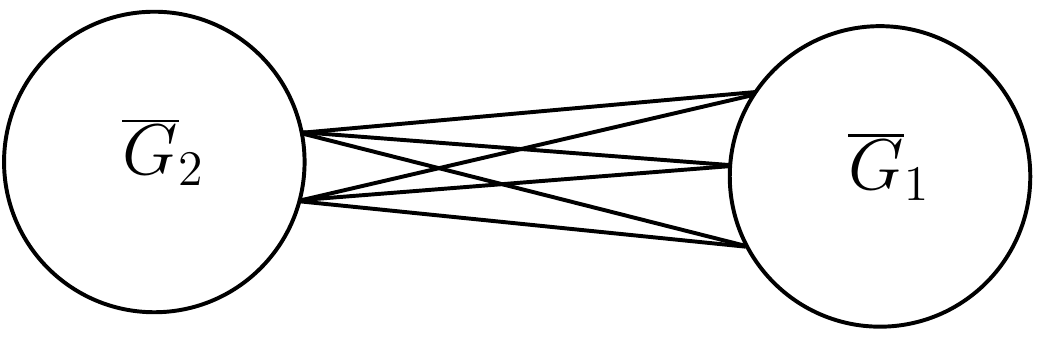}}\label{eq:cut_vertex3}
        \end{align}
    \end{addmargin}

    \noindent\underline{Assume that $\tau_v(\hat{G})\setminus v$ is not connected:}
    \begin{addmargin}[1em]{0em}
        This case follows from the previous.
        To see this, let $\tilde{G}$ be the graph $\tau_v(G^{(i)})$.
        From the above case we then know that $\tau_v(\tilde{G})\setminus v_{i+1}$ is connected.
        But this graph is exactly $G^{(i)}\setminus v_{i+1}$ since local complementations are involutions and the theorem follows.
    \end{addmargin}
\end{proof}

From \cref{thm:small_star} we can easily formulate an algorithm that finds a sequence of operations taking $G$ to a desired connected graph $G'$ on $V'$, by simply checking recursively if $\hat{G}\setminus v$ or $\tau_v(\hat{G})\setminus v$ is connected.
However, if the vertices in $V'$ are already close in $G$ and $G$ is a very large, it would be practical to not have to consider all the vertices in $G$.
Next, we present a more practical algorithm to find a sequence of local complementations and vertex-deletions that take some graph $G$ to a star graph\footnote{Note that any other connected graph on $\{a,b,c\}$ can easily be constructed.} on the vertices $a$, $b$ and $c$.
The algorithm performs the following steps:
\begin{itemize}
    \item Find the shortest path $P_1$ between $a$ and $b$ and connect these by doing $\tau$-operations along this path.
        This first step already give an algorithm for creating a star graph on vertices $a$ and $b$, see \cref{alg:S1}.
    \item Then do the same with $b$ and $c$ by finding the shortest path $P_2$ between $b$ and $c$.
        The question is now whether we removed the edge between $a$ and $b$ while connecting $b$ and $c$.
        This could only happen if $P_2$ goes through a vertex which is a common neighbor of $a$ and $b$.
        Furthermore, since $P_2$ is the shortest path from $c$ to $b$ this could only be the case for the last vertex on the path, before $b$.
        \begin{itemize}
            \item If $c$ is connected to $a$ before the last local complementation on $P_2$, then stop, since the induced graph is already connected.
            \item Assume that $c$ is not connected to $a$ before the last local complementation is performed.
                So in this case, after performing the local complementation along $P_2$, $a$ and $b$ will not be connected but they will both be connected to $c$.
        \end{itemize}
\end{itemize}

The full protocol is formalized in \cref{alg:S2}.

\begin{algorithm}[H]
    \caption{Producing $S_1$ on vertices $a$ and $b$}\label{alg:S1}
    \begin{algorithmic}[1]
        \State \textbf{INPUT}:\hspace{1.2em} A graph $G$ and two vertices $a,b\in V(G)$.
        \State \textbf{OUTPUT}: A sequence of vertices $\bs{v}$ such that $\tau_{\bs{v}}(G)[\{a,b\}]$ is a star graph.
        \State \hrulefill
        \State
        \State Find the shortest path $P$ between $a$ and $b$
        \State Perform $\tau_p$ for all $p\in P\setminus\{a,b\}$
    \end{algorithmic}
\end{algorithm}

\begin{algorithm}[H]
    \caption{Producing $S_2$ on vertices $a$, $b$ and $c$}\label{alg:S2}
    \begin{algorithmic}[1]
        \State \textbf{INPUT}:\hspace{1.2em} A graph $G$ and three vertices $a,b,c\in V(G)$.
        \State \textbf{OUTPUT}: A sequence of vertices $\bs{v}$ such that $\tau_{\bs{v}}(G)[\{a,b,c\}]$ is a star graph.
        \State \hrulefill
        \State
        \State Find the shortest path $P_1$ between $a$ and $b$
        \State Perform $\tau_p$ for all $p\in P_1\setminus\{a,b\}$
        \State Find the shortest path $P_2=(p_0=b,p_1,\dots,p_n,p_{n+1}=c)$
        \For{$i$ \textbf{in} $1,\dots,n$}
            \If{$a$ and $c$ are not adjacent}
                \State Perform $\tau_{p_i}$
            \EndIf
        \EndFor
    \end{algorithmic}
\end{algorithm}

\section{Conclusion}\label{sec:conclusion}
We have shown that deciding if a graph state $\ket{G'}$ can be obtained from another graph state $\ket{G}$ using \LCLPMCC\ is \NP-Complete, by showing that \VM\ is \NP-Complete.
The computational complexity of \VM\ was previously unknown and was posted as an open problem in~\cite{Pivot-minors2016}.
It is important to note that our results are for labeled graphs, since vertices correspond to physical qubits in the case of transforming graph states.

We presented an efficient algorithm for $\SVM$ when the input graph is restricted to be distance-hereditary.
It would be interesting to know if the same approach generalize to qudit graph states described by weighted graphs of low rank-width.

We presented an efficient algorithm for $\kSVM$ on circle graphs, that is the problem of deciding if the star graph on a subset of vertices $V'$, with $\abs{V'}=k$, is a vertex-minor of another graph.
However, the computational complexity of \kSVM\ or \kVM\ on general graphs is still unknown.

Below we list some more open questions regarding the computational complexity of deciding how graph states can be transformed using local operations:

\begin{itemize}
    \item Is the problem still \NP-Complete if one allows for arbitrary local operations and classical communication (LOCC), instead of restricting to only \LCLPMCC?
        Is there an efficient algorithm for states with bounded Schmidt-rank width also in this case?
        It has been shown that many graph state are equivalent under LC if and only if they are equivalent under stochastic LOCC (SLOCC)~\cite{Zeng2007}.
        An interesting question is therefore whether graph states described by distance-hereditary or circle graphs fall into this category.
    \item What is the computational complexity of deciding if a graph state $\ket{G}$ can be transformed in to another $\ket{G'}$ if the qubits are partitioned into multiple local sets, within which multi-qubit Clifford operations and multi-qubit Pauli measurements can be performed.
    \item How do the computational complexity results of this paper relate to the complexity of \emph{code switching} in stabilizer quantum error correction codes, e.g. fault-tolerant interconversion between two stabilizer codes~\cite{gottesman1997stabilizer}?
        More precisely, can one re-use the techniques presented here, to prove that it is computationally hard to decide whether one stabilizer code can be fault-tolerantly transformed into another.
        We conjecture that this could be possible, since many stabilizer codes (such as topological stabilizer codes) have a notion of `locality' with fault-tolerant conversions being the conversions that respect this locality~\cite{nautrup2017fault}.
\end{itemize}

While finishing this work we became aware of work by F. Hahn et al~\cite{Hahn2018}, where specific instances of \VM\ are considered in the context of quantum network routing.

\subsection*{Acknowledgements}
We thank Kenneth Goodenough, Think Le, Bas Dirkse, Victoria Lipinska, Stefan B\"{a}uml, Tim Coopmans, J\'{e}r\'{e}my Ribeiro, Ben Criger and Jens Eisert for discussions.
AD, JH and SW were supported by STW Netherlands, NWO VIDI grant and an ERC Starting grant.

\bibliographystyle{abbrv}
\bibliography{references,refs_misc}

\end{document}